\documentclass[11pt]{article}
\usepackage[a4paper,margin=2.6cm]{geometry}
\usepackage{amsmath,amssymb,amsthm}
\usepackage{float}
\usepackage{algorithm,algpseudocode}
\newcommand{\algheader}[2]{\bigskip\par\noindent\textbf{#1}\ \ #2\par\nobreak\vspace{2pt}\hrule height 0.4pt\nobreak\vspace{5pt}\par\nobreak}
\newcommand{\algfooter}{\par\vspace{3pt}\hrule height 0.4pt\bigskip}
\usepackage{url}
\usepackage{hyperref}

\newtheorem{theorem}{Theorem}[section]
\newtheorem{lemma}[theorem]{Lemma}
\newtheorem{proposition}[theorem]{Proposition}
\newtheorem{corollary}[theorem]{Corollary}
\theoremstyle{definition}
\newtheorem{definition}[theorem]{Definition}
\newtheorem{example}[theorem]{Example}
\theoremstyle{remark}
\newtheorem{remark}[theorem]{Remark}

\newcommand{\C}{\mathbb C}
\newcommand{\Q}{\mathbb Q}
\newcommand{\Z}{\mathbb Z}

\newcommand{\Cl}{\operatorname{Cl}}
\newcommand{\cO}{\mathcal O}
\newcommand{\cA}{\mathcal A}
\newcommand{\cS}{\mathcal S}
\newcommand{\cV}{\mathcal V}
\newcommand{\cP}{\mathcal P}
\newcommand{\cU}{\mathcal U}
\newcommand{\cC}{\mathcal C}
\newcommand{\Fbar}{\overline{F}}
\newcommand{\den}{\operatorname{den}}
\newcommand{\Const}{\operatorname{Const}}

\newcommand{\supp}{\operatorname{supp}}
\newcommand{\Div}{\operatorname{Div}}
\newcommand{\Jac}{\operatorname{Jac}}
\newcommand{\F}{\mathbb F}
\newcommand{\dv}{\operatorname{div}}
\newcommand{\Dbar}{\overline{D}}
\newcommand{\fp}{\mathfrak p}
\newcommand{\fd}{\mathfrak d}

\title{Parallel Integration over Simple Radical Extensions II:\\ Mixed Towers}
\author{Sam Blake}
\date{\today}

\begin{document}
\maketitle

\begin{abstract}
In Part I we extended the structure theorems underlying the Risch--Norman (parallel Risch) method to a simple radical extension $L=K(y)$, $y^m=q$, of a differential field $K=F(t_1,\dots,t_n)$ closed under the derivation. Here we remove the closure hypothesis: the radical may occupy any position in the tower, so that the derivatives of the generators above it involve $y$ --- the setting of Bronstein's algorithm for mixed elementary functions. The working ring is $\cA=\cO[t_{j+1},\dots,t_n]$, the integral closure of $F[t_1,\dots,t_n]$ in $L$: a Krull domain, free over the polynomial ring on Trager's basis, so that all factorisation remains in a unique factorisation domain. The denominator of the derivation is no longer an element but a divisor $\fd_D$ on $\cA$, and the valuation lemma takes the unified form $v_P(Dg)=v_P(g)-(1+v_P(\fd_D))$ at normal height-one primes, subsuming the shifts $\{1,e_P\}$ of Part I; the proof localises and requires no cancellation analysis. Stability of the class group and the unit group, $\Cl(\cA)\cong\Cl(\cO)$ and $\cA^*=\cO^*$, splits the admissible logands into $S$-units of $\cO$ --- computed by the machinery of Part I --- and irreducible polynomials moving in the upper variables, whose residues must be constants. A flattening lemma absorbs radical layers with polynomial integral closure into the transcendence basis, reducing certain radical towers to a single radical: Bronstein's example $\int(\log(x)+\sqrt{\log(x)+\sqrt{\log(x)}})/(1+\log(x))\,dx$ becomes a one-radical computation whose residues reproduce the roots of his resultant criterion exactly, with no Hermite reduction, Puiseux expansions or integral-basis computation. We prove degree bounds in the top variable and, through the leading part of the derivation at a place, bounds on the degrees in the lower variables and on the exponents of the special primes wherever that leading part has no kernel --- a condition decided in the towers of the paper, whose failure is exactly the known obstruction below the top, and where the classical guess is shown to miss an elementary integral. The resulting algorithm --- unlike the classical parallel method, whose failure proves nothing --- returns certificates of non-elementarity in two situations: a residue outside the constant field, and, when every bound in force is proved, a residue-free remainder that the linear system shows to be non-exact. Every example is computed by an accompanying SymPy implementation and verified by differentiation.
\end{abstract}

\section{Introduction}

The Risch--Norman method computes elementary antiderivatives over a differential field $F(t_1,\dots,t_n)$ by a single global ansatz: guess the denominator of the integral and the possible \emph{logands} --- the arguments $u_i$ of the logarithms in Liouville's theorem, in the terminology of the parallel-method literature --- bound the degree of the polynomial numerator, and solve a linear system. Part I of this series \cite{PartI} extended the structural facts justifying the method --- the Hermite-type shape of the denominator and the identification of the logands --- to a simple radical extension $L=K(y)$, $y^m=q$, under the standing hypothesis that the transcendental base $K$ is closed under the derivation. That hypothesis silently confines the radical to the \emph{top} of the tower: it excludes every integrand in which a logarithm or exponential of an algebraic function occurs, such as
\[
\int\log\bigl(x+\sqrt{x^2+1}\bigr)\,dx ,
\]
where $Dt=1/y\notin K$ for $t=\log(x+y)$, $y^2=x^2+1$.

Towers with algebraic extensions at arbitrary positions are the subject of Bronstein's algorithm for mixed elementary functions \cite{Bronstein89,Bronstein90}, which proceeds recursively: at each level an integral basis is computed by Trager's algorithm, integrands are normalised by Hermite reduction, residues are extracted by a resultant, and Risch differential equations are solved by matching Puiseux expansions at the special places. In this paper we develop the parallel counterpart for towers containing one simple radical: a single working ring, a single ansatz, no integral-basis computation (the basis is explicit), no Puiseux expansions, and no recursion. The price, as always with the parallel method, is completeness outside the cases where the degree bounds are proved.

The parallel method goes back to Norman and Moore \cite{NormanMoore77}; its structural justification for transcendental towers is due to Davenport \cite{Davenport82} and Davenport and Trager \cite{DavenportTrager85}, refined in Bronstein's last paper \cite{Bronstein07} and consolidated in \cite[Ch.~10]{Bronstein05}, whose \textbf{ParallelIntegrate} is the template for Section~\ref{sec:algorithm}. Boettner's thesis \cite{Boettner10} is the first treatment of mixed transcendental and algebraic extensions within the Risch--Norman framework, and Raab \cite{Raab12} has extended the parallel method in other directions; neither proves structure theorems for the algebraic case, which is what Part I and the present paper supply. On the recursive side, integration of pure algebraic functions descends from Trager \cite{Trager84} and Davenport \cite{Davenport81} to Bronstein's mixed algorithm \cite{Bronstein89,Bronstein90,Bronstein91}; its most recent completion is Schultz's revisiting of Trager's algorithm \cite{Schultz15}, which removes Trager's restrictions at branch and infinite places, decides torsion by reduction modulo two primes of good reduction, and --- rather than failing --- decomposes any algebraic differential into exact, second-kind and third-kind parts over a rationally constructed basis of second-kind normal forms. The overlap with the present paper is the case $n=1$, $D=d/dx$, where our method needs no integral basis and no divisor arithmetic (Cohen's and Schultz's own integrals, \S\S\ref{ex:cohen}--\ref{ex:schultz}, are computed by a residue classification and a continued fraction); Schultz's torsion decision is the natural completion of our unit search (Proposition~\ref{prop:nontorsion}), and his second-kind normal forms are the positive counterpart of our holomorphic-remainder certificate (Proposition~\ref{prop:certificates}(b)). The pseudo-elliptic integrals that appear as benchmarks have a history from Abel through Chebyshev and Zolotarev to Pappalardi and van der Poorten \cite{PvdP05}; see also \cite{BlakePE}.

The contributions are as follows.
\begin{enumerate}
\item The integral closure of $R=F[t_1,\dots,t_n]$ in $L$ is $\cA=\cO[t_{j+1},\dots,t_n]$, free over $R$ on the Trager basis of Part I, and Krull; canonical representations and all factorisation remain in the UFD $R$ (Section~\ref{sec:ring}).
\item A \emph{flattening lemma}: radical layers whose integral closure is a polynomial ring absorb into the transcendence basis, so that certain radical towers --- including Bronstein's example (E) --- reduce to the single-radical setting (Section~\ref{sec:ring}). This gives a partial answer to open problem (b) of Part I.
\item The denominator of the derivation is defined as a divisor $\fd_D$ on $\cA$, and the valuation lemma is proved in the unified local form $v_P(Dg)=v_P(g)-\delta_P$, $\delta_P=1+v_P(\fd_D)$, at all normal height-one primes (Section~\ref{sec:valuations}). This subsumes the shifts $\{1,e_P\}$ of Part I and corrects their taxonomy in the presence of non-monomial generators (Remark~\ref{rem:partone}).
\item The Hermite-type structure theorem and the exhaustive list of special primes for monomial upper generators (Sections~\ref{sec:structure}--\ref{sec:specials}).
\item $\Cl(\cA)\cong\Cl(\cO)$ and $\cA^*=\cO^*$, so the logands split into $S$-units of $\cO$ (Part I) and irreducibles moving in the upper variables; the residue at every normal prime must be a constant, a criterion containing Bronstein's constant-roots resultant test (Section~\ref{sec:logands}).
\item Degree bounds in the top variable, and bounds at every special prime and every place below the top at which the leading part of the derivation has no kernel, with a decision procedure for that condition in monomial towers (Section~\ref{sec:degree}); the algorithm (Section~\ref{sec:algorithm}), whose non-elementarity certificate thereby extends to every mixed tower in which all bounds in force are proved; and worked examples (Section~\ref{sec:examples}), each computed and verified by the accompanying SymPy implementation.
\end{enumerate}

\section{Preliminaries and setting}\label{sec:prelim}

All fields have characteristic $0$ and we follow the notation of \cite{Bronstein05,PartI}. Throughout, $(L,D)$ is a differential field built as follows.

\begin{enumerate}
\item[(T1)] $K_0=F(t_1,\dots,t_j)$ with $DF\subseteq F$, each $t_i$ transcendental over $F(t_1,\dots,t_{i-1})$ and $Dt_i\in F(t_1,\dots,t_i)$; in particular $DK_0\subseteq K_0$. Set $R_0=F[t_1,\dots,t_j]$.
\item[(T2)] $y^m=q\in R_0$ under the normalisations (N1)--(N2) of Part I ($q$ $m$-th-power-free, $y^m-q$ irreducible over $\Fbar K_0$), with squarefree decomposition $q=\prod_l Q_l^{\,l}$.
\item[(T3)] $t_{j+1},\dots,t_n$ are successively transcendental over the field below, with
$Dt_i\in K_0(y)(t_{j+1},\dots,t_i)$; the derivatives may involve $y$.
\item[(T4)] $\Const_D(L)=F$, where $L=K_0(y)(t_{j+1},\dots,t_n)$.
\end{enumerate}

The generators $t_i$ are \emph{not} assumed to be monomials except where explicitly stated (Sections~\ref{sec:specials} and \ref{sec:degree}); this freedom is used essentially by the flattening lemma. We write $R=F[t_1,\dots,t_n]$, a UFD, $K=F(t_1,\dots,t_n)$, and $\cO$ for the integral closure of $R_0$ in $K_0(y)$, with Trager basis $w_i=y^i/E_i$, $E_i=\prod_lQ_l^{\lfloor il/m\rfloor}$, and multiplication table as in Part I, Proposition~3.2. Note that $K$ is in general \emph{not} closed under $D$.

We use the strong Liouville theorem in the form of Part I, equation (2.1). Divisors are Weil divisors on $\operatorname{Spec}$ of a Krull domain, i.e.\ $\Z$-linear combinations of height-one primes; $v_P$ denotes the corresponding discrete valuation, and for $g\in L^*$, $\dv(g)=\sum_Pv_P(g)P$ and $\dv_\infty(g)=\sum_{v_P(g)<0}(-v_P(g))P$.

Verifying (T4) --- that adjoining the generators introduces no new constants --- is out of scope, as in Part I and in \cite[Ch.~10]{Bronstein05}; in practice it holds when the generators are Liouvillian monomials over the field below them, and is decided by the structure theorems for such towers. Its role here is confined to the residue certificate: every residue theorem below states that a certain element of a residue field lies in $\Fbar$, and if (T4) failed, a constant of $L$ outside $\Fbar$ could be mistaken for an obstruction. A false (T4) can therefore only produce a spurious \textbf{NotElementary} at the constancy decision of Algorithm~4, never a wrong integral, since every returned integral is verified by differentiation.

\section{The ring $\cA$ and the flattening lemma}\label{sec:ring}

\begin{proposition}\label{prop:ring}
Let $\cA:=\cO[t_{j+1},\dots,t_n]=\bigoplus_{i=0}^{m-1}R\,w_i$. Then:
\begin{enumerate}
\item[(i)] $\cA$ is the integral closure of $R$ in $L$, and is a Krull domain;
\item[(ii)] every $g\in L$ has a representation $g=\sum_{i=0}^{m-1}(a_i/d)\,w_i$ with $a_i,d\in R$, $\gcd(d,a_0,\dots,a_{m-1})=1$, unique up to units of $F$; we call $d=\den(g)$ the denominator of $g$;
\item[(iii)] $\Cl(\cA)\cong\Cl(\cO)$, induced by extension of divisors, and $\cA^*=\cO^*$.
\end{enumerate}
\end{proposition}

\begin{proof}
(i) Each $w_i$ is integral over $R_0\subseteq R$ and the upper variables lie in $R$, so $\cA$ is integral over $R$, and $\operatorname{Frac}(\cA)=K_0(y)(t_{j+1},\dots,t_n)=L$. By Part I, Proposition~3.2, $\cO$ is integrally closed, hence Krull (a Noetherian normal domain); a polynomial extension of a Krull domain is Krull and integrally closed \cite[\S1]{Fossum73}. An integrally closed ring that is integral over $R$ with fraction field $L$ is the integral closure.

(ii) Freeness over $R$ gives existence once a common denominator in $R$ is found; for $\alpha\in\cA$ one may take $N_{L/K}(\alpha)\in R$, since $\alpha\mid N_{L/K}(\alpha)$ in $\cA$ and the norm of an integral element lies in the integrally closed $R$. Uniqueness is freeness plus the gcd normalisation.

(iii) $\Cl(\cA[t])\cong\Cl(\cA)$ for any Krull domain is Gauss' theorem \cite[Thm.~8.1]{Fossum73}; iterate over the upper variables. For units, $(\mathfrak A[t])^*=\mathfrak A^*$ for any integral domain $\mathfrak A$.
\end{proof}

The height-one primes of $\cA$ fall into two families, by their contraction to $\cO$:
\begin{itemize}
\item \emph{type E} (extended): $P=\fp\cA$ for a height-one prime $\fp$ of $\cO$;
\item \emph{type M} (moving): $P\cap\cO=0$; these correspond to the irreducible polynomials of the UFD $\operatorname{Frac}(\cO)[t_{j+1},\dots,t_n]$, via $P\mapsto P\operatorname{Frac}(\cO)[t_{j+1},\dots,t_n]$.
\end{itemize}

\begin{lemma}[Flattening]\label{lem:flatten}
Suppose the tower contains a radical layer $z^{\mu}=c\,t_k$ with $c\in F^*$ and $t_k$ a generator of the field below the layer. Then replacing $t_k$ by $z$ in the generator list yields a tower of the form \textup{(T1)--(T4)} for the same field $L$, with one fewer radical layer: the layer has been absorbed into the transcendence basis, at the cost of $z$ being a non-monomial generator with $Dz=(c/\mu)\,Dt_k\,z^{1-\mu}$.

More generally, a radical layer $z^\mu=q_1$ absorbs whenever the integral closure of the ring below it in the extended field is a polynomial ring over $F$ whose fraction field is stable under $D$.
\end{lemma}

\begin{proof}
For the pure-root case, $F[\dots,t_k,\dots][z]/(z^\mu-ct_k)\cong F[\dots,z,\dots]$ via $t_k=z^\mu/c$: the layer's integral closure is the polynomial ring in the generators with $t_k$ replaced by $z$, and $Dz=Dq_1/(\mu z^{\mu-1})$ lies in the new fraction field. Transcendence of the remaining generators is unchanged since the new field equals the old one, and (T4) is a property of $L$. The general statement is a restatement of the hypothesis.
\end{proof}

\begin{remark}
The hypothesis is genuinely needed: for $z^2=t^2+1$ over $F(t)$ the affine curve is a torus, its coordinate ring has the nontrivial unit $t+z$ and is not a polynomial ring, and the layer cannot be absorbed --- the unit resurfaces as an unavoidable logand. Thus flattening applies to radical towers in which every layer but one is (after the normalisations of Part I) a pure root of an earlier generator; Example~\ref{ex:bronsteinE} is the canonical instance. This is a partial answer to open problem (b) of Part I; composita of non-flattening radicals remain open.
\end{remark}

\begin{lemma}[Change of derivation]\label{lem:rescale}
Let $r\in L^*$ and $D'=r^{-1}D$. Then $D'$ is a derivation of $L$ with the same field of constants as $D$; an extension $E\supseteq L$, with $D'$ extended to $E$ as $r^{-1}D$, is elementary over $(L,D')$ if and only if it is elementary over $(L,D)$; and for $f\in L$ and $g\in E$, $Dg=f$ if and only if $D'g=f/r$. Hence $f$ has an elementary integral over $(L,D)$ if and only if $f/r$ has one over $(L,D')$, with the same antiderivative.
\end{lemma}

\begin{proof}
Multiplication by $r^{-1}$ preserves additivity and the Leibniz rule, and $D'c=0$ if and only if $Dc=0$. An element $\theta$ of $E$ that is algebraic over a subfield remains so; one with $D\theta=Du/u$ satisfies $D'\theta=D'u/u$, and one with $D\theta=\theta\,Du$ satisfies $D'\theta=\theta\,D'u$: dividing the defining relation by $r$ preserves its shape, and the converse is the same computation with $r$ in place of $r^{-1}$. The last two statements are $D'g=r^{-1}Dg$.
\end{proof}

\noindent The lemma is applied to the towers that the flattening of Lemma~\ref{lem:flatten}, and the rational parametrisation of a conic radicand through a rational point, leave behind: a tower consisting of a single generator $t_1$ over the curve with $Dt_1=r\in F(t_1)^*$ --- a flattened root, $Dt_1=(c/\mu)\,t_1^{1-\mu}$, or the parameter $w$ of a conic, with $Dw$ the derivative of the parameter with respect to $x$. With $D'=r^{-1}D$ the field $L=F(t_1)(y)$ is exactly the setting of Part~I, $D't_1=1$, and the integrand becomes $f/r$; the exact degree bounds of Part~I and Proposition~\ref{prop:certificates}(b) below then apply verbatim. The irreducible factors $p$ of the numerator of $r$ are special primes of $D$ ($Dp=p'r\equiv0\pmod p$) and ordinary primes of $D'$, so the rescaled tower has no specials at all; the branch primes are unchanged. The implementations perform the rescaling on entry, and \S\ref{ex:nested} shows it deciding a nested radical.

\section{The denominator divisor and the valuation lemma}\label{sec:valuations}

In Part I the derivation satisfied $DK\subseteq K$ and its denominator was an element $\den_D(K)\in R$. That fails here: in the tower of the flagship example (Section~\ref{sec:examples}), the least common ``denominator'' of $Dy$ and $Dt$ is the branch divisor, generated by $y\in\cO\setminus R$. The correct object is a divisor.

\begin{definition}\label{def:eta}
For a height-one prime $P$ of $\cA$ let
\[
\eta_P:=\max\Bigl(0,\ \sup_{a\in\cA\setminus\{0\}}\bigl(-v_P(Da)\bigr)\Bigr),
\]
the \emph{pole order of $D$ at $P$}, and $\fd_D:=\sum_P\eta_P\,P$, the \emph{denominator divisor} of $(\cA,D)$.
\end{definition}

\begin{lemma}\label{lem:etagen}
Let $\mathcal G=\{t_1,\dots,t_n,w_1,\dots,w_{m-1}\}$. Then $\eta_P=\max\bigl(0,\max_{g\in\mathcal G}(-v_P(Dg))\bigr)$; in particular $\eta_P$ is finite, vanishes for all but finitely many $P$, and $\fd_D$ is an effective divisor. Moreover the same value is obtained with the supremum taken over the local ring $\cA_P$.
\end{lemma}

\begin{proof}
Every $a\in\cA$ is an $F$-polynomial in $\mathcal G$, so by the Leibniz rule $Da$ is an $\cA$-linear combination of $\{Dg:g\in\mathcal G\}$, whence $-v_P(Da)\le\max_g(-v_P(Dg))$; the bound is attained at $a=g$. Each $Dg$ has finitely many polar primes. For the local statement, $D(a/s)=Da/s-aDs/s^2$ with $s\notin P$ has $v_P\ge\min(v_P(Da),v_P(a)+v_P(Ds))\ge-\eta_P$.
\end{proof}

\begin{definition}\label{def:normal}
Fix $P$ and any $h\in L^*$ with $v_P(h)=\eta_P$, and set $\Dbar_P:=hD$, a derivation of $L$ mapping $\cA_P$ into $\cA_P$ by Lemma~\ref{lem:etagen}. Let $\pi$ be a local uniformiser at $P$. We call $P$ \emph{special} if $v_P(\Dbar_P\pi)\ge1$, and \emph{normal} if $v_P(\Dbar_P\pi)=0$.
\end{definition}

\begin{lemma}
The dichotomy of Definition~\ref{def:normal} is exhaustive and independent of the choices of $h$ and $\pi$.
\end{lemma}

\begin{proof}
$\pi\in\cA_P$ gives $v_P(\Dbar_P\pi)\ge0$, so the dichotomy is exhaustive. Changing $h$ multiplies $\Dbar_P$ by a $P$-unit. If $\pi'=u\pi$ with $u\in\cA_P^*$, then $\Dbar_P\pi'=u\Dbar_P\pi+\pi\Dbar_Pu$ and $v_P(\pi\Dbar_Pu)\ge1$, so $v_P(\Dbar_P\pi')=0$ iff $v_P(\Dbar_P\pi)=0$.
\end{proof}

\begin{theorem}[Valuation lemma]\label{thm:shift}
Let $P$ be a normal height-one prime of $\cA$ and set $\delta_P:=1+\eta_P=1+v_P(\fd_D)$. Then for every $g\in L^*$ with $v_P(g)\ne0$,
\[
v_P(Dg)=v_P(g)-\delta_P .
\]
For every $g\in L^*$ with $v_P(g)=0$, $v_P(Dg)\ge-\eta_P=1-\delta_P$.
\end{theorem}

\begin{proof}
Write $g=u\pi^k$ with $k=v_P(g)$ and $u\in\cA_P^*$. Then
\[
\Dbar_Pg=k\,u\,\pi^{k-1}\,\Dbar_P\pi+\pi^k\,\Dbar_Pu .
\]
Since $P$ is normal, $v_P(\Dbar_P\pi)=0$; since $\operatorname{char}F=0$ and $k\ne0$, the first summand has $v_P$ exactly $k-1$. Since $\Dbar_P$ maps $\cA_P$ to itself and $u\in\cA_P$, the second summand has $v_P\ge k$. Hence $v_P(\Dbar_Pg)=k-1$ and $v_P(Dg)=k-1-\eta_P$. The last assertion is Lemma~\ref{lem:etagen}.
\end{proof}

\begin{remark}
No cancellation analysis is required: the term $\pi^k\Dbar_Pu$ is subordinate \emph{automatically}, because $\Dbar_P$ preserves the local ring. The leading-coefficient computations of \cite[Lemma~10.2.1]{Bronstein05} and Part I, Lemma~4.2, reappear only in the closed-form evaluation of $\eta_P$ below.
\end{remark}

The definitions and results above use only that the ambient ring is a Krull domain and that $D$ is a derivation of its fraction field: Definition~\ref{def:eta}, Lemma~\ref{lem:etagen} (with an appropriate generating set), Definition~\ref{def:normal} and Theorem~\ref{thm:shift} will accordingly be applied below not only to $(\cA,D)$ but also to $(R_0,D)$ and $(\cO,D)$.

We now compute $\eta_P$ and the normal/special dichotomy in the two prime families. Write $\den_0:=\den_D(K_0)\in R_0$ and $\Dbar_0:=\den_0 D$, the derivation of $R_0$.

\begin{lemma}[The shift below the upper tower]\label{lem:lowershift}
Let $\fp$ be a height-one prime of $\cO$ lying over the irreducible $p\in R_0$, with ramification index $e:=e(\fp|p)$, and suppose $p$ is normal for $\Dbar_0$, i.e.\ $\gcd(p,\Dbar_0p)=1$. Set $\nu:=v_p(\den_0)$. Then, for the datum $(\cO,D)$,
\[
\eta_\fp=e(1+\nu)-1,
\]
$\fp$ is normal, and consequently, for every $g\in K_0(y)^*$,
\[
v_\fp(Dg)=v_\fp(g)-e(1+\nu)\ \ \text{if }v_\fp(g)\ne0,
\qquad
v_\fp(Dg)\ge1-e(1+\nu)\ \ \text{if }v_\fp(g)=0 .
\]
\end{lemma}

\begin{proof}
\emph{Base level.} Let $h\in K_0^*$ with $k:=v_p(h)\ne0$ and write $h=up^k$ with $u$ a $p$-unit of $(R_0)_{(p)}$. Then $\Dbar_0h=k\,u\,p^{k-1}\Dbar_0p+p^k\Dbar_0u$; the first term has $v_p$ exactly $k-1$ (normality gives $p\nmid\Dbar_0p$, and $k\ne0$ in characteristic $0$), the second has $v_p\ge k$. Hence $v_p(\Dbar_0h)=k-1$ and $v_p(Dh)=k-1-\nu$; and for $v_p(h)=0$ one has $v_p(Dh)\ge-\nu$. The latter bound is attained when $\nu>0$: some lower generator $t_i$ has $v_p(\den(Dt_i))=\nu$ (the lcm defining $\den_0$), and in lowest terms $v_p(Dt_i)=-\nu$ exactly.

\emph{Pole order on $\cO$.} From $Dy/y=\tfrac1m\sum_l l\,DQ_l/Q_l$ and $DE_i/E_i=\sum_l\lfloor il/m\rfloor\,DQ_l/Q_l$,
\begin{equation}\label{eq:fracpart}
\frac{Dw_i}{w_i}\;=\;\sum_l\Bigl\{\frac{il}{m}\Bigr\}\,\frac{DQ_l}{Q_l}\,,
\end{equation}
where $\{x\}$ denotes the fractional part.

If $\fp$ is unramified ($p\nmid Q$, $e=1$), then $v_\fp(DQ_l/Q_l)=v_p(DQ_l)\ge-\nu$ for every $l$, so $-v_\fp(Dw_i)\le\nu$; with the base level and Lemma~\ref{lem:etagen} for $(\cO,D)$, $\eta_\fp=\nu=e(1+\nu)-1$, attained as above. The uniformiser $\pi=p$ has $v_\fp(Dp)=1-1-\nu=-\nu=-\eta_\fp$, so $\fp$ is normal.

If $\fp$ is ramified, then $p\mid Q_{l_0}$ for exactly one $l_0$, $e=m/\gcd(m,l_0)$, $v_\fp(y)=l_0/\gcd(m,l_0)$, and
\[
v_\fp(w_i)=i\,v_\fp(y)-e\Bigl\lfloor\frac{il_0}{m}\Bigr\rfloor=\frac{il_0\bmod m}{\gcd(m,l_0)}\in[0,e)
\]
(Part I, \S4). Writing $Q_{l_0}=pc$ with $p\nmid c$ ($Q_{l_0}$ is squarefree), $\Dbar_0Q_{l_0}\equiv c\,\Dbar_0p\not\equiv0\pmod p$, so $v_p(\Dbar_0Q_{l_0})=0$ and
\[
v_\fp\Bigl(\frac{DQ_{l_0}}{Q_{l_0}}\Bigr)=e\,v_p(DQ_{l_0})-e\,v_p(Q_{l_0})=-e\nu-e=-e(1+\nu)
\]
exactly, while $v_\fp(DQ_l/Q_l)\ge-e\nu$ for $l\ne l_0$. Hence, in \eqref{eq:fracpart}: if $m\nmid il_0$ the $l_0$-term strictly dominates and
\[
v_\fp(Dw_i)=v_\fp(w_i)-e(1+\nu),\qquad v_\fp(w_i)\ge1
\]
($il_0\bmod m$ is then a nonzero multiple of $\gcd(m,l_0)$); if $m\mid il_0$ then $v_\fp(w_i)=0$ and $-v_\fp(Dw_i)\le e\nu$. Together with $-v_\fp(Dt_i)\le e\nu$ for $i\le j$, Lemma~\ref{lem:etagen} for $(\cO,D)$ gives $\eta_\fp\le e(1+\nu)-1$. Choose $i^*\in\{1,\dots,m-1\}$ with $i^*l_0\equiv\gcd(m,l_0)\pmod m$ (B\'ezout): then $v_\fp(w_{i^*})=1$, so $w_{i^*}$ is a uniformiser, $-v_\fp(Dw_{i^*})=e(1+\nu)-1$ attains the bound, and
\[
v_\fp\bigl(\Dbar_\fp w_{i^*}\bigr)=\eta_\fp+v_\fp(Dw_{i^*})=\bigl(e(1+\nu)-1\bigr)+\bigl(1-e(1+\nu)\bigr)=0,
\]
so $\fp$ is normal. In both cases Theorem~\ref{thm:shift}, applied to $(\cO,D)$, yields the displayed shift.
\end{proof}

\begin{proposition}[Type E primes: classification]\label{prop:typeE}
Let $P=\fp\cA$ be a type E prime, with $\fp$ lying over the irreducible $p\in R_0$ with ramification index $e_P$, and set
\[
\mu_P:=\max_{j<i\le n}\bigl(-v_P(Dt_i)\bigr)\qquad(\mu_P:=0\text{ if }n=j).
\]
\begin{enumerate}
\item[(i)] If $p$ is normal for $\Dbar_0$, then
\[
\eta_P=\max\Bigl(e_P\bigl(1+v_p(\den_0)\bigr)-1,\ \mu_P\Bigr),
\]
and
\[
P\ \text{is normal}\iff\mu_P\le e_P\bigl(1+v_p(\den_0)\bigr)-1,
\]
in which case $\delta_P=e_P(1+v_p(\den_0))$; otherwise $P$ is special.
\item[(ii)] If $p$ is special for $\Dbar_0$, then $P$ is special.
\end{enumerate}
In particular, $\delta_P\ge e_P$ at every normal type E prime.
\end{proposition}

\begin{proof}
Throughout, $v_P$ restricts to $v_\fp$ on $K_0(y)^*$ (a uniformiser of $\fp$ generates $P\cA_P$), and the generators of $\cA$ are those of $\cO$ together with $t_{j+1},\dots,t_n$, so Lemma~\ref{lem:etagen} gives
\begin{equation}\label{eq:etamax}
\eta_P=\max\bigl(\eta_\fp,\ \mu_P\bigr).
\end{equation}
Write $\nu:=v_p(\den_0)$.

(i) By Lemma~\ref{lem:lowershift}, $\eta_\fp=e_P(1+\nu)-1$, which with \eqref{eq:etamax} is the stated formula. Suppose $\mu_P\le\eta_\fp$, so $\eta_P=\eta_\fp$. Take the uniformiser $\pi\in K_0(y)$ produced in Lemma~\ref{lem:lowershift} ($\pi=p$ or $\pi=w_{i^*}$), for which $v_P(D\pi)=-\eta_\fp$; then $v_P(\Dbar_P\pi)=\eta_P+v_P(D\pi)=0$ and $P$ is normal, with $\delta_P=1+\eta_P=e_P(1+\nu)$. Suppose instead $\mu_P>\eta_\fp$. For \emph{any} uniformiser $\pi\in K_0(y)$ of $\fp$, Lemma~\ref{lem:lowershift} gives $v_P(D\pi)\ge-\eta_\fp$, whence
\[
v_P(\Dbar_P\pi)=\eta_P+v_P(D\pi)\ \ge\ \mu_P-\eta_\fp\ \ge\ 1,
\]
and $P$ is special (the dichotomy is independent of the uniformiser).

(ii) Let $\pi$ be a uniformiser of $\fp$ and set $w:=\pi^{e_P}/p\in K_0(y)$, a $\fp$-unit. Logarithmic differentiation of $\pi^{e_P}=pw$ gives
\[
\frac{D\pi}{\pi}=\frac1{e_P}\Bigl(\frac{Dp}{p}+\frac{Dw}{w}\Bigr).
\]
Since $p$ is special, $v_p(\Dbar_0p)\ge1$, so $v_p(Dp)\ge1-\nu$ and $v_P(Dp/p)\ge e_P(1-\nu)-e_P=-e_P\nu$; and $v_P(Dw/w)=v_P(Dw)\ge-\eta_P$ since $w$ is a $P$-unit. Moreover $\eta_P\ge e_P\nu$: for $\nu>0$ this is witnessed by a lower generator in lowest terms, as in the base level of Lemma~\ref{lem:lowershift}. Hence
\[
v_P(D\pi)\ \ge\ 1-\max\bigl(e_P\nu,\ \eta_P\bigr)\ =\ 1-\eta_P,
\]
so $v_P(\Dbar_P\pi)\ge1$ and $P$ is special.

The final claim is immediate from (i): at a normal type E prime, $\delta_P=e_P(1+\nu)\ge e_P$.
\end{proof}

\begin{example}[Normality depends on the tower]\label{ex:towerdep}
Over the curve $y^2=x^2+1$ one has $e=2$ and $\den_0=1$ at the branch prime $(y)$, so $e(1+\nu)-1=1$. In the tower $t=\log(x+y)$, $Dt=1/y$ gives $\mu=1\le1$: $(y)$ is \emph{normal}, $\delta=2$. In the tower $t=\log(y)$, $Dt=x/y^{2}$ gives $\mu=2>1$: $(y)$ is \emph{special}. The same prime of $\cO$ is classified oppositely in the two towers --- normality is a property of the tower, not of the curve. The criterion also recovers the classical facts: for $t=\log(g)$ with $v_\fp(g)\ne0$, Lemma~\ref{lem:lowershift} gives $-v_P(Dt)=e_P(1+\nu)$, exceeding the threshold, so the primes of the argument of a logarithm are special; and a hyperexponential $t=e^{\eta}$ contributes $\mu$-values $-v_P(D\eta)$, so exactly the sufficiently deep type E poles of $D\eta$ become special.
\end{example}

\begin{proposition}[Type M primes]\label{prop:typeM}
Write $\mathcal F:=K_0(y)$ and $T:=\mathcal F[t_{j+1},\dots,t_n]$, and let $h_u\in T$ be the least common multiple of the $T$-denominators of $Dt_{j+1},\dots,Dt_n$ (the \emph{moving denominator}), so that $\Dbar_u:=h_uD$ is a derivation of $T$ (as $D\mathcal F\subseteq\mathcal F$). Let $P$ be a moving prime, corresponding to the monic irreducible $p\in T$. Then $\cA_P=T_{(p)}$, $\eta_P=v_p(h_u)$, and $P$ is special in the sense of Definition~\ref{def:normal} if and only if $p\mid\Dbar_up$ in $T$, i.e.\ if and only if $p$ is special for $\Dbar_u$ in the sense of \cite[\S10.1]{Bronstein05}.
\end{proposition}

\begin{proof}
Every nonzero element of $\cO$ lies outside $P$, so localising $\cA=\cO[t_{j+1},\dots,t_n]$ at $P$ inverts $\cO\setminus\{0\}$: $\cA_P=T_{(p)}$, a DVR with uniformiser $p$. Since $\Dbar_uT\subseteq T$, one has $-v_p(Da)\le v_p(h_u)$ for all $a\in T$, whence $\eta_P\le v_p(h_u)$ by the generator computation of Lemma~\ref{lem:etagen}; and if $v_p(h_u)>0$, some $Dt_i$ written in lowest terms over $T$ has $v_p(Dt_i)=-v_p(h_u)$ exactly (the lcm), so $\eta_P=v_p(h_u)$. Finally, with $\Dbar_P=h'D$ for any $h'$ of $v_p(h')=\eta_P$,
\[
v_P(\Dbar_Pp)\;=\;\eta_P+v_p(Dp)\;=\;v_p\bigl(h_uDp\bigr)\;=\;v_p(\Dbar_up),
\]
so Definition~\ref{def:normal} reads: $P$ special $\iff p\mid\Dbar_up$.
\end{proof}

\begin{remark}[Reconciliation with Part I, and a correction]\label{rem:partone}
When the lower tower is monomial, every irreducible factor of $\den_0$ is special \cite[Thm.~10.2.2]{Bronstein05}, so normal $p$ have $v_p(\den_0)=0$ and Proposition~\ref{prop:typeE} returns $\delta_P=1+\eta_P\in\{1,e_P\}$: the taxonomy of Part I. With non-monomial generators the two notions of normality come apart. For $t=\sqrt{\log(x)}$ over $\Q(x)$ one has $Dt=1/(2xt)$, $\den_0=xt$ and $\Dbar_0t=\tfrac12$: thus $t$ is normal \textup(gcd definition\textup) yet divides $\den_0$, and the shift at $(t)$ is $\delta=2$, as $D(1/t)=-1/(2xt^3)$ confirms. Consequently the characterisation ``$p$ normal iff $p\nmid\den_D(K)$ and $p\nmid Dp$'' of Part I, Section~2, and the range $\delta_P\in\{1,e_P\}$, are correct only under the monomial hypothesis; the general statement is Theorem~\ref{thm:shift} with $\delta_P=1+v_P(\fd_D)$.
Precisely: Part I, Lemma~4.2 asserts $\delta_P\in\{1,e_P\}$ for every $P$ over a $p$ that is normal in the sense $\gcd(p,\Dbar_Rp)=1$, with no monomial hypothesis in its standing assumptions; its unramified-case proof takes $v_p(Dp)=0$, which requires in addition $p\nmid\den_D$. For $t=\sqrt{\log(x)}$, $Dt=1/(2xt)$, the prime $t$ is normal in that sense yet $\delta_{(t)}=2$. Part I's lemma is thus correct under the additional hypothesis $p\nmid\den_D$ --- automatic for monomial towers --- and Proposition~\ref{prop:typeE} is its general form; this is an erratum to Part I.
\end{remark}

\begin{example}\label{ex:dendiv}
Flagship tower: $t_1=x$, $y^2=x^2+1$, $t=\log(x+y)$, so $Dy=x/y$ and $Dt=1/y$. The polar divisors of $Dy$ and $Dt$ are both $\dv_0(y)$; since $\cO/(y)\cong\Q[x]/(x^2+1)$ is a field, $(y)$ is a height-one prime and
\[
\fd_D=\dv(y)
\]
is principal --- but generated by $y\in\cO\setminus R$, so no element of $R$ realises it without a spurious factor: $N_{L/K}(y)=-(x^2+1)$ would inflate $\fd_D$ by a square. Here $\Dbar=yD$ is a derivation of $\cA$ ($yDx=y$, $yDy=x$, $yDt=1$), $\delta_{(y)}=1+1=2=e_{(y)}$, and every other prime has $\eta_P=0$.
\end{example}

\section{Structure of elementary integrals}\label{sec:structure}

Fix $f\in L^*$ with an elementary integral over $L$. By the strong Liouville theorem (Part I, eq.~(2.1)) there are $v\in L$, $c_1,\dots,c_r\in\Fbar$ and $u_1,\dots,u_r\in(\Fbar L)^*$ with
\begin{equation}\label{eq:liou}
f\;=\;Dv+\sum_{i=1}^r c_i\,\frac{Du_i}{u_i}\,.
\end{equation}
The analysis takes place on the constant-field extension $\Fbar\cA:=\Fbar\otimes_F\cA$.

\begin{lemma}[Base change]\label{lem:basechange}
$D$ extends uniquely to $\Fbar L$ with $D\Fbar=0$; $\Fbar\cA$ is the integral closure of $\Fbar R$ in $\Fbar L$ and is Krull; and every height-one prime $\bar P$ of $\Fbar\cA$ contracts to a height-one prime $P$ of $\cA$ with $v_{\bar P}|_{L^*}=v_P$ and $\eta_{\bar P}=\eta_P$, and $\bar P$ is normal if and only if $P$ is. In particular, Theorem~\ref{thm:shift} holds on $\Fbar\cA$ for all $g\in(\Fbar L)^*$, with the same $\delta$.
\end{lemma}

\begin{proof}
By (N2), $y^m-q$ remains irreducible over $\Fbar K_0$, so $[\Fbar L:\Fbar K]=m$ and $\Fbar\cA=\bigoplus_i\Fbar R\,w_i$. It is integral over $\Fbar R$ with fraction field $\Fbar L$, and it is a filtered union of the rings $F'\cA$ over the finite extensions $F'/F$ inside $\Fbar$, each of which is normal ($F'\otimes_F-$ is an \'etale base change in characteristic $0$); hence $\Fbar\cA$ is integrally closed and Krull. A constant extension is unramified in codimension one, so a uniformiser $\pi$ of $\cA_P$ remains one in $(\Fbar\cA)_{\bar P}$ and $v_{\bar P}$ restricts to $v_P$ on $L^*$. The set $\mathcal G$ of Lemma~\ref{lem:etagen} still generates $\Fbar\cA$ as an $\Fbar$-algebra and $D$ kills $\Fbar$, so $\eta_{\bar P}=\eta_P$ by Lemma~\ref{lem:etagen}, and $v_{\bar P}(\Dbar_{\bar P}\pi)=v_P(\Dbar_P\pi)$ decides normality identically. The proof of Theorem~\ref{thm:shift} nowhere used the constant field.
\end{proof}

\begin{lemma}[Regrouping]\label{lem:regroup}
The representation \eqref{eq:liou} may be chosen with $c_1,\dots,c_r$ linearly independent over $\Q$.
\end{lemma}

\begin{proof}
Standard (Rosenlicht; cf.\ the proof of \cite[Thm.~5.5.3]{Bronstein05}): if $Mc_r=\sum_{i<r}m_ic_i$ with $M,m_i\in\Z$, $M\neq0$, then
\[
\sum_{i=1}^{r}c_i\frac{Du_i}{u_i}\;=\;\sum_{i=1}^{r-1}\frac{c_i}{M}\;\frac{D\bigl(u_i^{M}u_r^{\,m_i}\bigr)}{u_i^{M}u_r^{\,m_i}}\,,
\]
with the new logands again in $(\Fbar L)^*$; iterate until independence holds.
\end{proof}

\begin{lemma}[Logarithmic derivatives at normal primes]\label{lem:logder}
Let $P$ be a normal height-one prime of $\Fbar\cA$ with uniformiser $\pi$, and let $u\in(\Fbar L)^*$. Then
\[
v_P\!\left(\frac{Du}{u}-v_P(u)\,\frac{D\pi}{\pi}\right)\;\ge\;1-\delta_P,
\qquad
v_P\!\left(\frac{D\pi}{\pi}\right)=-\delta_P\,.
\]
In particular $v_P(Du/u)\ge-\delta_P$, with equality if and only if $v_P(u)\neq0$.
\end{lemma}

\begin{proof}
Write $u=w\pi^{k}$ with $k=v_P(u)$ and $w$ a $P$-unit. Then $Du/u=k\,D\pi/\pi+Dw/w$, and $v_P(Dw/w)=v_P(Dw)\ge-\eta_P=1-\delta_P$ by the last clause of Theorem~\ref{thm:shift}, while $v_P(D\pi/\pi)=(1-\delta_P)-1=-\delta_P$ by the same theorem applied to $\pi$.
\end{proof}

\begin{theorem}[Structure of elementary integrals]\label{thm:structure}
Let $f\in L^*$ have an elementary integral over $L$, and set
\[
S\;:=\;\{P\ \text{special}\}\;\cup\;\{P\ \text{normal}:\ v_P(f)\le-\delta_P\},
\]
a set of height-one primes of $\Fbar\cA$, finite whenever the special primes are finite in number \textup(e.g.\ under the hypotheses of Section~\ref{sec:specials}\textup). Then the representation \eqref{eq:liou} may be chosen so that:
\begin{enumerate}
\item[(i)] at every normal height-one prime $P$,
\[
v_P(v)\;\ge\;\min\bigl(0,\;v_P(f)+\delta_P\bigr),
\]
with equality $v_P(v)=v_P(f)+\delta_P$ whenever $v_P(v)<0$; consequently
\[
\dv_\infty(v)\big|_{\mathrm{normal}}\;\le\;\sum_{\substack{P\ \mathrm{normal}\\ v_P(f)\le-\delta_P-1}}\bigl(-v_P(f)-\delta_P\bigr)\,P\,;
\]
\item[(ii)] every $u_i$ is an $S$-unit: $\dv(u_i)$ is supported on $S$.
\end{enumerate}
\end{theorem}

\begin{proof}
Choose \eqref{eq:liou} with $c_1,\dots,c_r$ linearly independent over $\Q$ (Lemma~\ref{lem:regroup}); we show that (i) and (ii) then hold as stated, all valuations being taken on $\Fbar\cA$ (Lemma~\ref{lem:basechange}).

(i) Let $P$ be normal and suppose $k:=v_P(v)<0$. Theorem~\ref{thm:shift} gives $v_P(Dv)=k-\delta_P\le-1-\delta_P$, whereas $v_P\bigl(\sum_ic_iDu_i/u_i\bigr)\ge-\delta_P$ by Lemma~\ref{lem:logder}. The two valuations in \eqref{eq:liou} being distinct,
\[
v_P(f)\;=\;\min\Bigl(v_P(Dv),\,v_P\Bigl(\sum_ic_i\tfrac{Du_i}{u_i}\Bigr)\Bigr)\;=\;k-\delta_P\,,
\]
so $v_P(v)=v_P(f)+\delta_P$ and in particular $v_P(f)\le-\delta_P-1$. If $v_P(v)\ge0$ there is nothing to prove; the divisor bound restates the two cases.

(ii) Let $P$ be normal with $P\notin S$, i.e.\ $v_P(f)\ge1-\delta_P$; we claim $v_P(u_i)=0$ for every $i$. First, $v_P(v)\ge0$: otherwise (i) forces $v_P(f)\le-\delta_P-1$. Hence $v_P(Dv)\ge1-\delta_P$, by Theorem~\ref{thm:shift} if $v_P(v)>0$ and by its last clause if $v_P(v)=0$. Next, by Lemma~\ref{lem:logder},
\[
\sum_{i}c_i\frac{Du_i}{u_i}\;=\;\Bigl(\sum_{i}c_i\,v_P(u_i)\Bigr)\frac{D\pi}{\pi}\;+\;\varepsilon,
\qquad v_P(\varepsilon)\ge1-\delta_P\,.
\]
Substituting into \eqref{eq:liou} and using $v_P(f)\ge1-\delta_P$, $v_P(Dv)\ge1-\delta_P$, the term $\bigl(\sum_ic_iv_P(u_i)\bigr)D\pi/\pi$, of exact valuation $-\delta_P<1-\delta_P$ unless its coefficient vanishes, must vanish:
\[
\sum_{i=1}^{r}c_i\,v_P(u_i)\;=\;0\,.
\]
Since the $c_i$ are $\Q$-linearly independent and the $v_P(u_i)$ are integers, $v_P(u_i)=0$ for all $i$. Thus $\dv(u_i)$ is supported on $S$.
\end{proof}

\begin{remark}
Statement (ii) with $P$ ranging over the moving primes, combined with the identification $\tau_P(f)=\sum_ic_iv_P(u_i)$ of Section~\ref{sec:logands}, is the residue-constancy criterion; statement (ii) with $P$ of type E feeds the $S$-unit and torsion machinery of Part I through Proposition~\ref{prop:split}. The proof above replaces the global squarefree-lcm argument of \cite[Thm.~10.2.1]{Bronstein05} and Part I, Theorem~5.1 by a purely local computation at each prime; the global inputs enter only through $\delta_P=1+v_P(\fd_D)$. Being local, the same proofs apply at any discrete valuation of $\Fbar L$ with finite pole order $\eta$, whether or not it arises from a height-one prime of $\Fbar\cA$; this is used at the place at infinity in Section~\ref{sec:degree}.
\end{remark}

\begin{corollary}[Element form]\label{cor:element}
Write $f=\sum_i(a_i/d)w_i$ as in Proposition~\ref{prop:ring}(ii), and factor $d=d_\ast d_n$ in $R$, where $d_n$ collects those irreducible factors of $d$ \emph{all} of whose height-one primes in $\Fbar\cA$ are normal, and $d_\ast$ the rest; let $d_n=\prod_{l\ge1}d_l^{\,l}$ be the squarefree decomposition. Recall that $\delta_P\ge e_P$ at every normal prime $P$, where $e_P:=v_P(p)$ for the contraction $(p)=P\cap R$: at moving primes $e_P=1\le\delta_P$, and at type E primes this is Proposition~\ref{prop:typeE}. Then the rational part $v$ of Theorem~\ref{thm:structure} satisfies
\[
\den(v)\ \Big|\ s\,\prod_{l\ge2}d_l^{\,l-1}
\]
for some $s\in R$ each of whose irreducible factors has at least one special prime in $\Fbar\cA$.
\end{corollary}

\begin{proof}
We bound $v_P(v)$ at every normal prime $P$, with $(p)=P\cap R$ and $e_P=v_P(p)$.

If $p\nmid d$ then $v_P(f)\ge0>-\delta_P$, so Theorem~\ref{thm:structure}(i) gives $v_P(v)\ge0$.

If $p\mid d_n$ with $v_p(d)=l$, then $v_P(f)\ge-v_P(d)=-e_P\,l$, so
\[
v_P(v)\;\ge\;\min\bigl(0,\,-e_P\,l+\delta_P\bigr)\;\ge\;-e_P\,l+e_P\;=\;-v_P\bigl(p^{\,l-1}\bigr),
\]
using $\delta_P\ge e_P$. In particular $l=1$ contributes no pole.

If $p\mid d_\ast$, then $v_P(v)\ge\min(0,-e_P\,v_p(d)+\delta_P)$ is again finite and bounded below.

Hence $v\cdot\prod_{l\ge2}d_l^{\,l-1}$ has poles only at primes lying over irreducible factors of $d_\ast$ and at special primes; both kinds of prime contract to irreducibles of $R$ possessing at least one special prime of $\Fbar\cA$ (for $d_\ast$-factors, by the definition of $d_\ast$). As $v$ has finitely many poles, some product $s$ of such irreducibles clears them.
\end{proof}

\begin{remark}
The valuation lemma says nothing at a special prime, and in \cite[\S10.3]{Bronstein05} the exponents of the special primes are guessed. Section~\ref{sec:lower} bounds them by the shift of the uniformiser at every special prime where the leading part of $D$ has no kernel in negative degrees, which covers the specials of the towers of this paper with the exceptions listed in Remark~\ref{rem:decided}; the algorithm guesses only there. Everything else in the ansatz is proved: the normal part of the denominator is $\prod_{l\ge2}d_l^{\,l-1}$, and the logands are $S$-units.
\end{remark}

\section{Special primes of monomial upper towers}\label{sec:specials}

Assume in this section that each upper generator $t_i$ ($i>j$) is a primitive, hyperexponential or hypertangent monomial over the field below it (which contains $y$): respectively $Dt_i\in$, $Dt_i/t_i\in$, or $Dt_i/(1+t_i^2)\in K_0(y)(t_{j+1},\dots,t_{i-1})$, with $t_i$ transcendental and no new constants. Recall from Proposition~\ref{prop:typeM} the moving denominator $h_u\in T=\mathcal F[t_{j+1},\dots,t_n]$, $\mathcal F=K_0(y)$, and the derivation $\Dbar_u=h_uD$ of $T$.

\begin{theorem}\label{thm:specials}
The special height-one primes of $\cA$ are exactly:
\begin{enumerate}
\item[(i)] the type E primes over $\Dbar_0$-special $p\in R_0$, together with the type E primes over $\Dbar_0$-normal $p$ at which some upper generator satisfies $-v_P(Dt_i)\ge e_P(1+v_p(\den_0))$ \textup(Proposition~\ref{prop:typeE}\textup);
\item[(ii)] the moving primes dividing the moving denominator: the primes of the irreducible factors of $h_u$ in $T$;
\item[(iii)] the moving primes $(t_i)$ for hyperexponential $t_i$ \textup(including $t_i=e^{\eta}$ with $\eta$ algebraic over $K_0$\textup);
\item[(iv)] the moving primes over the factors of $1+t_i^2$ for hypertangent $t_i$.
\end{enumerate}
In particular the special primes are finite in number.
\end{theorem}

\begin{proof}
The type E classification is Proposition~\ref{prop:typeE}, giving (i). By Proposition~\ref{prop:typeM}, a moving prime with monic irreducible $p\in T$ is special if and only if $p\mid\Dbar_up$ in $T$, i.e.\ if and only if $p$ is special for the derivation $\Dbar_u$ of $T$ in the sense of \cite[\S10.1]{Bronstein05}. Now $L=\mathcal F(t_{j+1},\dots,t_n)$ is a tower of primitive, hyperexponential and hypertangent monomial extensions of the differential field $\mathcal F$, which is closed under $D$ and has $\Const_D(\mathcal F)=F$ by (T4). The classification of the $\Dbar_u$-specials of $T$ is therefore \cite[Thm.~10.2.2]{Bronstein05} with the coefficient field $F$ replaced by $\mathcal F$: its proof --- Gauss' lemma and induction on the main variable, together with the single-variable classification of the special polynomials of primitive, hyperexponential and hypertangent monomials (the monomial-extension theory of \cite[\S3.4]{Bronstein05} and, for the hypertangent case, \cite[\S5.10]{Bronstein05}) --- uses only that the coefficient field is a differential field of characteristic $0$ with the same constants, and carries over verbatim.

We record the forward inclusions, which are direct computations. If $p\mid h_u$ has main variable $t_i$, then $Dp$ lies in $\mathcal F(t_{j+1},\dots,t_{i-1})[t_i]$, because the coefficients of $p$ and the derivatives $Dt_{i'}$, $i'\le i$, do; write $bDp\in T$ with $b\in\mathcal F[t_{j+1},\dots,t_{i-1}]$, $b\ne0$. Then $b\,\Dbar_up=h_u\,(bDp)$ is divisible by $p$, and $p\nmid b$ as $t_i$ occurs in $p$ but not in $b$; since $p$ is irreducible, $p\mid\Dbar_up$. If $t_i$ is hyperexponential, $\Dbar_ut_i=(h_u\,Dt_i/t_i)\,t_i$ with $h_uDt_i/t_i\in T$, so $t_i\mid\Dbar_ut_i$. If $t_i$ is hypertangent with $Dt_i=w_i(1+t_i^2)$, then over $\Fbar$, $D(t_i\pm\sqrt{-1})=w_i(t_i+\sqrt{-1})(t_i-\sqrt{-1})$ is divisible by $t_i\pm\sqrt{-1}$, so each factor of $1+t_i^2$ is special.
\end{proof}

\begin{remark}
The tower-dependence of Example~\ref{ex:towerdep} is confined to type E: a moving prime dividing the polar divisor of an upper generator is \emph{always} special, because $D\mathcal F\subseteq\mathcal F$ contributes no pole at a moving prime --- there is no analogue of the lower term $e_P(1+v_p(\den_0))-1$ that could dominate the upper contribution. The normal/special threshold of Proposition~\ref{prop:typeE} is thus an intrinsically algebraic phenomenon, invisible both in the transcendental theory and at the moving primes here.
\end{remark}

\begin{remark}
For non-monomial upper generators the moving specials are not classified; as in \cite[\S10.3]{Bronstein05}, one searches for special polynomials of low total degree with undetermined coefficients, and the result enters the guess for $s$ with unchanged heuristic status. The type E classification of Proposition~\ref{prop:typeE}, by contrast, does not use the monomial hypothesis and remains exhaustive for arbitrary upper derivations.
\end{remark}

\section{Logands: content, splitting and residues}\label{sec:logands}

This section identifies the possible logands and shows that their divisors are determined by residues of $f$. All statements are unchanged under the base change of Lemma~\ref{lem:basechange}, and we do not distinguish $\cA$ from $\Fbar\cA$ notationally. Recall $\mathcal F=K_0(y)$ and $T=\mathcal F[t_{j+1},\dots,t_n]$.

\begin{lemma}[Gauss valuations and content]\label{lem:gauss}
Let $\fp$ be a height-one prime of $\cO$ and $P=\fp\cA$ the extended prime. Then for $g=\sum_\alpha g_\alpha t^\alpha\in T$ \textup(multi-index $\alpha$ over the upper variables, $g_\alpha\in\mathcal F$\textup),
\[
v_P(g)\;=\;\min_\alpha v_\fp(g_\alpha).
\]
Consequently, for a monic irreducible $p\in T$ with moving prime $P_p$,
\[
\dv_\cA(p)\;=\;P_p\;+\;\sum_\fp\Bigl(\min_\alpha v_\fp(p_\alpha)\Bigr)\,\fp\cA\;=:\;P_p+\gamma(p)\cA ,
\]
where $\gamma(p)\in\Div(\cO)$ is the \emph{content divisor} of $p$; the moving part of $\dv_\cA(p)$ is exactly $P_p$.
\end{lemma}

\begin{proof}
Let $\pi$ be a uniformiser of $\fp$. Localising $\cA$ at $P$ inverts $\cO\setminus\fp$, so $\cO_\fp[t_{j+1},\dots,t_n]\subseteq\cA_P$; conversely every $s\in\cA\setminus P$ has $\min_\alpha v_\fp(s_\alpha)=0$, hence is invertible in the localisation of $\cO_\fp[t_{j+1},\dots,t_n]$ at $(\pi)$, and $\cA_P=\cO_\fp[t_{j+1},\dots,t_n]_{(\pi)}$. In this polynomial ring over the DVR $\cO_\fp$, $\pi$ is prime (the quotient is a polynomial ring over the field $\cO_\fp/\fp$), and $\pi^k\mid g$ iff $\pi^k$ divides every coefficient; hence $v_P(g)=\max\{k:\pi^k\mid g\}=\min_\alpha v_\fp(g_\alpha)$ for $g$ with coefficients in $\cO_\fp$, and for general $g\in T$ by multiplicativity after clearing a denominator from $\cO$. The divisor formula follows: at the moving prime of a \emph{different} monic irreducible, $v(p)=0$; at $P_p$, $v(p)=1$; at extended primes, $v(p)$ is the stated minimum.
\end{proof}

\begin{proposition}[Logand splitting]\label{prop:split}
Let $S$ be a finite set of height-one primes of $\cA$, $S=S_E\sqcup S_M$ its split by type, and let $u\in L^*$ have $\dv_\cA(u)$ supported on $S$. For $P\in S_M$ let $p_P\in T$ be its monic irreducible and $n_P:=v_P(u)$. Then
\[
u\;=\;w\cdot\prod_{P\in S_M}p_P^{\,n_P},\qquad w\in\mathcal F^*=K_0(y)^*,
\]
and
\[
\dv_\cO(w)\;=\;\iota\bigl(\dv_\cA(u)\big|_{S_E}\bigr)\;-\;\sum_{P\in S_M}n_P\,\gamma(p_P),
\]
where $\iota$ contracts extended primes to $\cO$. In particular $w$ is an $S'$-unit of $\cO$ for the finite set $S':=\iota(S_E)\cup\bigcup_{P\in S_M}\supp\gamma(p_P)$.
\end{proposition}

\begin{proof}
Set $w:=u\cdot\prod p_P^{-n_P}$. By Lemma~\ref{lem:gauss} and the choice of the $n_P$, $v_Q(w)=0$ at \emph{every} moving prime $Q$: distinct monic irreducibles do not interact, and at moving $Q\notin S$ both $u$ and the $p_P$ are units. Writing $w=a/b$ with $a,b\in T$ coprime, any irreducible factor of $a$ or of $b$ would produce a moving prime with $v_Q(w)\ne0$; hence $a,b\in\mathcal F$ and $w\in\mathcal F^*$. For $w\in\mathcal F^*$ one has $\dv_\cA(w)=\dv_\cO(w)\cA$ (Lemma~\ref{lem:gauss} applied to constant polynomials), so contracting $\dv_\cA(u)-\sum n_P\dv_\cA(p_P)$ gives the displayed formula.
\end{proof}

\begin{remark}[The class-group obstruction, made explicit]\label{rem:content}
Rescaling $p_P$ by $\lambda\in\mathcal F^*$ shifts $\gamma(p_P)$ by the principal divisor $\dv_\cO(\lambda)$, so the class $[\gamma(p_P)]\in\Cl(\cO)$ is an invariant of the moving prime $P$ --- it is the image of $[P]$ under the inverse of the extension isomorphism $\Cl(\cO)\cong\Cl(\cA)$ of Proposition~\ref{prop:ring}(iii), which the factorisation above realises explicitly. When $\cO$ is a UFD the content can be normalised away and the splitting is exact, as in the transcendental theory; in general the $S$-unit computation on the curve must absorb the contents, and this is precisely where the machinery of Parts I--II enters the mixed setting.
\end{remark}

\begin{theorem}[Residues]\label{thm:residues}
Let $P$ be a normal height-one prime with uniformiser $\pi$ --- a prime of $\Fbar\cA$, or the place at infinity of a hypertangent top variable --- and let $f\in L^*$ with $v_P(f)\ge-\delta_P$. Define the \emph{residue}
\[
\tau_P(f)\;:=\;\Bigl(f\,\frac{\pi}{D\pi}\Bigr)\Big|_P\ \in\ \kappa(P).
\]
Then:
\begin{enumerate}
\item[(i)] $\tau_P(f)$ is independent of the choice of $\pi$, and for any $h\in(\Fbar L)^*$ with $e:=v_P(h)\ne0$,
\[
\tau_P(f)\;=\;e\cdot\Bigl(f\,\frac{h}{Dh}\Bigr)\Big|_P\,;
\]
\item[(ii)] $\tau_P(f)\ne0$ if and only if $v_P(f)=-\delta_P$;
\item[(iii)] if $f$ has an elementary integral over $L$, with $(v,c_i,u_i)$ as in Theorem~\ref{thm:structure}, then
\[
\tau_P(f)\;=\;\sum_{i=1}^{r}c_i\,v_P(u_i)\ \in\ \Fbar :
\]
every residue at a normal prime is a constant.
\end{enumerate}
\end{theorem}

\begin{proof}
$v_P(f\,\pi/D\pi)=v_P(f)+\delta_P\ge0$ by Theorem~\ref{thm:shift}, so the reduction is defined, and it vanishes iff the valuation is positive, which is (ii). Replacing $\pi$ by $w\pi$ ($w$ a $P$-unit) changes $D\pi/\pi$ by $Dw/w$, of valuation $\ge1-\delta_P$; this alters $f\pi/D\pi$ by an element of valuation $\ge1$, leaving the reduction unchanged. For the $h$-formula, write $h=u_0\pi^{e}$: then $Dh/h=e\,D\pi/\pi+Du_0/u_0$ with $v_P(Du_0/u_0)\ge1-\delta_P$, so
\[
\frac{h\,D\pi}{\pi\,Dh}\;=\;\frac{D\pi/\pi}{Dh/h}\;\equiv\;\frac1e\pmod P ,
\]
and $(f\,h/Dh)|_P=\tfrac1e\,\tau_P(f)$.

(iii) By Theorem~\ref{thm:structure}(i), $v_P(v)<0$ would force $v_P(f)\le-\delta_P-1$; hence $v_P(v)\ge0$ and $v_P(Dv)\ge1-\delta_P$. By Lemma~\ref{lem:logder},
\[
f\;=\;Dv\;+\;\rho\,\frac{D\pi}{\pi}\;+\;\varepsilon,\qquad \rho:=\sum_ic_i\,v_P(u_i),\quad v_P(\varepsilon)\ge1-\delta_P .
\]
Multiplying by $\pi/D\pi$ (of valuation $\delta_P$) and reducing mod $P$ kills the first and third terms (valuation $\ge1$) and leaves $\tau_P(f)=\rho$, which lies in $\Fbar$ since $c_i\in\Fbar$ and $v_P(u_i)\in\Z$.
\end{proof}

\begin{corollary}[The logand divisors are determined]\label{cor:divisors}
Suppose $f$ has an elementary integral over $L$, so that $\tau_P(f)\in\Fbar$ at every normal $P$. Let $r_1,\dots,r_l$ be a $\Q$-basis of the $\Q$-span of $\{\tau_P(f)\}_P$, write $\tau_P(f)=\sum_k\rho_{k,P}\,r_k$ with $\rho_{k,P}\in\Q$, and set
\[
\Delta_k\;:=\;\sum_P\rho_{k,P}\,P\ \in\ \Div(\Fbar\cA)\otimes\Q
\]
\textup(a finite sum: $\tau_P(f)\ne0$ only at the normal primes of $S$\textup). Then there are $M\in\Z_{>0}$ and $\tilde u_1,\dots,\tilde u_l\in(\Fbar L)^*$ such that
\begin{enumerate}
\item[(i)] $v_P(\tilde u_k)=M\rho_{k,P}$ at every normal prime $P$, i.e.\ the normal part of $\dv(\tilde u_k)$ is $M\Delta_k$, and
\[
f\;-\;\sum_{k=1}^{l}\frac{r_k}{M}\,\frac{D\tilde u_k}{\tilde u_k}\;=\;Dv\;+\;\sum_h\beta_h\frac{D\nu_h}{\nu_h}
\]
where every $\nu_h$ has divisor supported on the \emph{special} primes and $\beta_h\in\Fbar$;
\item[(ii)] by Proposition~\ref{prop:split}, each such $\nu_h$ --- and each $\tilde u_k$, up to its special part --- factors as an $S'$-unit of $\cO$ times monic irreducibles of $T$; in particular, elementary integrability requires each $M\Delta_k$ to be realisable, its moving part by the monic irreducibles $p_P$ \textup(with their contents\textup) and the remaining $\cO$-divisor by an $S'$-unit of $\cO$ --- Trager's problem of points of finite order, solved for our curves by the machinery of Parts I--II.
\end{enumerate}
\end{corollary}

\begin{proof}
Write $Mc_i=\sum_km_{ik}r_k$ with $m_{ik}\in\Z$ and a common denominator $M$, and set $\tilde u_k:=\prod_iu_i^{m_{ik}}$; then $\sum_ic_iDu_i/u_i=\sum_k(r_k/M)D\tilde u_k/\tilde u_k$, so the displayed identity holds with an \emph{empty} residual sum for this choice. At a normal $P$, Theorem~\ref{thm:residues}(iii) and the $\Q$-independence of the $c_i$ give
\[
\rho_{k,P}\;=\;\frac1M\sum_im_{ik}\,v_P(u_i)\;=\;\frac1M\,v_P(\tilde u_k),
\]
which is (i) for these $\tilde u_k$; their special parts are unconstrained. If instead candidates $\tilde u_k'$ with the same normal divisors $M\Delta_k$ are used (as the algorithm must), the discrepancies $\tilde u_k/\tilde u_k'$ are units at every normal prime, i.e.\ have divisor supported on the special primes: they join the residual sum, whose logands $\nu_h$ are as claimed. Part (ii) is Proposition~\ref{prop:split} applied to the $\nu_h$ and to the realisation of $M\Delta_k$.
\end{proof}

\begin{proposition}[Canonical residue at any pole order]\label{rem:deepres}
Let $P$ be a normal height-one prime (of $\Fbar\cA$, or the hypertangent place at infinity) with uniformiser $\pi$, and let $\lambda_P\in\kappa(P)^*$ be the class of $\pi^{\delta_P-1}D\pi$, which is independent of $\pi$. For every $f\in L$ there is $w\in L$ with $v_P(f-Dw)\ge-\delta_P$, and the class of $f-Dw$ in $\pi^{-\delta_P}\cA_P/\pi^{1-\delta_P}\cA_P$ does not depend on $w$. The \emph{canonical residue}
\[
\hat\tau_P(f)\;:=\;\tau_P(f-Dw)\;=\;\overline{(f-Dw)\,\pi/D\pi}\ \in\ \kappa(P)
\]
is $\Fbar$-linear, vanishes on $DL$, satisfies $\hat\tau_P(Du/u)=v_P(u)$, and equals $\tau_P(f)$ when $v_P(f)\ge-\delta_P$. Consequently, if $f$ has an elementary integral with data $(v,c_i,u_i)$ as in Theorem~\ref{thm:structure}, then $\hat\tau_P(f)=\sum_ic_iv_P(u_i)\in\Fbar$ at every normal prime, whatever the pole order of $f$ there; Corollary~\ref{cor:divisors} holds with $\hat\tau_P$ in place of $\tau_P$ and no hypothesis on $S$.
\end{proposition}

\begin{proof}
\emph{Existence.} If $v_P(f)\ge-\delta_P$ take $w=0$. Otherwise let $v_P(f)=-k-\delta_P$ with $k\ge1$ and let $\mu\in\kappa(P)^*$ be the class of $\pi^{k+\delta_P}f$. For $c\in\cA_P$, $D(c\pi^{-k})=\pi^{-k}Dc-kc\pi^{-k-1}D\pi$; the first term has valuation $\ge-k+1-\delta_P$ (Theorem~\ref{thm:shift}, last clause), the second has valuation $-k-\delta_P$ with leading class $-k\bar c\lambda_P$. Choosing $\bar c=-\mu/(k\lambda_P)$ (possible: $\cA_P\to\kappa(P)$ is surjective and $k\lambda_P\ne0$ in characteristic $0$) gives $v_P(f-D(c\pi^{-k}))>-k-\delta_P$; iterating strictly increases the valuation and terminates.

\emph{Well-definedness.} If $w,w'$ both reduce $f$, then $g:=w-w'$ has $v_P(Dg)\ge-\delta_P$; if $v_P(g)<0$ this contradicts $v_P(Dg)=v_P(g)-\delta_P<-\delta_P$ (Theorem~\ref{thm:shift}), so $g\in\cA_P$ and $v_P(Dg)\ge1-\delta_P$: the two reduced forms agree modulo $\pi^{1-\delta_P}\cA_P$. Independence of $\pi$ is as in Theorem~\ref{thm:residues}(i).

\emph{Properties.} Linearity is clear (reductions add). For $f=Dg$ take $w=g$. For $f=Du/u$, Lemma~\ref{lem:logder} gives $v_P(Du/u)\ge-\delta_P$ already, so $w=0$ and $\hat\tau_P=\tau_P=v_P(u)$. The final claim follows by applying Theorem~\ref{thm:residues}(iii) to $f-Dw$, which has the elementary integral $v-w+\sum c_i\log(u_i)$.
\end{proof}

\begin{remark}
The local Hermite reduction in the proof is the tower analogue of the Rothstein--Trager residue. When $\delta_P=1$ \emph{and} $D$ acts trivially on the coefficients of the completion $\kappa(P)((\pi))$ --- a constant place of a constant curve with $D=d/dx$, the case of \S\ref{ex:chebyshev} --- the reduction can be read off a Laurent expansion, $\hat\tau_P(f)$ being the coefficient of $\pi^{-1}$ divided by $\overline{D\pi}$, and this is what the implementation does at Hermite-order poles. The condition is not decorative: in the completion $D$ takes the form $(D\pi)\,\partial_\pi+\partial$ with $\partial$ the induced derivation of the coefficient field, and the $\pi^{-1}$-coefficient of $D(c\pi^{-k})$ is $\partial c$ when $k=1$, not zero --- for $t=\log(x)$ and $P=(t)$, the Laurent coefficient of $D(g(x)/t)$ at $t^{-1}$ is $g'(x)$. Off the constant-coefficient case --- a deep pole at a place with non-constant coordinates, or at a normal prime of $\delta_P\ge2$ --- the implementation runs the reduction of the proof as written, with $\lambda_P=\overline{\pi^{\delta_P-1}D\pi}$ in place of the leading coefficient, all arithmetic performed in $\kappa(P)$ by reduction modulo $p$ with no substitution of a root; this is the general case of \textsc{CanonicalResidue} (Algorithm~2) below, and it is what lets the transcendental instance of the method reproduce Bronstein's parallel Risch \cite{Bronstein07} verbatim.
\end{remark}

\begin{remark}[Residues are blind to units]\label{rem:blind}
A logand with \emph{empty} divisor --- a unit of $\cO$, such as the Pell unit $x+y$ over $y^2=x^2+1$ --- has $v_P(u)=0$ everywhere and contributes nothing to any residue. Such logands can nevertheless be forced (Example~\ref{ex:unit}), so the candidate list must always include generators of $\cO^*_{S}/F^*$ alongside the residue-determined divisors and the special logands; their coefficients, like those of the special logands, are unknowns of the linear system. This is the parallel-method reflection of the unit-rank computations of Parts I--II.
\end{remark}

\begin{remark}
At moving primes, Theorem~\ref{thm:residues}(iii) is the constancy criterion for moving residues; applied to the tower of Example~\ref{ex:bronsteinE} it reproduces the roots of the resultant $R(z)$ of \cite{Bronstein89} exactly, and the $h$-formula of (i) recovers, with $e=e_P$, the ramification-weighted residues $\tau_p(f)=r_p\,\mathrm{value}_p(f\,P/P')$ of \cite{Bronstein89}. The type E residues feed the $S$-unit and torsion machinery of Part I through Corollary~\ref{cor:divisors}(ii). No residues are computed at special primes: their logands enter the ansatz with unknown coefficients, exactly as in \cite[\S10.3]{Bronstein05}.
\end{remark}

\begin{remark}[Structure of the tower]
Multiplicative dependencies among candidate logands over $L$ are controlled by $\cO_S^*/F^*$, which Part I computes for $q$ univariate: e.g.\ $\log(x-\sqrt{x^2+1})=-\log(x+\sqrt{x^2+1})$ arises from the relation $(x+y)(x-y)=-1$ in $\cO^*$. A self-contained structure theorem for these towers is left open; assumption (T4) is taken as given.
\end{remark}

\section{Degree bounds}\label{sec:degree}

\subsection{The top variable}\label{sec:degreetop}

Let $t:=t_n$ and $E:=K_0(y)(t_{j+1},\dots,t_{n-1})$, so that $L=E(t)$, and let $v_\infty$ be the degree valuation of $L$ over $E$: $v_\infty(a/b)=\deg_tb-\deg_ta$ for $a,b\in E[t]$, with uniformiser $\pi:=1/t$, valuation ring $\cV_\infty=E[\pi]_{(\pi)}$ and residue field $E$. Assume that $t$ is a monomial over $E$: primitive ($Dt=w\in E$), hyperexponential ($Dt=wt$) or hypertangent ($Dt=w(1+t^2)$), with $w\in E^*$. Two observations frame the section. First, the radical is invisible at $v_\infty$: it lies inside the coefficient field $E$, so $L/E$ is a rational function field with a \emph{single} place over $t=\infty$, and no ramification correction arises; the ramification of the curve at infinity concerns bounds in the curve variables, not the top one (Remark~\ref{rem:curvevars}). Second, since $\cO\subset E$, the basis elements $w_i$ are $v_\infty$-units, so for $f=\sum_i(a_i/d)w_i$,
\begin{equation}\label{eq:vinff}
v_\infty(f)\ \ge\ \deg_td-\deg_ta,\qquad \deg_ta:=\max_i\deg_ta_i\,.
\end{equation}

\begin{lemma}\label{lem:atinfty}
$\eta_{v_\infty}=0$ in all three cases, and:
\begin{enumerate}
\item[(a)] if $t$ is primitive, then $v_\infty$ is special, $v_\infty(D\pi/\pi)=1$, and $v_\infty(Dg)\ge v_\infty(g)$ for all $g\in L^*$;
\item[(b)] if $t$ is hyperexponential, then $v_\infty$ is special, $D\pi/\pi=-w$, and $v_\infty(Dg)\ge v_\infty(g)$ for all $g\in L^*$;
\item[(c)] if $t$ is hypertangent, then $v_\infty$ is normal with $\delta_{v_\infty}=1$.
\end{enumerate}
In cases \textup{(a)} and \textup{(b)}, moreover, $v_\infty(Du/u)\ge0$ for every $u\in(\Fbar L)^*$.
\end{lemma}

\begin{proof}
From $D\pi=-\pi^2Dt$ one computes $D\pi=-w\pi^2$, $-w\pi$, $-w(1+\pi^2)$ in the three cases, of $v_\infty$ equal to $2,1,0$ respectively. For $e\in E$, $De\in E$ has $v_\infty\ge0$, and $\cV_\infty$ is the localisation of $E[\pi]$ at $(\pi)$, so the generator computation of Lemma~\ref{lem:etagen} gives $\eta_{v_\infty}=\max(0,-v_\infty(D\pi))=0$. Hence $\Dbar_{v_\infty}=D$ and Definition~\ref{def:normal} reads off $v_\infty(D\pi)$: special, special, and normal with $\delta_{v_\infty}=1+\eta_{v_\infty}=1$. In cases (a),(b), for $g\in L^*$ write $g=u\pi^k$, $u$ a $v_\infty$-unit: $Dg/g=k\,D\pi/\pi+Du/u$ with $v_\infty(D\pi/\pi)\ge0$ and $v_\infty(Du/u)=v_\infty(Du)\ge-\eta_{v_\infty}=0$, so $v_\infty(Dg)\ge v_\infty(g)$; the same identity gives $v_\infty(Du/u)\ge0$ for arbitrary $u\in(\Fbar L)^*$ (the base change of Lemma~\ref{lem:basechange} applies at $v_\infty$ with the same proof).
\end{proof}

\begin{theorem}[Top-variable bounds]\label{thm:top}
Let $f\in L^*$ have an elementary integral over $L$, and let $v,c_i,u_i$ be as in Theorem~\ref{thm:structure}. Then:
\begin{enumerate}
\item[(a)] if $t$ is primitive:\quad $v_\infty(v)\ \ge\ \min\bigl(0,\,v_\infty(f)\bigr)-1$;
\item[(b)] if $t$ is hyperexponential:\quad $v_\infty(v)\ \ge\ \min\bigl(0,\,v_\infty(f)\bigr)$;
\item[(c)] if $t$ is hypertangent:\quad $v_\infty(v)\ \ge\ \min\bigl(0,\,v_\infty(f)+1\bigr)$, and $v_\infty(u_i)=0$ for every $i$ unless $v_\infty(f)\le-1$.
\end{enumerate}
\end{theorem}

\begin{proof}
Write $\Sigma:=\sum_ic_iDu_i/u_i$, so that $f=Dv+\Sigma$; in cases (a),(b), $v_\infty(\Sigma)\ge0$ by the last clause of Lemma~\ref{lem:atinfty}.

(a) We claim: $k:=v_\infty(v)\le-1$ implies $v_\infty(Dv)\le k+1$. Write $v=\sum_{i=0}^{N}v_it^i+v'$ with $N=-k$, $v_N\ne0$, $v_i\in E$ and $v_\infty(v')\ge1$. Then
\[
Dv\;=\;\sum_{i=0}^{N}\bigl(Dv_i+(i+1)\,w\,v_{i+1}\bigr)t^i\;+\;Dv',\qquad v_{N+1}:=0,
\]
and $v_\infty(Dv')\ge v_\infty(v')\ge1$ by Lemma~\ref{lem:atinfty}(a), so $Dv'$ contributes only to powers $t^i$ with $i\le-1$; as $N-1\ge0$, the coefficients of $t^{N}$ and $t^{N-1}$ in $Dv$ are $Dv_N$ and $Dv_{N-1}+Nwv_N$. Suppose both vanish. Then $v_N\in\Const_D(E)\subseteq F^*$, hence $D(v_{N-1}+Nv_Nt)=Dv_{N-1}+Nv_Nw=0$, so $v_{N-1}+Nv_Nt\in\Const_D(L)=F$ and $t\in E$ --- contradicting transcendence, since $Nv_N\ne0$ in characteristic $0$. This proves the claim. Now if $k\le-2$, then $v_\infty(Dv)\le k+1\le-1<0\le v_\infty(\Sigma)$, so $v_\infty(f)=v_\infty(Dv)\le k+1$ and $v_\infty(v)=k\ge v_\infty(f)-1$; if $k\ge-1$ the bound holds trivially. Both cases give (a).

(b) Here $D$ acts diagonally on coefficients: $D(v_it^i)=(Dv_i+iwv_i)\,t^i$. Write $v=\sum_{i=0}^{N}v_it^i+v'$ as in (a). If $k=v_\infty(v)\le-1$, the coefficient of $t^N$ in $Dv$ is $Dv_N+Nwv_N$; were it zero, $D(v_Nt^N)=0$ would give $v_Nt^N\in\Const_D(L)=F^*$, making $t$ algebraic over $E$ --- contradiction. Hence $v_\infty(Dv)=k<0\le v_\infty(\Sigma)$, so $v_\infty(f)=k$, which is (b).

(c) By Lemma~\ref{lem:atinfty}(c), $v_\infty$ is a normal discrete valuation of $\Fbar L$ with $\delta=1$ and finite pole order; as noted at the end of Section~\ref{sec:structure}, Lemma~\ref{lem:logder} and the proof of Theorem~\ref{thm:structure} apply verbatim at such a valuation, yielding $v_\infty(v)\ge\min(0,v_\infty(f)+1)$ and $\sum_ic_iv_\infty(u_i)=0$ --- hence, by $\Q$-independence, $v_\infty(u_i)=0$ for all $i$ --- whenever $v_\infty(f)\ge1-\delta=0>-1$.
\end{proof}

\begin{corollary}[The classical guess, proved at the top]\label{cor:topdeg}
With $f=\sum_i(a_i/d)w_i$ and the denominator of Corollary~\ref{cor:element}, the numerator $b=v\cdot s\prod_ld_l^{\,l-1}\in\cA$ satisfies
\[
\deg_{t_n}(b)\ \le\ \varepsilon\;+\;\max\bigl(\deg_{t_n}a,\ \deg_{t_n}d\bigr)\;+\;\deg_{t_n}(s),
\qquad
\varepsilon=\begin{cases}1,&t_n\ \text{primitive},\\[2pt]0,&t_n\ \text{hyperexponential or hypertangent}.\end{cases}
\]
If $t_n$ is primitive, $s$ may be taken free of $t_n$ \textup(Theorem~\ref{thm:specials}\textup), and the bound is exactly the guess \textup{(10.4)} of \cite[\S10.3]{Bronstein05}, now proved in the top variable.
\end{corollary}

\begin{proof}
$\deg_{t_n}(b)=-v_\infty(b)=-v_\infty(v)+\deg_{t_n}\bigl(s\prod_ld_l^{\,l-1}\bigr)\le-v_\infty(v)+\deg_{t_n}(s)+\deg_{t_n}(d)$, since $\deg_{t_n}\prod_ld_l^{\,l-1}\le\deg_{t_n}d_n\le\deg_{t_n}d$. By Theorem~\ref{thm:top}, $-v_\infty(v)\le\max(0,-v_\infty(f))+\varepsilon$, and $-v_\infty(f)\le\deg_{t_n}a-\deg_{t_n}d$ by \eqref{eq:vinff}; finally $\max(0,\deg_{t_n}a-\deg_{t_n}d)+\deg_{t_n}d=\max(\deg_{t_n}a,\deg_{t_n}d)$. For primitive $t_n$, no special prime of Theorem~\ref{thm:specials} involves $t_n$: type E primes contract to $R_0$; $h_u$ is free of $t_n$ because $Dt_n\in E$ and the lower $Dt_i$ do not involve $t_n$; and $t_n$ is neither hyperexponential nor hypertangent.
\end{proof}

\begin{remark}[Sharpness]
All three bounds are attained. Over $y^2=x^2+1$ with $t=\log(x+y)$: $\int\bigl(t/y\bigr)dx=\tfrac12t^2$ attains \textup{(a)} ($v_\infty(f)=-1$, $v_\infty(v)=-2$). With $t=e^{y}$: $\int(x/y)\,t\,dx=t$ attains \textup{(b)}. Over $E=\Q(x)$ with $t=\tan(x)$: $\int t^2\,dx=t-x$ attains \textup{(c)} ($v_\infty(f)=-2$, $v_\infty(v)=-1$).
\end{remark}

\begin{remark}[What remains heuristic]\label{rem:curvevars}
(i) The anticipated ramification correction at $t_n=\infty$ is vacuous: the curve lies in the coefficient field, and $L/E$ has a single place at infinity. Ramification of the curve at infinity enters only for bounds in the \emph{curve} variable, which Section~\ref{sec:lower} treats place by place. (ii) Bounds in the intermediate variables $t_i$, $j<i<n$, and (iii) the exponents of the special primes --- among them the negative powers of a hyperexponential $t_n$, which are poles at the finite special prime $(t_n)$ --- are the subject of Section~\ref{sec:lower}: they are proved at every place where Proposition~\ref{prop:kernelfree} applies, and the lower-variable counterexamples of Part I, \S6.3 are exactly places where it does not (Example~\ref{ex:chainpreview}). (iv) Together with the top-variable theorem of Part I (radical topmost), the top-variable bound is now proved whichever kind of generator --- radical or transcendental monomial --- is topmost.
\end{remark}

\subsection{Places below the top: the leading part of the derivation}\label{sec:lower}

Theorem~\ref{thm:shift} bounds the pole of $v$ at a normal prime by the pole of $f$, and Theorem~\ref{thm:top} bounds the degree in the top variable. The exponents $e_\sigma$ of the special primes (Step~10 of Algorithm~4) and the degrees in the variables below the top (Step~15) were left, until now, to the guess of \cite[\S10.3]{Bronstein05}; Remark~\ref{rem:curvevars} records them as heuristic. This subsection treats both at once. At a normal prime $D$ lowers every nonzero valuation by the same amount $\delta_P$, which is why no cancellation analysis was needed there. At a special prime, or at a place at infinity, $D$ shifts valuations by \emph{at least} some amount, with equality for most elements but not for all, and the exceptions are precisely what a bound must control. We isolate the homogeneous component of $D$ of lowest weight --- its \emph{leading part} --- and prove that the pole of the integral exceeds the pole of the integrand by at most the shift of the uniformiser whenever the leading part has no kernel in negative degrees, or, when the uniformiser does not attain the minimal shift, whenever its first nonvanishing coefficient is not a derivative in the residue field. Both hypotheses are decided in the towers of this paper (Proposition~\ref{prop:kernelfree}), the top-variable theorem is the case of the top place (Remark~\ref{rem:topcase}), their failure is exactly the obstruction of Part I, \S6.3, and Example~\ref{ex:chainpreview} shows that when they fail the classical guess is genuinely wrong, not merely unproved.

\subsubsection*{Shift and leading part}

Throughout, $v$ is a discrete valuation of $\Fbar L$, normalised so that $v((\Fbar L)^*)=\Z$, with valuation ring $\cV_v$, a uniformiser $\pi$, and residue field $\kappa_v$; bars denote residue classes. The derivation is \emph{$v$-bounded} if the \emph{shift}
\[
\sigma_v(h)\ :=\ v(Dh)-v(h),\qquad h\in(\Fbar L)^*,\ Dh\neq0,
\]
is bounded below; we then write $s_v:=\min_h\sigma_v(h)$, the \emph{shift of $D$ at $v$}. At a height-one prime $P$ of $\Fbar\cA$ one has $s_v\ge-\delta_P$: for $h=u\pi^k$ with $u\in\cA_P^*$, $Dh=ku\pi^{k-1}D\pi+\pi^kDu$ and both $v_P(D\pi)$ and $v_P(Du)$ are at least $-\eta_P$ by Lemma~\ref{lem:etagen}. At normal primes $s_v=-\delta_P$ and every $h$ with $v_P(h)\neq0$ attains it (Theorem~\ref{thm:shift}); at special primes and at the places at infinity the shift is attained by some elements and exceeded by others.

\begin{lemma}[Leading part]\label{lem:leading}
Let $D$ be $v$-bounded with shift $s=s_v$. Then $D(\pi^a\cV_v)\subseteq\pi^{a+s}\cV_v$ for every $a\in\Z$, and the rules
\[
\partial_v\bar c\ :=\ \overline{\pi^{-s}Dc}\quad(c\in\cV_v),\qquad
\rho_v\ :=\ \overline{\pi^{-s-1}D\pi}
\]
define a derivation $\partial_v$ of $\kappa_v$ vanishing on $\Fbar$, and an element $\rho_v\in\kappa_v$. For $c\in\cV_v^*$ and $a\in\Z$,
\begin{equation}\label{eq:leading}
D(c\pi^a)\ \equiv\ \bigl(\partial_v\bar c+a\rho_v\,\bar c\bigr)\,\pi^{a+s}\pmod{\pi^{a+s+1}\cV_v}.
\end{equation}
Consequently $v(Dh)=v(h)+s$ if the class $\partial_v\bar c+a\rho_v\bar c$ of the leading form $\bar c\,\pi^a$ of $h$ is nonzero, and $v(Dh)\ge v(h)+s+1$ otherwise. Moreover $\rho_v\ne0$ if and only if $\sigma_v(\pi)=s$.
\end{lemma}

\begin{proof}
For $u\in\cV_v\setminus\{0\}$, $v(Du)\ge v(u)+s\ge s$, and $v(D(\pi^au))\ge a+s$ by the same inequality; this is the inclusion. If $c\in\pi\cV_v$ then $Dc\in\pi^{1+s}\cV_v$, so $\partial_v$ is well defined on $\kappa_v$; it is additive and satisfies the Leibniz rule because $D$ does, and it kills $\Fbar$ because $D$ does. Since $v(D\pi)\ge1+s$, $\rho_v$ is well defined, and it is nonzero exactly when $v(D\pi)=1+s$, i.e.\ $\sigma_v(\pi)=s$. Finally $D(c\pi^a)=\pi^aDc+ac\pi^{a-1}D\pi$, and reducing modulo $\pi^{a+s+1}$ gives \eqref{eq:leading}.
\end{proof}

The map $\bar c\pi^a\mapsto(\partial_v\bar c+a\rho_v\bar c)\pi^{a+s}$ on the graded ring $\bigoplus_a\pi^a\cV_v/\pi^{a+1}\cV_v\cong\kappa_v[\bar\pi,\bar\pi^{-1}]$ is the \emph{leading part} of $D$ at $v$. Its kernel in degree $a$ is nontrivial if and only if $-a\rho_v$ is a logarithmic derivative in $(\kappa_v,\partial_v)$, and we call
\[
N_v\ :=\ \bigl\{a\in\Z:\ \partial_v\bar c=-a\rho_v\,\bar c\ \text{for some }\bar c\in\kappa_v^*\bigr\}
\]
the \emph{kernel degrees} of $v$. Always $0\in N_v$; if $\rho_v=0$ then $N_v=\Z$.

\begin{theorem}[Bounds at $v$-bounded places]\label{thm:lower}
Let $D$ be $v$-bounded with shift $s=s_v$, let $r:=\sigma_v(\pi)-s\ge0$, and let $f\in L$ have an elementary integral, written as in \eqref{eq:liou} with $v\in\Fbar L$ its rational part \textup(denoted $g$ in this theorem, to avoid a clash with the valuation\textup). Suppose one of the following holds:
\begin{enumerate}
\item[(i)] $r=0$ and $N_v\cap\Z_{<0}=\emptyset$;
\item[(ii)] $r\ge1$, $\Const_{\partial_v}(\kappa_v)=\Fbar$, and $\rho_v^{(r)}:=\overline{\pi^{-s-1-r}D\pi}\in\kappa_v^*$ is not of the form $\partial_v\bar c$ with $\bar c\in\kappa_v$.
\end{enumerate}
Then
\[
v(g)\ \ge\ \min\bigl(0,\ v(f)-s\bigr)-r\ =\ \min\bigl(-r,\ v(f)-\sigma_v(\pi)\bigr).
\]
\end{theorem}

\begin{proof}
Write $\Sigma:=\sum_ic_iDu_i/u_i$, so $f=Dg+\Sigma$, and note $v(\Sigma)\ge s$, since $v(Du_i)\ge v(u_i)+s$. Let $a:=v(g)$ and suppose $a<0$.

(i) Here $\sigma_v(\pi)=s$ and $\rho_v\ne0$. By Lemma~\ref{lem:leading}, $v(Dg)=a+s$ unless $a\in N_v$; as $a<0$ and $N_v\cap\Z_{<0}=\emptyset$, $v(Dg)=a+s<s\le v(\Sigma)$, hence $v(f)=a+s$ and $a=v(f)-s$.

(ii) Here $\rho_v=0$, so the kernel of the leading part in degree $a$ consists of the classes $\bar c\pi^a$ with $\partial_v\bar c=0$, i.e.\ $\bar c\in\Fbar$; and $\rho_v^{(r)}$ is the first nonvanishing coefficient of $D\pi$, $v(D\pi)=s+1+r$. We claim $v(Dg)\le a+s+r$. Suppose not. Then the leading form of $g$ lies in the kernel: $g=c\pi^a+g_1$ with $c\in\Fbar^*$ (a constant, lifted to the constant $c\in\Fbar\subset\Fbar L$) and $v(g_1)\ge a+1$. Now $D(c\pi^a)=ac\pi^{a-1}D\pi$ has valuation exactly $a+s+r$, so $Dg_1=Dg-D(c\pi^a)$ has all its coefficients of orders $a+s+1,\dots,a+s+r-1$ equal to zero and its coefficient of order $a+s+r$ equal to $-ac\rho_v^{(r)}\ne0$. Inductively, as long as $v(g_i)=a+i$ with $1\le i\le r-1$: the class of $Dg_i$ in order $a+i+s$ is $\partial_v\bar c_i$ for the leading coefficient $\bar c_i$ of $g_i$ (Lemma~\ref{lem:leading} with $\rho_v=0$), and it must vanish, so $c_i\in\Fbar$; put $g_{i+1}:=g_i-c_i\pi^{a+i}$, whose derivative differs from $Dg_i$ by $(a+i)c_i\pi^{a+i-1}D\pi$, of valuation $a+i+s+r>a+s+r$. We arrive at $g_r$ with $v(g_r)\ge a+r$ and with the same coefficient $-ac\rho_v^{(r)}$ in order $a+s+r$ as $Dg_1$. If $v(g_r)>a+r$ then $v(Dg_r)\ge a+r+1+s$, a contradiction; so $v(g_r)=a+r$, and by Lemma~\ref{lem:leading} the coefficient of $Dg_r$ in order $a+s+r$ is $\partial_v\bar c_r$ for the leading coefficient $\bar c_r$ of $g_r$. Hence $\partial_v\bar c_r=-ac\rho_v^{(r)}$, and $\rho_v^{(r)}=\partial_v(-\bar c_r/(ac))$, contradicting (ii). This proves the claim.

Now if $a<-r$ then $v(Dg)\le a+s+r<s\le v(\Sigma)$, so $v(f)=v(Dg)\le a+s+r$ and $a\ge v(f)-s-r$. If $-r\le a<0$ the bound $a\ge-r$ holds trivially. In both cases $a\ge\min(-r,v(f)-s-r)$, and for $a\ge0$ there is nothing to prove.
\end{proof}

\begin{remark}[The top place]\label{rem:topcase}
Theorem~\ref{thm:top} is Theorem~\ref{thm:lower} at $v=v_\infty$ of Section~\ref{sec:degree}, with $\kappa_v=E$ and $\partial_v=D_E$. For $t$ primitive, $s=0$, $D\pi=-w\pi^2$ gives $\rho_v=0$, $r=1$ and $\rho_v^{(1)}=-w$; hypothesis (ii) reads $w\notin D_EE$, which is the transcendence of $t$, and the conclusion is $v_\infty(v)\ge\min(0,v_\infty(f))-1$. For $t$ hyperexponential, $s=0$, $\rho_v=-w\ne0$, $r=0$, and $a\in N_v$ would give $c\in E^*$ with $D_Ec/c=aw$, i.e.\ $D(ct^{-a})=0$, impossible for $a\ne0$; the conclusion is $v_\infty(v)\ge\min(0,v_\infty(f))$. For $t$ hypertangent, $s=-1$, $\partial_v=0$ on $E$, $\rho_v=-w\ne0$, $N_v=\{0\}$, and the conclusion is $v_\infty(v)\ge\min(0,v_\infty(f)+1)$. The three cancellation arguments of the proof of Theorem~\ref{thm:top} are the three shapes of the leading part.
\end{remark}

\begin{corollary}[Exponents and degrees]\label{cor:lowerbounds}
In the situation of Theorem~\ref{thm:lower}, write $\varepsilon_v:=\max\bigl(0,\,s_v-v(f)\bigr)+r$, so that the theorem reads $v(g)\ge-\varepsilon_v$ whenever $v(g)<0$.
\begin{enumerate}
\item[(a)] If $v=v_P$ for a special height-one prime $P$ over the irreducible $\sigma\in R$ with $e_P=v_P(\sigma)$, and every prime over $\sigma$ satisfies a hypothesis of the theorem, then the exponent of $\sigma$ in $\den(g)$ is at most $\max_{P\mid\sigma}\lceil\varepsilon_{v_P}/e_P\rceil$.
\item[(b)] If $v$ ranges over the places $P$ of $\Fbar L$ over $t_i=\infty$ described below, $e_P:=-v_P(t_i)$, and every one satisfies a hypothesis of the theorem, then for $g=b/(sD_v)$ with $b\in\Fbar\cA$ as in Step~16 of Algorithm~4,
\[
\deg_{t_i}(b)\ \le\ \deg_{t_i}(sD_v)+\max_{P}\bigl\lceil\varepsilon_{v_P}/e_P\bigr\rceil .
\]
\end{enumerate}
\end{corollary}

\begin{proof}
(a) is the theorem at each $P\mid\sigma$, converted to the exponent of $\sigma$ through $v_P(\sigma)=e_P$. (b) At each $P$, $v_P(b)=v_P(g)+v_P(sD_v)\ge-\varepsilon_{v_P}-e_P\deg_{t_i}(sD_v)$, since $sD_v$ is a polynomial in $t_i$ over a field on which $v_P$ is trivial; and $\deg_{t_i}(b)\le\max_P\lceil-v_P(b)/e_P\rceil$ for $b\in\Fbar\cA$, because the coordinates $a_i$ of $b=\sum_ia_iw_i$ are recovered from the values of $b$ on the $m$ sheets over the generic point, and $\deg_{t_i}$ of a polynomial is $-v_P/e_P$ at any place over $t_i=\infty$ where it does not cancel.
\end{proof}

\subsubsection*{The places and their data}

The places at which the algorithm needs bounds are of four kinds; in each of them $\Fbar L$ has a \emph{coefficient field} $C\subset\Fbar L$ on which $v$ is trivial and which maps isomorphically onto $\kappa_v$, so that the shift, $\partial_v$ and $\rho_v$ are read off from the generators.
\begin{enumerate}
\item[(P$_\infty$)] $j=1$, and $v$ is the Gauss extension to $\Fbar L$ of a place $P$ of the curve $\Fbar K_0(y)=\Fbar(t_1)(y)$ over $t_1=\infty$: $v(\sum_\alpha a_\alpha t^\alpha)=\min_\alpha v_P(a_\alpha)$ for $a_\alpha\in\Fbar K_0(y)$ and monomials $t^\alpha$ in $t_2,\dots,t_n$. Then $\kappa_v=\Fbar(\bar t_2,\dots,\bar t_n)$ and $C=\Fbar(t_2,\dots,t_n)$; $e_P=-v_P(t_1)\in\{1,e_\infty\}$ as in Part I, \S6.1.
\item[(P$_\sigma$)] $j=1$, and $v=v_P$ for a special type E prime $P$ over $\sigma\in\Fbar[t_1]$, i.e.\ the Gauss extension of a place of the curve over $\sigma$; $\kappa_v$ and $C$ as in (P$_\infty$), $e_P\in\{1,e_P\}$ as in Lemma~\ref{lem:lowershift}.
\item[($\infty_i$)] $v$ is the degree valuation in an upper variable $t_i$, $i>j$, over the field generated by $y$ and the other generators,
\[
C_i:=\Fbar K_0(y)(t_{j+1},\dots,\widehat{t_i},\dots,t_n),
\]
with $\pi=1/t_i$ and $\kappa_v\cong C=C_i$.
\item[(M)] $v=v_P$ for a moving special prime $P$ whose monic irreducible $p\in T$ is linear in some upper variable $t_i$ with coefficients free of $t_i$: the primes $(t_i)$ and $(t_i\mp\sqrt{-1})$ of Theorem~\ref{thm:specials}(iii)--(iv), and the arguments $p$ of logarithms $t_k=\log(p)$ that are linear in their main variable. Then $\pi=p$ and $\kappa_v\cong C=C_i$ via $t_i\mapsto$ the root of $p$.
\end{enumerate}
In each case $\Fbar L$ is a finite extension of $C(\pi)$ and embeds in $C((\pi))$, and $D$ extends to $C((\pi))$ termwise; the extension is a derivation which agrees with $D$ on $C(\pi)$ and, in cases (P$_\infty$), (P$_\sigma$), on $y$ (differentiate $y^2=q$), hence on $\Fbar L$. This gives:

\begin{lemma}[Shift from the generators]\label{lem:shiftgen}
At a place of one of the four kinds, with $c_1,\dots,c_l$ the generators of $C$ listed in the case description \textup(the $t_k$, $k\neq i$, and $y$ when $y\in C$\textup),
\[
s_v\ =\ \min\bigl(\sigma_v(\pi),\ v(Dc_1),\dots,v(Dc_l)\bigr),
\]
$\partial_v\bar c_k$ is the coefficient of $\pi^{s_v}$ in the expansion of $Dc_k$, and $\rho_v^{(r)}$ is the coefficient of $\pi^{s_v+1+r}$ in the expansion of $D\pi$ \textup($\rho_v=\rho_v^{(0)}$\textup). In particular $D$ is $v$-bounded.
\end{lemma}

\begin{proof}
For $c\in C^*$, $Dc=\sum_k(\partial c/\partial c_k)Dc_k$ with $\partial c/\partial c_k\in C$ of valuation $\ge0$, so $v(Dc)\ge\min_kv(Dc_k)$; for $h=\sum_bc_b\pi^b$ in $C((\pi))$, $D(c_b\pi^b)=\pi^bDc_b+bc_b\pi^{b-1}D\pi$ has valuation at least $b+\min(\min_kv(Dc_k),\sigma_v(\pi))$, and the termwise extension is continuous. Hence $\sigma_v(h)\ge\min(\sigma_v(\pi),\min_kv(Dc_k))$ for every $h$, and the generators themselves attain the right-hand side. The descriptions of $\partial_v$ and $\rho_v^{(r)}$ are the definitions read in $C((\pi))$.
\end{proof}

We say that a generator $c_k$ of $C$ \emph{attains the shift} if $v(Dc_k)=s_v$, equivalently if $\partial_v\bar c_k\ne0$; the non-attaining generators are constants of $\partial_v$. Recall that an upper generator $t_k$ is \emph{primitive} if $Dt_k$ lies in the field below $t_k$, \emph{hyperexponential} if $Dt_k/t_k$ does, \emph{hypertangent} if $Dt_k/(1+t_k^2)$ does, and call $t_k$ \emph{independent} if no $Dt_{k'}$, $k'\ne k$, involves $t_k$ (nor does $q$, automatically). Set $\lambda_k:=\partial_v\bar t_k$.

\begin{proposition}[Deciding the hypotheses]\label{prop:kernelfree}
Let $v$ be a place of one of the four kinds, $s=s_v$, $r=\sigma_v(\pi)-s$.
\begin{enumerate}
\item[(K1)] \textup{(Constant coefficients.)} Suppose $r=0$, and that every attaining generator of $C$ other than those defining the place \textup($t_1$, $y$ in \textup{(P$_\infty$)}, \textup{(P$_\sigma$)}; none in \textup{($\infty_i$)}, \textup{(M)}\textup) is an upper generator $t_k$ that is either primitive with $\lambda_k\in\Const_{\partial_v}(\kappa_v)$ or hyperexponential with $\lambda_k=\mu_k\bar t_k$, $\mu_k\in\Const_{\partial_v}(\kappa_v)^*$, and that $\rho_v\in\Const_{\partial_v}(\kappa_v)$ is not a $\Q$-linear combination of the $\mu_k$ of the attaining hyperexponential generators \textup(with none of them, this says $\rho_v\ne0$, which is $r=0$\textup). Then $N_v=\{0\}$, and Theorem~\ref{thm:lower}(i) applies. The hypotheses on the coefficients hold in particular when $\rho_v$, the $\lambda_k$ and the $\mu_k$ lie in $\Fbar(\bar t_{k'}:\ t_{k'}\text{ non-attaining})$.
\item[(K2)] \textup{(Transcendence.)} Suppose $v$ is of kind \textup{($\infty_i$)} or \textup{(M)} for an independent monomial $t_i$ \textup(with $p=t_i$ if $t_i$ is hyperexponential, $p=t_i\mp\sqrt{-1}$ if hypertangent\textup). Then Theorem~\ref{thm:lower} applies: by (i) with $N_v=\{0\}$ if $t_i$ is hyperexponential or hypertangent, and by (ii) with $r=1$ if $t_i$ is primitive.
\item[(K3)] \textup{(Residues.)} Suppose $r\ge1$, $n=1$ with $Dt_1=1$ and no radical \textup(so $v$ is of kind \textup{($\infty_i$)} for a generator above $K_0=F(x)$ that is the only other generator\textup), so that $\kappa_v=\Fbar(x)$ and $\partial_v=d/dx$. Then Theorem~\ref{thm:lower}(ii) applies if and only if the rational function $\rho_v^{(r)}$ has a nonzero residue.
\item[(K4)] \textup{(Exponentials over the curve.)} Suppose $v$ is of kind \textup{(P$_\infty$)} or \textup{(P$_\sigma$)}, $r\ge1$, and every upper generator attains the shift and is hyperexponential or hypertangent with $\lambda_k=\mu_k\bar t_k$, resp.\ $\lambda_k=\mu_k(1+\bar t_k^2)$, $\mu_k\in\Fbar^*$, the $\mu_k$ of the hyperexponential generators being linearly independent over $\Q$. Then Theorem~\ref{thm:lower}(ii) applies.
\end{enumerate}
\end{proposition}

\begin{proof}
(K1) Let $\mathcal C:=\Const_{\partial_v}(\kappa_v)$ and let $A$ be the set of attaining generators; by hypothesis $\kappa_v=\mathcal C(\bar t_k:k\in A)$ with $\partial_v\bar t_k=\lambda_k\in\mathcal C$ for the primitive and $\partial_v\bar t_k=\mu_k\bar t_k$, $\mu_k\in\mathcal C^*$, for the hyperexponential generators. Let $A'\subseteq A$ be a maximal subset with $\{\bar t_k:k\in A'\}$ algebraically independent over $\mathcal C$, and $K':=\mathcal C(\bar t_k:k\in A')$; then $\kappa_v/K'$ is finite, and since $\partial_v$ extends uniquely to the algebraic closure of $K'$ and commutes with its automorphisms over $K'$, the norm and trace of $\kappa_v/K'$ satisfy $\partial_vN(\bar c)/N(\bar c)=\operatorname{Tr}(\partial_v\bar c/\bar c)$. Let $\bar c\in\kappa_v^*$ with $\partial_v\bar c=-a\rho_v\bar c$, $a\ne0$, and $N:=N_{\kappa_v/K'}(\bar c)\in K'^*$; since $-a\rho_v\in\mathcal C\subseteq K'$, $\partial_vN=-a'\rho_vN$ with $a':=a[\kappa_v:K']\ne0$. Write $N=\gamma\prod_lp_l^{e_l}$ with $\gamma\in\mathcal C^*$ and distinct monic irreducibles $p_l$ of $\mathcal C[\bar t_k:k\in A']$; then $\sum_le_l\partial_vp_l/p_l=-a'\rho_v\in\mathcal C$, and clearing denominators, $p_l\mid\partial_vp_l$ for each $l$. Now $\partial_v=\Lambda+\sum_{k\in A'_{\mathrm e}}\mu_k\bar t_k\,\partial/\partial\bar t_k$ with $\Lambda:=\sum_{k\in A'_{\mathrm p}}\lambda_k\,\partial/\partial\bar t_k$, the sums running over the primitive and the hyperexponential members of $A'$, does not raise the total degree, so $\nu_l:=\partial_vp_l/p_l\in\mathcal C$; writing $p_l=\sum_\alpha p_{l,\alpha}\,\bar t_{\mathrm e}^{\,\alpha}$ with $p_{l,\alpha}\in\mathcal C[\bar t_k:k\in A'_{\mathrm p}]$ and comparing the coefficients of the monomials $\bar t_{\mathrm e}^{\,\alpha}$ in $\partial_vp_l=\nu_lp_l$ gives $\Lambda p_{l,\alpha}=(\nu_l-\alpha\cdot\mu)\,p_{l,\alpha}$, and since $\Lambda$, a derivation with constant coefficients, lowers the degree of every nonconstant polynomial and kills the constants, $p_{l,\alpha}\ne0$ forces $\nu_l=\alpha\cdot\mu$. Hence $-a'\rho_v=\sum_le_l\nu_l\in\sum_k\Z\mu_k$, so $\rho_v$ is a $\Q$-linear combination of the $\mu_k$ (with no hyperexponential generator, $\rho_v=0$, i.e.\ $r\ge1$), contradicting the hypothesis. Hence $N_v=\{0\}$. The last sentence holds because the non-attaining generators are constants.

(K2) Since $t_i$ is independent, $C_i$ is closed under $D$, $\Const_D(C_i)=\Fbar$ by (T4), and by Lemma~\ref{lem:shiftgen} $\partial_v$ is the derivation induced by $D$ on $C_i\cong\kappa_v$ shifted by $\pi^{-s}$: at $v_\infty$ of a primitive or hyperexponential $t_i$, and at $(t_i)$, $(t_i\mp\sqrt{-1})$, one has $s=0$ and $\partial_v=D|_{C_i}$; at $v_\infty$ of a hypertangent $t_i$, $s=-1$ and $\partial_v=0$. The three computations of Remark~\ref{rem:topcase} then apply verbatim with $E$ replaced by $C_i$: for $t_i$ primitive, $\rho_v^{(1)}=-w\notin DC_i$ since otherwise $t_i-c$ would be a constant for some $c\in C_i$, contradicting $t_i\notin C_i$ (a transcendence-degree count: each generator above $t_i$ is transcendental over the field below it, which contains $t_i$); for $t_i$ hyperexponential at $v_\infty$ and at $(t_i)$, $\rho_v=\mp w$ and $a\in N_v$ would make $ct_i^{\pm a}$ a constant; for $t_i$ hypertangent at $(t_i\mp\sqrt{-1})$, $\rho_v=\pm2\sqrt{-1}\,w$ and $a\in N_v$ would make $c\,\bigl((t_i-\sqrt{-1})/(t_i+\sqrt{-1})\bigr)^{\pm a}$ a constant; at $v_\infty$ of a hypertangent, $\partial_v=0$ and $\rho_v\ne0$ give $N_v=\{0\}$ directly.

(K3) Here $\Const_{d/dx}(\Fbar(x))=\Fbar$, and a rational function is a derivative in $\Fbar(x)$ if and only if all its residues vanish (its polynomial part and its higher-order pole parts are derivatives, its simple-pole part is a derivative only if it is zero).

(K4) $\kappa_v=\Fbar(\bar t_{j+1},\dots,\bar t_n)$ with $\partial_v\bar t_k=\mu_k\bar t_k$ or $\mu_k(1+\bar t_k^2)$ is a tower of hyperexponential and hypertangent monomials over the constant field $\Fbar$ with derivation $\partial_v$: no new constants arise (for the hyperexponential layers by the $\Q$-independence of the $\mu_k$, by \cite[Thm.~5.1.2]{Bronstein05}; for the hypertangent layers by \cite[\S5.10]{Bronstein05}), so $\Const_{\partial_v}(\kappa_v)=\Fbar$. And no element of such a tower has a nonzero constant derivative: by induction on the number of layers, in the top layer $\bar t$ write $\bar c$ as a Laurent polynomial in $\bar t$ (hyperexponential) or a polynomial in $\bar t$ (hypertangent) plus partial fractions at the other places; $\partial_v$ preserves the order of every pole away from $0,\infty$ (resp.\ away from $\infty$), so those parts vanish, and the monomial (resp.\ polynomial) part has constant derivative only through its degree-zero term, which reduces to the layer below. Since $\rho_v^{(r)}\in\Fbar^*$ (the coefficients of $D\pi$ are constants, $\pi$ and $D\pi$ lying in $\Fbar K_0(y)$), hypothesis (ii) holds.
\end{proof}

\begin{example}[The Euler tower of $\int\arcsin(\sqrt{x+1})/\sqrt x\,dx$]\label{ex:eulerbounds}
The parametrisation $w=\sqrt x+\sqrt{x+1}$ gives the transcendental tower $\Q(i)(w,t)$, $Dw=2w^3/(w^4-1)$, $Dt=-2iw^2/(w^4-1)$, integrand $f=2tw/(w^2-1)$. At $w=\infty$: $\sigma_v(w)=v(Dw)-v(w)=1+1=2$ and $v(Dt)=2$, so $s=2$, $r=0$, $t$ attains and is primitive with $\lambda_t=-2i\in\Fbar$, $\rho_v=-2\in\Fbar$; (K1) gives $v(g)\ge\min(0,v(f)-2)=-1$, since $v(f)=1$. At the special prime $w$ ($w^3\mid Dw$): $\sigma_v(w)=3-1=2$, $v_w(Dt)=2$, so again $s=2$, $r=0$, (K1) applies, and $v_w(g)\ge\min(0,v_w(f)-2)=-1$ since $v_w(f)=1$. Hence $s=w$ in Step~10 with the \emph{proved} exponent $1$, and $\deg_wb\le1+1=2$ in Step~15 by Corollary~\ref{cor:lowerbounds}; with the top-variable bound $\deg_tb\le2$ the ansatz has $9$ coefficients where the guess of \cite[\S10.3]{Bronstein05} (with the exponent of $w$ found by retry) has $24$, and the system solves at once: $g=\bigl(i(w^2+1)+(w^2-1)t\bigr)/w$, both bounds attained. Every bound in force is proved; \S\ref{ex:arcsinb} follows the same integrand with a parameter $b$ in place of $1$, where the system has no solution and the one hypothesis of Proposition~\ref{prop:certificates}(b) that is not met is the completeness of the $S'$-units.
\end{example}

\begin{example}[The obstruction, and a chain]\label{ex:chainpreview}
At $v_\infty$ of a lower logarithm $t_i$ on which a later generator depends, (K2) fails and the hypothesis of Theorem~\ref{thm:lower}(ii) may fail with it. The second example of Part I, \S6.3, is the tower $t_1=\log(x)$, $t_2=\log(xt_1+1)$; at $v=v_\infty$ of $t_1$: $s=0$, $\rho_v=0$, $r=1$, $\rho_v^{(1)}=-1/x$, and $\partial_v\bar t_2=1/x$ is the leading coefficient of $Dt_2=(t_1+1)/(xt_1+1)$, so $\rho_v^{(1)}=\partial_v(-\bar t_2)$ is a derivative in $\kappa_v=\Fbar(x,\bar t_2)$: the second-order cancellation of Part I is the first link of a chain, and each further link costs one more degree in $t_2$ (the solution $c_r$ of $\partial_v\bar c_r=-ac\rho_v^{(1)}$ has degree one in $\bar t_2$, and so on). \S\ref{ex:chain} exhibits an integrand of this tower whose integral needs four links, on which the classical bound --- and the ``$+2$'' of the implementations before this section --- fails. In the same way, at the special prime $(W)$ of the tower $t_1=x$, $t_2=W$ with $DW=2W/(x(1+W))$ --- $W=W(x^2)$ the Lambert function, a non-monomial generator --- one has $s=0$, $r=0$ and $\rho_v=2/x$, the logarithmic derivative of $x^2$ in $\kappa_v=\Fbar(x)$, so $N_v=\Z$ and the exponent of $W$ remains a guess (the retry finds $1$).
\end{example}

\begin{remark}[What is decided and what is not]\label{rem:decided}
The algorithm (Algorithm~6) applies (K1)--(K4) at every place of Step~10 and Step~15 and records whether the bound is proved. On the fifty integrals of the companion paper, thirty-eight are computed with every bound in force proved on every rung of the retry ladder; in the other twelve the undecided place is a special prime over the curve variable that is a branch prime or at which the two coordinates of a derivative have equal valuation --- cases the implementations leave to the guess, though the theory decides them with the expansions of Part I, \S6.1 in place of the Gauss valuation. The hyperexponential clause of (K1) is used once: in $\int x^3e^{\arcsin(x)}/\sqrt{1-x^2}\,dx$ (Problem~10 there: $t_1=\arcsin(x)$, $t_2=e^{t_1}$ over $y^2=1-x^2$), at the place at infinity of the curve variable, $s=1$, $r=0$ and both upper generators attain the shift, with $\lambda_1=\mu=\mp i$ at the two places and $\rho_v=-1$; $\rho_v\notin\Q\mu$, so $N_v=\{a\in\Z:-a\rho_v\in\Z\mu\}=\{0\}$ and the degree in $x$ is $3$, where the classical guess is $5$. The implementations decide the membership $\rho_v\in\sum_k\Q\mu_k$ by linear algebra over $\Q$ in a common number field of $\rho_v$ and the $\mu_k$, and only when these are numbers: at a place at infinity they are read from the expansion of $y$, at a special prime from the classes modulo the prime when those are constants, and in every other case the place is left to the guess. Summed over the rungs, the ans\"atze shrink from $7053$ unknowns to $1966$ and the rungs from $103$ to $74$; the two systems of Example~\ref{ex:eulerbounds} are typical. What remains heuristic is exactly: the kernel degrees at a place where $\rho_v$ or some $\lambda_k$ is not a $\partial_v$-constant (a logarithmic derivative in the residue field, as at $(W)$ above); the second form at a lower primitive on which a later generator depends (the chains); and the second form at a place whose residue field carries constants beyond $\Fbar$ (a non-attaining upper generator together with an attaining hyperexponential one), where the cancellation condition involves the next coefficients of the non-attaining generators. Each of these is a Risch-type decision in $(\kappa_v,\partial_v)$ --- is $-a\rho_v$ a logarithmic derivative, is $\rho_v^{(r)}$ a derivative --- and the reduction systems of \cite{DuRaab25} are the natural tool; for the chains, Example~\ref{ex:chainpreview} shows the slack is bounded by one plus the degree bound in the variable above, which the implementations use as the guess when (K2) fails.
\end{remark}

\section{The algorithm}\label{sec:algorithm}

We now assemble the results of Sections~\ref{sec:ring}--\ref{sec:degree} into an algorithm at the level of detail of \textbf{ParallelIntegrate} in \cite[\S10.3]{Bronstein05}, whose structure --- a once-per-field phase followed by a per-integrand phase --- we retain. Two subalgorithms come first: \textsc{Classify}, which plays the role that \textbf{SplitFactor} plays in the transcendental algorithm (there, normal and special parts are separated by a gcd; here, by the valuation data of Section~\ref{sec:valuations}), and \textsc{Realise}, which has no transcendental counterpart: it converts the residue divisors of Corollary~\ref{cor:divisors} into logands, and is where the $S$-unit, norm and torsion machinery of Parts I--II enters.

\algheader{Algorithm 1}{$\textsc{Classify}(P)$\quad(* prime classification; Definition~\ref{def:normal} on generators *)}
\begin{algorithmic}[1]
\Statex (* Given a height-one prime $P$ of $\cA$, over the irreducible $p$ ($p\in R_0$ for type E, monic $p\in T$ for type M), return $(e_P,\eta_P,\delta_P,\mathrm{special}_P)$. Valuations of elements $\sum_ic_iw_i$ are computed from the coefficients: $v_P=\min_i\bigl(e_P\,v_p(c_i)+v_P(w_i)\bigr)$, with $v_P(w_i)$ from Lemma~\ref{lem:lowershift} at branch primes and $v_P(w_i)=0$ elsewhere. *)
\State $e_P\gets v_P(p)$ \Comment{$=m/\gcd(m,l_0)$ at a branch prime over $Q_{l_0}$; $=1$ otherwise}
\State $\eta_P\gets\max\bigl(0,\ \max_{g\in\mathcal G}(-v_P(Dg))\bigr)$, $\mathcal G=\{t_1,\dots,t_n,w_1,\dots,w_{m-1}\}$ \Comment{Lemma~\ref{lem:etagen}}
\State $\pi\gets p$ \textbf{if} $e_P=1$ \textbf{else} $w_{i^*}$ with $i^*l_0\equiv\gcd(m,l_0)\pmod m$ \Comment{Lemma~\ref{lem:lowershift}}
\State $\mathrm{special}_P\gets\bigl[\eta_P+v_P(D\pi)\ge1\bigr]$ \Comment{Definition~\ref{def:normal}: $v_P(\Dbar_P\pi)\ge1$}
\State \textbf{return} $(e_P,\ \eta_P,\ \delta_P=1+\eta_P,\ \mathrm{special}_P)$
\end{algorithmic}
\algfooter

In closed form, \textsc{Classify} returns the taxonomy of Propositions~\ref{prop:typeE} and \ref{prop:typeM}: over $\Dbar_0$-normal $p\in R_0$, $\eta_P=\max(e_P(1+v_p(\den_0))-1,\mu_P)$ with specialness iff $\mu_P\ge e_P(1+v_p(\den_0))$; over $\Dbar_0$-special $p$, always special; at moving primes, $\eta_P=v_p(h_u)$ with Bronstein's criterion $p\mid\Dbar_up$. The generic form above is what the implementation executes.

\algheader{Algorithm 2}{$\textsc{CanonicalResidue}(P,f)$\quad(* Proposition~\ref{rem:deepres} *)}
\begin{algorithmic}[1]
\Statex (* Given a normal height-one prime $P$ with uniformiser $\pi$ and $f\in L$, return $\hat\tau_P(f)\in\kappa(P)$. *)
\State $k\gets-v_P(f)-\delta_P$;\quad \textbf{if} $k<0$ \textbf{return} $0$ \Comment{Thm.~\ref{thm:residues}(ii)}
\State $\lambda_P\gets\overline{\pi^{\delta_P-1}D\pi}\in\kappa(P)^*$
\State \textbf{while} $k\ge1$: $\mu\gets\overline{\pi^{k+\delta_P}f}$;\ $c\gets$ a lift to $\cA_P$ of $-\mu/(k\lambda_P)$;\ $f\gets f-D(c\pi^{-k})$;\ $k\gets-v_P(f)-\delta_P$ \Comment{each pass raises $v_P(f)$}
\State \textbf{return} $\overline{f\pi/D\pi}$ \quad(* $=e_P\,(f\,p/Dp)|_P$ when $P$ is the only prime over $p$ with $v_P(p)=e_P$; Thm.~\ref{thm:residues}(i) *)
\Statex (* When $D$ acts trivially on the coefficients of $\kappa(P)((\pi))$ --- a constant place with $D=d/dx$ --- the loop is a Laurent expansion and the result is the $\pi^{-1}$-coefficient divided by $\overline{D\pi}$; in general the coefficients themselves are differentiated and the loop must be run as written. *)
\end{algorithmic}
\algfooter

\algheader{Algorithm 3}{$\textsc{Realise}(\Delta)$\quad(* residue divisor to logands; Corollary~\ref{cor:divisors} *)}
\begin{algorithmic}[1]
\Statex (* Given $\Delta=\sum_P\rho_PP\in\Div(\Fbar\cA)\otimes\Q$ supported on normal primes, return a list of pairs $(\gamma_j,u_j)$, $\gamma_j\in\Q$, $u_j\in(\Fbar L)^*$, with $\sum_j\gamma_j\dv(u_j)\big|_{\mathrm{normal}}=\Delta$; or \textbf{NonTorsion}; or ``failed''. *)
\State $M\gets$ least common denominator of the $\rho_P$;\quad $n_P\gets M\rho_P\in\Z$;\quad $\mathcal L\gets\emptyset$
\State \textbf{for} each irreducible $p$ (of $R_0$, or monic in $T$) carrying primes of $\supp\Delta$: \Comment{Lemma~\ref{lem:gauss}, Prop.~\ref{prop:split}}
\Statex \hskip1.5em (a) \textbf{if} $n_P=n$ at every prime over $p$: append $(n/M,\,p)$ to $\mathcal L$
\Statex \hskip1.5em (b) \textbf{else if} $p=a{g^*}^2+bg^*+c$ in a generator $g^*$ with $b^2-4ac=s^2q$ and $n_P$ constant on each of the two sheets: append $(n_\pm/M,\ 2ag^*+b\mp sy)$ \Comment{$N(2ag^*+b\mp sy)=4ap$}
\Statex \hskip1.5em (c) \textbf{else} (norm search): \textbf{for} small $k,d_b$ solve $a^2-qb^2=c\,p^{\,k}$ for $a,b\in R$, $\deg b\le d_b$; for a solution $u=a\pm by$ vanishing on a subset $\Sigma$ of the primes over $p$ with orders $\mathrm{ord}_P(u)$, \textbf{if} $n_P/\mathrm{ord}_P(u)=\nu$ is constant on $\Sigma$: append $(\nu/M,\,u)$, remove $\Sigma$ \Comment{$\dv(u)=\sum_\Sigma\mathrm{ord}_P(u)P$ + content; Prop.~\ref{prop:split}}
\Statex \hskip1.5em (d) \textbf{else} (torsion; implemented for $y^2=q$, $\deg q=3$, constant coordinates), on the places $P=(\rho,y_0)$ with $n_P\ne0$ left by (a)--(c), over all $p$ together:
\Statex \hskip3em \textbf{for} each such $P$ whose class $[P-\infty]$ has an order $\mu\le N_{\max}$ ($\psi_\mu(P)=0$, $\psi_N$ the division polynomials of $y^2=q$; or $\mu P=\infty$ by the group law): append $(n_P/(M\mu),\,h_P)$ with $h_P$ the Miller function, $\dv(h_P)=\mu P-\mu\infty$ ($h_{k+1}=h_k\,\ell_{kP,P}/v_{(k+1)P}$); remove $P$
\Statex \hskip3em $S\gets\sum_Pn_PP$ over the remaining $P$ by the group law;\quad $\mu\gets$ the order of $S$ ($1$ if $S=\infty$) \Comment{$\sum_Pn_P(P-\infty)$ has degree $0$ and class $[S-\infty]$}
\Statex \hskip3em \textbf{if} $\mu\le N_{\max}$: $u\gets$ the product of the lines $\ell$ and verticals $v$ of the additions that sum $\mu\sum_Pn_PP$ to $\infty$ ($\dv(\ell_{A,B}/v_{A+B})=A+B-(A+B)-\infty$), so that $\dv(u)=\mu\sum_Pn_P(P-\infty)$; append $(1/(M\mu),\,u)$; \textbf{else} leave $\Delta$ unrealised
\State \textbf{if} $\Delta$ is unrealised: \textbf{if} $[S-\infty]$ is provably non-torsion (reduction modulo good primes, Proposition~\ref{prop:nontorsion}; \cite[\S4.2]{Schultz15}) \textbf{return NonTorsion}; \textbf{else return} ``failed''
\State \textbf{return} $\mathcal L$
\end{algorithmic}
\algfooter

\algheader{Algorithm 4}{$\textsc{ParallelIntegrateMixed}(f,D)$\quad(* Parallel Integration, mixed towers *)}
\begin{algorithmic}[1]
\Statex (* Given a tower (T1)--(T4) and $f\in L$, return an elementary integral of $f$; or \textbf{NotElementary} with a certificate; or ``failed'', in which case it is unknown whether $f$ has an elementary integral over $L$. *)
\Statex (* Steps 1--6 need to be done once per field $L$. *)
\State flatten every pure-root layer (Lemma~\ref{lem:flatten}); \textbf{if} the tower is then a single generator $t_1$ over the curve with $Dt_1=r\in F(t_1)^*$, replace $D$ by $r^{-1}D$ and $f$ by $f/r$ (Lemma~\ref{lem:rescale}); fix $R=F[t_1,\dots,t_n]$, $w_i=y^i/E_i$, $\cA=\bigoplus_iRw_i$
\State $\den_0\gets\mathrm{lcm}_{i\le j}\,\den(Dt_i)$;\quad $h_u\gets\mathrm{lcm}_{i>j}\,\den_T(Dt_i)$ \Comment{Prop.~\ref{prop:typeM}}
\State $\cP\gets$ irreducible factors (over $F$ or $\Fbar$) of $\den_0$, of $Q_1\cdots Q_k$, of $h_u$, and of the $T$-numerators of the $Dt_i$, $i>j$
\State \textbf{for} $p\in\cP$, $P\mid p$: $(e_P,\eta_P,\delta_P,\mathrm{special}_P)\gets\textsc{Classify}(P)$;\quad $\fd_D\gets\sum\eta_PP$
\State $\cS\gets\{P\ \mathrm{special}\}$ \Comment{$=$ the list of Thm.~\ref{thm:specials} for monomial upper towers; else \textbf{FindSpecials} as in \cite[\S10.3]{Bronstein05}}
\State $\cU\gets$ generators of $\cO^*_{S'}/F^*$; record whether the search terminated (for $\deg q$ even: the continued fraction of $y$ found periodic within the height bound, or $[\infty_+-\infty_-]$ certified non-torsion by Proposition~\ref{prop:nontorsion}, in which case $\cU=\emptyset$ is complete) \Comment{Rem.~\ref{rem:blind}; continued fractions/$S$-units, Part I}
\Statex (* The remaining steps are done for each integrand; Steps 7--14 once, Steps 15--20 once per guessed input: a retry with a raised guess at a place that Algorithm~6 leaves undecided re-enters at Step 15, the specials split over $\Fbar$ at Steps 5--6 (Remark~\ref{rem:taxonomy}). *)
\State $a_i,d\gets$ canonical form $f=\sum_i(a_i/d)w_i$, $d\in R$, $\gcd(d,a_0,\dots)=1$ \Comment{Prop.~\ref{prop:ring}: rationalise by norms}
\State $d=d_\ast d_n$, $d_n\gets$ the factors of $d$ \emph{all} of whose primes are normal;\quad $d_n=\prod_ld_l^{\,l}$ squarefree \Comment{Cor.~\ref{cor:element}}
\State $D_v\gets\prod_{p\mid d_n}p^{\,\max_{P\mid p}\lceil\max(0,\,-v_P(f)-\delta_P)/e_P\rceil}$ \Comment{Hermite part; $D_v\mid\prod_ld_l^{\,l-1}$ by Thm.~\ref{thm:structure}(i)}
\State $s\gets\prod_{\sigma}\sigma^{e_\sigma}$ over the irreducible $\sigma\in R$ carrying a special prime (the factors of $d_\ast$ and the candidates $\cS$), with the exponents $e_\sigma$ of Step~15 \Comment{Cor.~\ref{cor:lowerbounds}(a); a guess only where Algorithm~6 is undecided}
\State $\cC\gets\bigl\{(P,\hat\tau_P(f)):P\ \mathrm{normal},\ v_P(f)\le-\delta_P,\ \hat\tau_P(f)\ne0\bigr\}$ by \textsc{CanonicalResidue} at every place over each such $p$ (the closed form $e_P(f\,p/Dp)|_P$ when $v_P(f)=-\delta_P$); \textbf{if} $t_n$ is hypertangent, include $v_\infty$ ($\delta=1$) \Comment{Thm.~\ref{thm:residues}, Prop.~\ref{rem:deepres}, Lemma~\ref{lem:atinfty}}
\State \textbf{for} $(P,\tau)\in\cC$: decide $\tau\in\Fbar$ \emph{exactly} --- write $\tau$ in the canonical form $a+b\sqrt A$ of $\kappa(P)$ over $\Fbar(\text{remaining generators})$; $\tau\in\Fbar$ iff $b=0$ and $a$ is constant; \textbf{if} some $\tau\notin\Fbar$: \textbf{return NotElementary}$(P,\tau)$ \Comment{Prop.~\ref{prop:certificates}(a)}
\State $r_1,\dots,r_l\gets\Q$-basis of $\{\hat\tau_P(f)\}$; $\Delta_k\gets\sum_P\rho_{k,P}P$; $\mathcal L_k\gets\textsc{Realise}(\Delta_k)$; on \textbf{NonTorsion} \textbf{return NotElementary}; on ``failed'' \textbf{return} ``failed'' \Comment{Cor.~\ref{cor:divisors}}
\State $g\gets f-\sum_k r_k\sum_{(\gamma,u)\in\mathcal L_k}\gamma\,Du/u$ \Comment{residue-free at every normal prime}
\State \textbf{for} every prime $P$ over a special $\sigma$, and every place $P$ over $t_i=\infty$, $i=1,\dots,n$ (the places of Section~\ref{sec:lower}): $(\varepsilon_P,\mathrm{proved}_P)\gets\textsc{PlaceBound}(P,g)$ (Algorithm~6);\quad $e_\sigma\gets\max_{P\mid\sigma}\lceil\varepsilon_P/e_P\rceil$;\quad $b_i\gets\deg_{t_i}(sD_v)+\max_{P\mid t_i=\infty}\lceil\varepsilon_P/e_P\rceil$ \Comment{Cor.~\ref{cor:lowerbounds}; at $t_n=\infty$ this is Cor.~\ref{cor:topdeg}, and for $n=1$, $D=d/dx$ it is Part I \cite[Cor.~7.4]{PartI}}
\State $b\gets\sum_{i=0}^{m-1}\sum_{\alpha\le\text{bounds}}u_{\alpha,i}\,t^\alpha w_i$ with undetermined $u_{\alpha,i}\in F$; unknown constants $\gamma_u$ ($u\in\cU$), $\alpha_\sigma$ ($\sigma\in\cS$)
\State clear denominators and equate the coefficients of equal monomials in each $w_i$-coordinate of
\[
g\;=\;D\Bigl(\frac{b}{s\,D_v}\Bigr)
+\sum_{u\in\cU}\gamma_u\frac{Du}{u}
+\sum_{\sigma\in\cS}\alpha_\sigma\frac{D\sigma}{\sigma}
\]
\State solve the linear system for the $u_{\alpha,i}$, $\gamma_u$, $\alpha_\sigma$ by \textsc{AssembleAndSolve} (Algorithm~5): every constant in one number field, the columns by polynomial arithmetic, the matrix reduced to row echelon form and the solution verified on it
\State \textbf{if} it has no solution: \textbf{if} every $\mathrm{proved}_P$ of Step~15 is true, $\cU$ is complete (Step~6 terminated), and $g$ is \emph{verified} residue-free at every affine place (\textsc{CanonicalResidue}, all sheets): \textbf{return NotElementary}(holomorphic remainder); \textbf{else return} ``failed'' \Comment{Prop.~\ref{prop:certificates}(b)}
\State verify by differentiation, and \textbf{return}
\[
\frac{b}{s\,D_v}\;+\;\sum_k r_k\sum_{(\gamma,u)\in\mathcal L_k}\gamma\log(u)\;+\;\sum_{u\in\cU}\gamma_u\log(u)\;+\;\sum_{\sigma\in\cS}\alpha_\sigma\log(\sigma) .
\]
\end{algorithmic}
\algfooter

\algheader{Algorithm 5}{$\textsc{AssembleAndSolve}(g,\,sD_v,\,\mathrm{bounds},\,\cU,\,\cS)$\quad(* Steps 17--18 of Algorithm~4 over one number field *)}
\begin{algorithmic}[1]
\Statex (* Given the residual $g=\sum_i(g_i/d_g)\,w_i$ of Step~14, the denominator $sD_v$, the degree bounds of Step~15 and the candidate sets $\cU$, $\cS$, return the coefficients $u_{\alpha,i}$, $\gamma_u$, $\alpha_\sigma$ of a solution of the system of Step~17, or ``no solution''. Every constant is carried as an element of one number field, and no operation is performed on an expression that contains an unknown. *)
\State $K\gets F(\theta)$, $\theta$ a primitive element of the field generated by the algebraic constants of $Dt_1,\dots,Dt_n$, $q$, $sD_v$, $g$, $\cU$ and $\cS$ --- the field of the tower (Remark~\ref{rem:onefield}), computed once and shared by every step; every constant is written as a polynomial in $\theta$ over $F$ \Comment{$K=F$ when all constants are rational}
\State $\Delta\gets\operatorname{lcm}\bigl(\den(Dt_1),\dots,\den(Dt_n)\bigr)$ in $K[t]$;\quad $\Lambda\gets(sD_v)^2\,\Delta\cdot2q$ \Comment{one denominator for every monomial column}
\State \textbf{for} each monomial $t^\alpha$ within the bounds and each coordinate $w_i$: the column $D(t^\alpha w_i/sD_v)$ as a tuple of numerators over $\Lambda$, by the product rule in $K[t]$ --- for $m=2$, $D(h)=\sum_k\bigl(\partial_{t_k}h\cdot sD_v-h\,\partial_{t_k}(sD_v)\bigr)Dt_k/(sD_v)^2$ with $h=t^\alpha$, and $D(hy)=(Dh)\,y+h\,Dq/(2y)$ reduced modulo $y^2=q$ --- with no cancellation \Comment{$m\ge3$: by the derivation on the Trager basis, one cancelled tuple per column}
\State \textbf{for} $u\in\cU$ and $\sigma\in\cS$: the columns $Du/u$ and $D\sigma/\sigma$, cancelled once each
\State \textbf{for} each coordinate $w_i$: $L_i\gets$ the lcm in $K[t]$ of the distinct denominators of the columns and of $g_i$; multiply each numerator by $L_i$ over its denominator; the coefficient of each monomial of $K[t]$ is one row of the matrix $A$, whose columns are the unknowns in a fixed order, and of the right-hand side $b$ \Comment{the equations of Step~17 up to a polynomial factor per coordinate}
\State $[A\mid b]$ to reduced row echelon form over $K$ (sparse); \textbf{if} the last column is a pivot column: \textbf{return} ``no solution''; \textbf{else} $x\gets$ the solution with the non-pivot unknowns $0$ \Comment{unique for the fixed column order: the row space of $A$ is the annihilator of the solution space, so the pivot columns do not depend on the scaling}
\State verify $Ax=b$ in $K$; \textbf{return} $x$, with the rational part $\sum_\alpha x_{\alpha,i}\,t^\alpha/(sD_v)$ cancelled in $K[t]$ and every coefficient written back in the radicals of $F$
\end{algorithmic}
\algfooter

\algheader{Algorithm 6}{$\textsc{PlaceBound}(v,g)$\quad(* Theorem~\ref{thm:lower} and Proposition~\ref{prop:kernelfree} at one place *)}
\begin{algorithmic}[1]
\Statex (* Given a place $v$ of one of the kinds (P$_\infty$), (P$_\sigma$), ($\infty_i$), (M) of Section~\ref{sec:lower}, with uniformiser $\pi$ and coefficient field $C$ generated by $c_1,\dots,c_l$, and the residual $g$ of Step~14, return $(\varepsilon_v,\mathrm{proved})$ with $v(\text{rational part})\ge-\varepsilon_v$: proved when $\mathrm{proved}$ is true, the classical guess otherwise. Valuations of tuples $\sum_ia_iw_i$ at the places over $t_1=\infty$ are read from the expansion of $y$ (Part I, \S6.1); elsewhere $v$ is a Gauss valuation and $v(\sum_ia_iw_i)=\min_iv(a_iw_i)$. *)
\State $\sigma_k\gets v(Dc_k)$ for $k=1,\dots,l$;\quad $\sigma_\pi\gets v(D\pi)-1$;\quad $s\gets\min(\sigma_\pi,\sigma_1,\dots,\sigma_l)$;\quad $r\gets\sigma_\pi-s$ \Comment{Lemma~\ref{lem:shiftgen}}
\State $A\gets\{k:\sigma_k=s\}$ (the attaining generators);\quad $\lambda_k\gets$ the coefficient of $\pi^{s}$ in $Dc_k$ for $k\in A$;\quad $\rho\gets$ the coefficient of $\pi^{s+1+r}$ in $D\pi$
\State $\mathrm{proved}\gets$ \textbf{true} \textbf{if} one of the following holds, \textbf{false} otherwise: \Comment{Prop.~\ref{prop:kernelfree}}
\Statex \hskip1.5em (K2) $v$ belongs to a monomial $t_i$ ($v_\infty$, or $(t_i)$, $(t_i\mp\sqrt{-1})$) on which no other generator depends;
\Statex \hskip1.5em (K1) $r=0$, every $c_k$ with $k\in A$ is a primitive or hyperexponential upper generator, $\rho$ and the $\lambda_k$ are free of $\{\bar c_k:k\in A\}$, and $\rho$ is not a $\Q$-linear combination of the $\mu_k=\lambda_k/\bar c_k$ of the hyperexponential ones (the implementations decide this by linear algebra over $\Q$ in a common number field when $\rho$ and the $\mu_k$ are numbers, and guess otherwise);
\Statex \hskip1.5em (K3) $r\ge1$, $\kappa_v=\Fbar(x)$ with $\partial_v=d/dx$, and $\rho$ has a nonzero residue (Rothstein--Trager);
\Statex \hskip1.5em (K4) $r\ge1$, $v$ lies over the curve variable, every upper generator attains and is hyperexponential or hypertangent with constant $\mu_k$, the hyperexponential $\mu_k$ linearly independent over $\Q$ (the implementations accept at most one)
\State $\varepsilon_v\gets\max\bigl(0,\,s-v(g)\bigr)+r$;\quad \textbf{if not} $\mathrm{proved}$: $\varepsilon_v\gets\varepsilon_v+e$, $e$ the retry count $0,1,2,\dots$ \Comment{Cor.~\ref{cor:lowerbounds}; the guess contains \cite[(10.4)]{Bronstein05}}
\State \textbf{return} $(\varepsilon_v,\mathrm{proved})$
\end{algorithmic}
\algfooter

\begin{remark}[Return taxonomy]\label{rem:taxonomy}
\textbf{ParallelIntegrate} of \cite[\S10.3]{Bronstein05} has two exits: an integral, or ``failed'' with unknown status. The mixed algorithm has three, and the boundary between them is exactly the proved/heuristic boundary of the theory. \textbf{NotElementary} is returned with a machine-checkable certificate --- a non-constant residue (Theorem~\ref{thm:residues}(iii)) or a non-torsion residue divisor (Corollary~\ref{cor:divisors}) --- and is rigorous. ``failed'' can arise only from a guessed input --- a special exponent or a degree bound at a place that Algorithm~6 leaves undecided (Remark~\ref{rem:decided}), or an inconclusive unit search; a retry with another guess costs Steps 15--20 alone, the classification, the residues, their realisation and the unit and $S'$-unit searches being computed once per integrand and reused. When every bound in force is proved --- always for $n=1$ with $D=d/dx$, where Step~15 returns the exact bounds of Part I (Lemma~\ref{lem:rescale} reduces $Dt_1\in F(t_1)^*$ to this case), and in the mixed towers of Section~\ref{sec:examples} at every place except those of Example~\ref{ex:chainpreview} --- insolvability of the final linear system after all residues have been realised is itself a certificate: the residual $f-\sum_k(r_k/M)D\tilde u_k/\tilde u_k$ has no residues, and is therefore a differential of the second kind that is not exact; we call it the \emph{holomorphic remainder}; on $\int dx/((x-2)\sqrt{x^3+1})$, whose residue divisor is $3$-torsion yet whose remainder is $\tfrac13\,dx/y$, the algorithm certifies non-elementarity this way. Everything else --- the denominator $D_v$, the logand divisors, the unit and special candidate sets, the degree bounds and special exponents at every decided place --- is proved in Sections~\ref{sec:structure}--\ref{sec:degree}.
\end{remark}

The two ways in which the algorithm proves non-elementarity deserve a precise statement, since the classical parallel method proves nothing by failing.

\begin{proposition}[Certificates]\label{prop:certificates}
Let $f\in L$.
\begin{enumerate}
\item[(a)] \textup{(Residue certificate.)} If at some normal height-one prime $P$ of $\Fbar\cA$ \textup(or the hypertangent place at infinity\textup) the canonical residue $\hat\tau_P(f)\in\kappa(P)$ of Proposition~\ref{rem:deepres} does not lie in $\Fbar$, then $f$ has no elementary integral over $L$.
\item[(b)] \textup{(Holomorphic-remainder certificate.)} Let $m=2$ with $q$ squarefree, suppose that every bound in force at Step~15 is proved \textup(Theorem~\ref{thm:lower} at every special prime and at every place at infinity; for $n=1$, $D=d/dx$ --- after Lemma~\ref{lem:rescale} when $Dt_1\in F(t_1)^*$ --- this is automatic, and the bounds are those of Part~I \cite[Cor.~7.4]{PartI}\textup), and suppose the algorithm has produced $\tilde u_1,\dots,\tilde u_l\in(\Fbar L)^*$ and $r_k\in\Fbar$ such that
\[
g\;:=\;f-\sum_k\frac{r_k}{M}\frac{D\tilde u_k}{\tilde u_k}
\]
has zero residue at every affine place of the curve, that the group of units of $\Fbar\cO$ is known \textup(trivial modulo $\Fbar^*$ when $\deg q$ is odd; generated by the fundamental unit $\varepsilon$ when $\deg q$ is even and the continued fraction of $y$ has been found periodic\textup), and that the linear system of Step~17 for $g$ --- with the rational part ranging over the finite-dimensional space cut out by Theorem~\ref{thm:structure}(i) at the normal primes and by Corollary~\ref{cor:lowerbounds} at the special primes and at the places at infinity, and the logarithmic part over the span of $\cU$ and $\cS$ \textup(for $n=1$: $\gamma\,D\varepsilon/\varepsilon$\textup) --- is inconsistent. Then $f$ has no elementary integral over $L$.
\end{enumerate}
\end{proposition}

\begin{proof}
(a) is Proposition~\ref{rem:deepres}: $\hat\tau_P(f)=\sum_ic_iv_P(u_i)\in\Fbar$ for every elementary integral, and $\Fbar$ is embedded in $\kappa(P)$.

(b) Suppose $f$ has an elementary integral. Then so does $g$, and by Theorem~\ref{thm:structure} $g=Dv+\sum_i c_iDu_i/u_i$ with $v\in\Fbar L$, $u_i\in(\Fbar L)^*$, and the $c_i$ may be taken $\Q$-linearly independent by regrouping. At every normal prime $P$ --- every affine prime when $n=1$ with $D=d/dx$, since there is no $\den_0$ and the branch primes have $\delta_P=e_P$ by Proposition~\ref{prop:typeE} --- Theorem~\ref{thm:residues}(iii) and Proposition~\ref{rem:deepres} give $0=\hat\tau_P(g)=\sum_ic_iv_P(u_i)$, whence $v_P(u_i)=0$ for all $i$ by $\Q$-independence: each $u_i$ is an $S'$-unit, its divisor supported on the special primes and at infinity, hence a constant times a product of powers of the special logands $\sigma\in\cS$ and of the generators $\cU$ of $\cO^*_{S'}$, which are complete by hypothesis; for $n=1$ this is a constant times a power of $\varepsilon$, and $\sum_ic_iDu_i/u_i=\gamma\,D\varepsilon/\varepsilon$ \textup($\gamma=0$ when $\deg q$ is odd\textup). The valuation lemma bounds the denominator of $v$ at the normal primes exactly \textup(Theorem~\ref{thm:structure}(i)\textup), and Theorem~\ref{thm:lower} bounds it at every special prime and at every place at infinity, all bounds in force being proved; so $v$ lies in the finite-dimensional $\Fbar$-space over which the linear system was solved, the logarithmic part lies in the span of its columns, and the system would have a solution --- contradiction.
\end{proof}

\begin{remark}
Two points are not decorative. First, ``does not lie in $\Fbar$'' in (a) is a membership question in the field $\kappa(P)$, which must be decided exactly: the residue is reduced to the canonical form $a+b\sqrt A$ with $a,b$ rational functions of the remaining generators over the algebraic constants, and is constant iff $b=0$ and $a\in\Fbar$; a residue that is a constant in disguise, such as $(x+\sqrt{x^2+1})(x-\sqrt{x^2+1})$, must not be mistaken for an obstruction. Second, (b) requires the residual $g$ to be \emph{verified} residue-free rather than assumed so, since an incorrectly realised logand would make the system fail for a reason that proves nothing about $f$; the implementation recomputes the classical residues of $g\,dx$ at every affine place before issuing the certificate. Outside (a) and (b) --- a special exponent or a degree bound at a place that Algorithm~6 leaves undecided, an inconclusive unit search --- the algorithm returns ``failed'' and claims nothing.
\end{remark}

An inconclusive unit search need not stay inconclusive. For $m=2$ the units of $\Fbar\cO$ are the functions with divisor $k(\infty_+-\infty_-)$, and their existence is a torsion question that reduction modulo primes decides in the negative --- the method Schultz \cite[\S4.2]{Schultz15} uses to complete Trager's algorithm, which for our curves needs no divisor arithmetic at all.

\begin{proposition}[Non-torsion by reduction]\label{prop:nontorsion}
Let $q\in\Q[x]$ be squarefree of even degree $2g+2$, $X$ the curve $y^2=q$, and $c=[\infty_+-\infty_-]\in\Jac(X)(\Q)$. For a prime $p$ not dividing $2\,\mathrm{lc}(q)\operatorname{disc}(q)$ at which $\mathrm{lc}(q)$ is a square, let $N_p$ be the order of the reduction of $c$ in $\Jac(X_{\F_p})(\F_p)$; $N_p$ is the degree of the fundamental unit of $\F_p[x][y]/(y^2-q)$, the first convergent $h_k+yk_k$ of the \textup(periodic\textup) continued fraction of $\sqrt q$ over $\F_p$ with $h_k^2-qk_k^2\in\F_p^*$. If $c$ has finite order $N$ then $N=N_pp^{a_p}$ with $a_p\ge0$ for every such $p$. Consequently, if for two such primes $p_1\neq p_2$ the ratio $N_{p_1}/N_{p_2}$, in lowest terms, has a numerator that is not a power of $p_1$ or a denominator that is not a power of $p_2$, then $c$ is not torsion, $\Fbar\cO^*=\Fbar^*$, and no logand of an elementary integral over $L$ can have its divisor supported at infinity.
\end{proposition}
\begin{proof}
$X$ has good reduction at $p$ and $\infty_\pm$ are $\F_p$-rational under the stated conditions, so $c$ reduces to a class $\bar c$ of $\Jac(X_{\F_p})(\F_p)$. A function on $X_{\F_p}$ with divisor $k(\infty_+-\infty_-)$ is a unit of $\F_p[x][y]$, and conversely; the least $k>0$ is the order $N_p$ of $\bar c$, and it is the degree of the fundamental unit since $v_{\infty_-}(h+yk)=-\deg h$ for the convergents. The continued fraction of $\sqrt q$ over a finite field is periodic, so the search terminates. Reduction is injective on the torsion of order prime to $p$ --- \cite[Prop.~VII.3.1]{Silverman09} for $g=1$, and for the Jacobian of any genus by the good-reduction case of \cite[Appendix]{Katz81} --- so if $c$ has order $N=p^{a}N'$ with $p\nmid N'$, then $\bar c$ has order $N'p^{b}$ with $0\le b\le a$, i.e.\ $N_p\mid N$ with $N/N_p$ a power of $p$. Two primes with incompatible $N_{p_i}$ therefore exclude every finite order. The last assertion is Corollary~\ref{cor:divisors}(ii): an $S'$-unit of $\cO$ supported at infinity has divisor $k(\infty_+-\infty_-)$.
\end{proof}

Proposition~\ref{prop:nontorsion} converts Cohen's $-72$ variant from ``failed'' into a certificate (\S\ref{ex:cohen}): with no units, the hypothesis of Proposition~\ref{prop:certificates}(b) that the unit group be known is met, and the exact linear system decides. The positive direction --- confirming a candidate order and constructing the unit --- is the continued fraction over $\Q$ itself, and for the one-place curves of Algorithm~3(d) the division polynomials.

\begin{remark}[Block structure]\label{rem:blocks}
The linear system of step 15 is block-structured by $w_i$-coordinate and by power of the top variable, as every worked example of Section~\ref{sec:examples} displays; the blocks are the decoupled Risch differential equations of \cite{Bronstein89}'s $n$-th-root case, solved simultaneously. Special logands supported on the branch enter only the coordinate blocks fixed by the corresponding inertia (observed in \S\ref{ex:flagship}), which the implementation exploits to shrink the system.
\end{remark}

\begin{remark}[One number field]\label{rem:onefield}
Algorithm~5 is the part of the implementation whose cost the theory does not set, and it is where the first version of the implementation spent most of its time. The columns were formed as expressions, the equations were read off by constructing a polynomial whose coefficients were linear forms in the unknowns, the solution was checked by substituting it into every equation and simplifying, and the rational part was cancelled as an expression --- and every one of these operations on an expression with an algebraic constant (a norm such as $2\pm2\sqrt2$, a root of a special split over $\Fbar$) re-derived the number field or ran the simplifier's assumption queries on nested radicals. With every constant an element of one field $K=F(\theta)$, computed once per system, and the assembly done in the sparse polynomial ring $K[t]$, the assembly and the verification become a negligible part of the running time; the primitive element is the one the generic linear solver would compute internally in any case, so nothing new is computed, it is computed once. On the fifty integrals of the companion paper this reduced the total running time from $369$ to $104$ seconds, the three-generator systems of $229$--$357$ unknowns from $17$--$60$ seconds to under four, with every result unchanged up to the normalisation of its rational part. The same field then became the field of the whole computation: the constants of a tower generate a number field that is built once and cached for the run (extended by a new field when a new constant appears --- an $S'$-unit norm, a root of a split special), and the coordinate arithmetic of $\cA$ (Section~\ref{sec:ring}), the Hensel and linear realisation of residue classes (Algorithm~3), the group law and the Miller functions of Algorithm~3(d) all run in $K[t]$ with exact zero tests, so that no operation of the implementation re-derives a field or simplifies an expression to test it for zero; G\"unther's integral (\S\ref{ex:gunther}), whose six order-six classes had cost seven and a half minutes of such tests, takes four seconds. What remains in the generic expression layer is the norm search of Algorithm~3(c), a nonlinear solve over $\Q$. The Mathematica port of the implementation (\texttt{ParallelMixed.wl}) does the same with \texttt{ToNumberField}: the constants become \texttt{AlgebraicNumber}s of one field, whose arithmetic is canonical (a zero is $0$), the augmented matrix is row reduced exactly, and the solution is written back in the radicals of $F$ by linear algebra over $\Q$ on a monomial basis of the field; the built-in \texttt{Solve} cannot decide whether a combination of \texttt{Root} objects vanishes, and \texttt{Together} on pairs that carry the unknowns is intractable, which had made the conic cases of the companion paper run for hours in that port.
\end{remark}

\section{Examples}\label{sec:examples}

Each example below follows the same format: the integral and, where applicable, its provenance in the literature; its solution by the parallel algorithm of Section~\ref{sec:algorithm}; and the SymPy session that computes it, so that theory and computation are read side by side. The sessions use \texttt{parallel\_mixed.py}, the SymPy implementation of Algorithms 1--4 accompanying this paper (for $m\le2$ and squarefree $q$; every returned integral is verified by differentiation), with the once-only preamble
\begin{verbatim}
>>> from sympy import symbols, S
>>> from parallel_mixed import Tower, parallel_integrate_mixed
>>> x, t, u = symbols('x t u', positive=True); q = x**2 + 1
\end{verbatim}
integrands are passed as coordinate pairs $(f_0,f_1)\leftrightarrow f_0+f_1y$ on the basis $(1,y)$, and every returned integral is verified by differentiation. The script \texttt{examples.py} accompanying the paper runs the entire section.

\subsection{The flagship: $\int\log(x+\sqrt{x^2+1})\,dx$}\label{ex:flagship}

Tower and $\fd_D$ as in Example~\ref{ex:dendiv}. The integrand $f=t$ has no finite poles, so $S=\{(y)\}$, with candidate logands the special $y$ and the Pell unit $u=x+y\in\cO^*$ ($u\bar u=-1$); no moving logands. The ansatz is
\[
\int f=\sum_{k=0}^{2}(A_k+B_ky)t^k+\alpha\log(y)+\beta\log(x+y),\qquad \deg_xA_k,B_k\le1 ,
\]
using $\deg_t\le\deg_tf+1=2$ (primitive bound) and the transfer bound in $x$. Differentiating with $Dt^k=kt^{k-1}/y$, $D\log(y)=x/(x^2+1)$, $D\log(x+y)=1/y$ and equating $w_i$-coordinates per power of $t$: the $t^2$-equation forces $B_2=0$, $A_2$ constant; the $t^1$-equation gives $A_1'=1$ and $B_1'(x^2+1)+B_1x+2A_2=0$, hence $A_1=x$, $B_1=0$, $A_2=0$; the $t^0$-equation splits into the rational part $A_0'+\alpha x/(x^2+1)=0$, forcing $\alpha=0$, and the $y$-part $B_0'(x^2+1)+B_0x+x+\beta=0$, solved by $B_0=-1$, $\beta=0$. Thus
\[
\int\log\bigl(x+\sqrt{x^2+1}\bigr)\,dx=x\log\bigl(x+\sqrt{x^2+1}\bigr)-\sqrt{x^2+1}.
\]
Note the per-$w_i$, per-$t$-power block structure of the linear system: Bronstein's decoupled Risch equations \cite{Bronstein89} appearing as block diagonality.

\begin{verbatim}
>>> T = Tower([x, t], [(1,0), (0, 1/q)], q=q)      # t = log(x+y)
>>> parallel_integrate_mixed((t, 0), T)
t*x - sqrt(x**2 + 1)
\end{verbatim}

\subsection{An irreplaceable unit logand: $\int\bigl(1+xe^{\sqrt{x^2+1}}\bigr)/\sqrt{x^2+1}\,dx$}\label{ex:unit}

In the flagship tower the unit's logarithm \emph{is} the generator $t$, so its coefficient is linearly dependent on the polynomial part. To force the unit, take the hyperexponential $t=e^{y}$ over the same curve: $Dt=t\,x/y$, and
\[
f=\frac{1+x\,t}{y}=\frac{(1+xt)\,y}{x^2+1},
\]
so $a_0=0$, $a_1=1+xt$, $d=x^2+1$. The tower data: $\fd_D=\dv(y)$ as in Example~\ref{ex:dendiv} (now via $-v_{(y)}(Dt)=1$), and $(y)$ is again normal with $\delta_{(y)}=2$ --- the tie case $\mu=1=e(1+\nu)-1$ of Proposition~\ref{prop:typeE}. The moving denominator is trivial, and the only special prime is $(t)$.

The integrand has a single finite pole, of order $1$ at the branch prime: $v_{(y)}(f)=-1$. Since $\delta_{(y)}=2$, this pole is \emph{sub-critical}: $v_{(y)}(f)=1-\delta_{(y)}$, so $(y)\notin S$, no Hermite denominator arises, and $\tau_{(y)}(f)=0$ by Theorem~\ref{thm:residues}(ii). (Indeed both $Dt$ and $D\log(x+y)=1/y$ produce order-$1$ poles at the branch: below the threshold $\delta_P$, poles are created by $Dv$ and by logarithmic derivatives alike --- an algebraic relaxation invisible in the transcendental theory, where $\delta_P=1$ leaves no room below the threshold.) All residues therefore vanish, and by Remark~\ref{rem:blind} the candidate logands are exactly the residue-invisible ones: the Pell unit $x+y\in\cO^*$, and the special $\log(t)$, which is redundant since $D\log(t)=Dy$.

With the proved bound $\deg_t(b)\le\max(\deg_ta,\deg_td)=1$ (Theorem~\ref{thm:top}(b)) and the heuristic bound in $x$, the ansatz is
\[
\int f=\sum_{k=0}^{1}(A_k+B_ky)\,t^k+\alpha\log(t)+\beta\log(x+y),\qquad A_k,B_k\in\Q[x].
\]
Differentiating and equating $w_i$-coordinates per power of $t$ gives four blocks:
\begin{alignat*}{2}
(t^1,w_0):&\quad A_1'+B_1x=0, &\qquad (t^1,w_1):&\quad B_1'(x^2+1)+(A_1+B_1)x=x,\\
(t^0,w_0):&\quad A_0'=0, &\qquad (t^0,w_1):&\quad B_0'(x^2+1)+B_0x+\alpha x+\beta=1 .
\end{alignat*}
The $t^1$ blocks force $A_1=1$, $B_1=0$; the $(t^0,w_1)$ block --- where $\beta$ is determined --- has the general polynomial solution $B_0=-\alpha$, $\beta=1$. The general solution is therefore
\[
\int f\;=\;t+\log(x+y)+\alpha\,(\log(t)-y)+c ,
\]
with $\alpha,c$ arbitrary constants ($\log(t)-y=0$ is the hyperexponential redundancy, exactly as in \cite[Ex.~10.3.3]{Bronstein05}); choosing $\alpha=c=0$,
\[
\int\frac{1+x\,e^{\sqrt{x^2+1}}}{\sqrt{x^2+1}}\,dx\;=\;e^{\sqrt{x^2+1}}+\log\bigl(x+\sqrt{x^2+1}\bigr).
\]
The coefficient $\beta=1$ is forced: no polynomial $B_0$ can produce the constant on the right of the $(t^0,w_1)$ block. The unit logand, invisible to every residue, is irreplaceable.

\begin{verbatim}
>>> T = Tower([x, t], [(1,0), (0, x*t/q)], q=q)    # t = exp(y)
>>> parallel_integrate_mixed((0, (1+x*t)/q), T)
t + log(x + sqrt(x**2 + 1))
\end{verbatim}

\subsection{Non-elementarity at a moving prime}\label{ex:moving}

In the flagship tower, $f=1/(x\,t)$ has a simple pole at the moving prime $(t)$, which is normal with $\delta=1$, and
\[
\tau_{(t)}(f)=\Bigl(f\,\frac{t}{Dt}\Bigr)\Big|_{t=0}=\frac{y}{x}\notin\Fbar ,
\]
so $\int dx/\bigl(x\log(x+\sqrt{x^2+1})\bigr)$ is not elementary, by Theorem~\ref{thm:residues}, in one line.

\begin{verbatim}
>>> T = Tower([x, t], [(1,0), (0, 1/q)], q=q)      # t = log(x+y)
>>> parallel_integrate_mixed((1/(x*t), 0), T, verbose=True)
  (t): unramified, delta = 1, v_P(f) = -1
      places ['y=sqrt(x**2 + 1)', 'y=-sqrt(x**2 + 1)']:
          residues [sqrt(x**2 + 1)/x, -sqrt(x**2 + 1)/x]
('not elementary', t, sqrt(x**2 + 1)/x)
\end{verbatim}

\subsection{A nested radical: $\int\bigl(\log(x)+\sqrt{\log(x)+\sqrt{\log(x)}}\bigr)/(1+\log(x))\,dx$, Bronstein's Example (E)}\label{ex:bronsteinE}

This is Example (E) of \cite{Bronstein89}, worked again on p.~138 of \cite{Bronstein90}: writing $\theta=\log(x)$ and $Y=\sqrt{\theta+\sqrt\theta}$, the integrand is $f=(\theta+Y)/(1+\theta)$, a tower of two square roots --- the nested radical of the abstract. The inner layer $z^2=\theta$ is a pure root and flattens (Lemma~\ref{lem:flatten}): with $t:=\sqrt{\log(x)}$, $Dt=1/(2xt)$, the field is $\Q(x,t)(y)$ with the \emph{single} radical
\[
y^2=q=t^2+t,\qquad f=\frac{t^2+y}{1+t^2},
\]
$q$ squarefree, (N2) satisfied. The tower data: $\den_0=xt$; $x$ is special ($\Dbar_0x=xt$; indeed $x=\exp(t^2)$ in these coordinates), $t$ is normal with $v_t(\den_0)=1$; the branch primes lie over $(t)$ and $(t+1)$ with $e=2$, and Proposition~\ref{prop:typeE} gives $\delta_{P_t}=e(1+v_t(\den_0))=4$ (check: $v_{P_t}(y)=1$, $Dy=(2t+1)/(4xty)$, $v_{P_t}(Dy)=-3=1-4$), $\delta_{P_{t+1}}=2$ --- a live instance of Remark~\ref{rem:partone}.

The flattening is a change of coordinates on one differential field, so nothing about elementary integrability is altered by it. The denominator $1+t^2$ of $f$ is coprime to $q$ and to $\den_0$, so $\eta_P=0$ and $\delta_P=1$ at every prime $P$ over it; and $D(t^2+1)=2t\,Dt=1/x$ is a unit at $P$, so $v_P(\Dbar_P\pi)=0$ with $\pi=t^2+1$: every such $P$ is normal (Definition~\ref{def:normal}), not by assumption but by inspection. Over $\Fbar$ there are four of them, $t=\pm i$, $y=y_0:=\pm\sqrt{-1\pm i}$, each with residue field $\kappa(P)=\Fbar(x)$, $x$ transcendental; at each, the numerator $t^2+y$ takes the nonzero value $-1+y_0$, so $v_P(f)=-1=-\delta_P$ and Theorem~\ref{thm:residues} gives
\[
\tau_P(f)=\Bigl(f\,\frac{t^2+1}{D(t^2+1)}\Bigr)\Big|_P=x\,\bigl(-1+y_0\bigr)\;\in\;\Fbar(x)\setminus\Fbar .
\]
If $f$ had an elementary integral over $L$, Theorem~\ref{thm:residues}(iii) would force $\tau_P(f)=\sum_ic_iv_P(u_i)\in\Fbar$; hence it has none (Proposition~\ref{prop:certificates}(a)). The conclusion is insensitive to the constant field: with $F=\C$ the same four residues arise and $x\notin\C$, and an elementary antiderivative in the classical sense would lie in an elementary extension of $L$, since $L$ is itself elementary over $\C(x)$. Moreover the four residues are precisely the roots of Bronstein's resultant
\[
R(z)=z^4+4xz^3+8x^2z^2+8x^3z+5x^4,
\]
as the symmetric functions confirm: $e_1=-4x$, $e_2=8x^2$, $e_3=-8x^3$, $e_4=5x^4$. The parallel method thus reproduces the criterion of \cite{Bronstein89} from four residue evaluations, with no Hermite reduction, no Puiseux expansions, no resultant and no integral-basis computation.

\begin{verbatim}
>>> TE = Tower([x, t], [(1,0), (1/(2*x*t), 0)], q=t**2+t)
>>> parallel_integrate_mixed((t**2/(1+t**2), 1/(1+t**2)), TE, verbose=True)
  (t**2 + 1): unramified, delta = 1, v_P(f) = -1
      places [...]: residues [x*(-1 + sqrt(-1 + I)),
          -x*(1 + sqrt(-1 + I)), x*(-1 + sqrt(-1 - I)),
          -x*(1 + sqrt(-1 - I))]
('not elementary', t**2 + 1, x*(-1 + sqrt(-1 + I)))
\end{verbatim}

\subsection{Flattening and $\delta_P=2$: $\int e^{\sqrt x}\,dx$}\label{ex:expsqrt}

The layer $y^2=x$ is a pure root, so Lemma~\ref{lem:flatten} absorbs it: with $u:=\sqrt x$ the field is $\Q(u)(t)$, $t=e^{u}$, $Du=1/(2u)$, $Dt=t/(2u)$. No radical remains, but $u$ is not a monomial, and $(u)$ is a \emph{normal} prime dividing $\den_D\doteq u$ with $\delta_{(u)}=2$: the tower of Remark~\ref{rem:partone}, with the exponential on top. The integrand $f=t$ has no finite poles and no residues; the only special prime is $(t)$. Theorem~\ref{thm:top}(b) gives $\deg_t\le1$, and the ansatz $P(u)\,t+A_0$ (the redundant $\alpha\log(t)=\alpha u$ omitted) yields, after clearing $2u$, the single decoupled equation
\[
P'+P=2u
\]
--- the Risch differential equation of the $n$-th-root case of \cite{Bronstein89}, appearing here as one block of the linear system --- with unique polynomial solution $P=2u-2$. Hence
\[
\int e^{\sqrt x}\,dx\;=\;2\bigl(\sqrt x-1\bigr)\,e^{\sqrt x}.
\]
This example is flattenable by design: its subject is the flattening itself and the ensuing $\delta_P=2$. The hyperexponential mechanism with the radical irreducibly present is certified by \S\ref{ex:unit}, where $t=e^{y}$ over the unflattenable curve.

\begin{verbatim}
>>> T = Tower([u, t], [(1/(2*u),0), (t/(2*u),0)])  # u=sqrt(x), t=exp(u)
>>> parallel_integrate_mixed((t, 0), T)
2*t*u - 2*t
\end{verbatim}

\subsection{A residue at the hypertangent place at infinity: $\int\tan(\sqrt x)/\sqrt x\,dx$}\label{ex:tansqrt}

Flatten as before: $L=\Q(u)(t)$ with $t=\tan(u)$, $Du=1/(2u)$, $Dt=(1+t^2)/(2u)$, and $f=t/u$. The special primes are the factors of $1+t^2$; the prime $(u)$ is normal with $\delta_{(u)}=2$ (threshold $e(1+\nu)-1=1$, $\mu_{(u)}=1$: another tie), and $v_{(u)}(f)=-1$ is sub-critical, so no finite residues arise. The place at infinity, however, is \emph{normal} here (Lemma~\ref{lem:atinfty}(c)), with $\delta_{v_\infty}=1$ and $v_\infty(f)=-1=-\delta_{v_\infty}$, and carries a residue: with $\pi=1/t$ and $D\pi=-\tfrac1{2u}(1+\pi^2)$,
\[
\tau_{v_\infty}(f)\;=\;\Bigl(f\,\frac{\pi}{D\pi}\Bigr)\Big|_{v_\infty}\;=\;\Bigl(\frac{-2}{1+\pi^2}\Bigr)\Big|_{\pi=0}\;=\;-2\ \in\ \Q .
\]
Corollary~\ref{cor:divisors} with basis $r_1=-2$ gives $\Delta_1=v_\infty$, realised with $M=2$ by $\tilde u_1=(1+t^2)^{-1}$ ($v_\infty(\tilde u_1)=2$; its special part, the poles at $t=\pm\sqrt{-1}$, is unconstrained), so the candidate logarithmic part is $\tfrac{r_1}{M}\log(\tilde u_1)=\log(1+t^2)$. Theorem~\ref{thm:top}(c) gives $v_\infty(v)\ge\min(0,-1+1)=0$: \emph{no} polynomial part in $t$ at all; and indeed $f-D\log(1+t^2)=0$, so $v$ is a constant and
\[
\int\frac{\tan(\sqrt x)}{\sqrt x}\,dx\;=\;\log\bigl(1+\tan^2(\sqrt x)\bigr)\;=\;-2\log\bigl(\cos(\sqrt x)\bigr) .
\]
As a check of Theorem~\ref{thm:residues}(iii): $c_1v_\infty(u_1)=1\cdot v_\infty(1+t^2)=-2=\tau_{v_\infty}(f)$. Like \S\ref{ex:expsqrt}, this example is flattenable by design --- its subject is the place at infinity; the hypertangent mechanism over an unflattenable curve follows.

\begin{verbatim}
>>> T = Tower([u, t], [(1/(2*u),0), ((1+t**2)/(2*u),0)])   # t = tan(u)
>>> parallel_integrate_mixed((t/u, 0), T, verbose=True)
  (u): unramified, delta = 2, v_P(f) = -1  [sub-critical]
  tower special: candidate log(t**2 + 1)
log(t**2 + 1)
\end{verbatim}

\subsection{Hypertangent over the curve}\label{ex:tancurve}

Take $t=\tan\left(y\right)$ over $y^2=x^2+1$: $Dt=x(1+t^2)/y$, a hypertangent monomial whose derivative involves $y$; the tower does not flatten. The data: $h_u=1$; the special primes are the factors of $1+t^2$; and at the branch, $\mu_{(y)}=-v_{(y)}(Dt)=1=e(1+\nu)-1$, the tie once more: $(y)$ is normal with $\delta_{(y)}=2$.

First, non-elementarity in one line. For $f=t$, i.e.\ the integrand $\tan\left(\sqrt{x^2+1}\right)$: at the (normal) place at infinity, $v_\infty(f)=-1=-\delta_{v_\infty}$, and with $\pi=1/t$, $D\pi=-\tfrac{x}{y}(1+\pi^2)$,
\[
\tau_{v_\infty}(f)\;=\;\Bigl(f\,\frac{\pi}{D\pi}\Bigr)\Big|_{v_\infty}\;=\;\Bigl(\frac{-y}{x(1+\pi^2)}\Bigr)\Big|_{\pi=0}\;=\;-\frac{y}{x}\ \notin\ \Fbar ,
\]
so, by Theorem~\ref{thm:residues}(iii) applied at $v_\infty$,
\[
\int\tan\left(\sqrt{x^2+1}\right)\;dx\quad\text{is not elementary.}
\]
The residue at the hypertangent place at infinity now takes values in the function field of the curve, and its non-constancy is the certificate --- the $v_\infty$-counterpart of \S\ref{ex:moving}.

Second, an elementary instance with a genuinely algebraic polynomial part:
\[
f\;=\;x\,(1+t^{2})\;+\;\frac{3xt}{y}\,.
\]
The branch pole is sub-critical ($v_{(y)}(f)=-1>-2$), there are no moving poles, and $v_\infty(f)=-2$, so Theorem~\ref{thm:top}(c) gives $v_\infty(v)\ge\min(0,-1)=-1$: $\deg_t(v)\le1$. The ansatz $\sum_{k\le1}(A_k+B_ky)t^k+\beta\log(1+t^2)+\gamma\log(x+y)$ solves uniquely: the $(t^2,w_0)$ block forces $B_1=1$, the $(t^2,w_1)$ block $A_1=0$, the $(t^1,w_1)$ block --- where $\beta$ is determined --- reads $B_1'(x^2+1)+B_1x+2\beta x=3x$, giving $\beta=1$, and the remaining blocks leave $A_0$ constant and $B_0=\gamma=0$. Hence
\begin{multline*}
\int\Bigl(x\bigl(1+\tan^{2}\left(\!\sqrt{x^2+1}\right)\bigr)+\frac{3x\tan\left(\sqrt{x^2+1}\right)}{\sqrt{x^2+1}}\Bigr)dx\\
=\;\sqrt{x^2+1}\,\tan\left(\sqrt{x^2+1}\right)\;+\;\log\bigl(1+\tan^{2}\left(\!\sqrt{x^2+1}\right)\bigr),
\end{multline*}
attaining the bound of Theorem~\ref{thm:top}(c) sharply, with the top coefficient of the rational part lying in the $w_1$-component: the polynomial part of the integral is itself algebraic, $v=y\,t$ --- a configuration invisible both in the transcendental theory and in Part I.

\begin{verbatim}
>>> T = Tower([x, t], [(1,0), (0, x*(1+t**2)/q)], q=q)     # t = tan(y)
>>> parallel_integrate_mixed((x*(1+t**2), 3*x*t/q), T, verbose=True)
  (x**2 + 1): branch, delta = 2, v_P(f) = -1  [sub-critical]
  tower special: candidate log(t**2 + 1)
  unit candidate: A + B*y with deg_x B = 0
t*sqrt(x**2 + 1) + log(t**2 + 1)
\end{verbatim}

\noindent The non-elementary instance is certified by the residue at $v_\infty$, with $\pi=1/t$ and $D\pi=-\eta(1+t^2)/t^2$:

\begin{verbatim}
>>> parallel_integrate_mixed((t, 0), T, verbose=True)
  v_oo (hypertangent top): delta = 1, v_oo(f) = -1, residue 0 + (-1/x)*y
('not elementary', 'v_oo', (0, -1/x))
\end{verbatim}

\noindent The residue $-y/x\in\kappa(v_\infty)=\Q(x)(y)$ is the one computed above; it is not a constant, and Proposition~\ref{prop:certificates}(a) applies.

\subsection{Moving logands split by the curve}\label{ex:split}

The moving mechanism is genuinely exercised only when the radical cannot be flattened away and the moving factorisation happens over the curve. Return to the flagship tower --- $y^2=x^2+1$, $t=\log(x+y)$, $Dt=1/y$ --- which does \emph{not} flatten (the coordinate ring of the curve has the unit $x+y$ and is not a polynomial ring; see the remark after Lemma~\ref{lem:flatten}), and consider
\[
f\;=\;\frac{t^{3}+(4+x-x^{2})\,t-(1+5x)\,y}{y\,\bigl(t^{2}-x^{2}-1\bigr)}
\;=\;\frac{t}{y}+\frac{2(1-x)}{y\,(t-y)}+\frac{3(1+x)}{y\,(t+y)}\,,
\]
with canonical data $a_0=-(1+5x)(x^2+1)$, $a_1=t^{3}+(4+x-x^{2})t$, $d=(x^2+1)\,q$, $q:=t^{2}-x^{2}-1$.

The $R$-irreducible $q$ remains irreducible over $\Fbar K=\Fbar(x,t)$ --- $x^2+1$ is not a square in $\Fbar(x)$ --- but splits over $\mathcal F=K_0(y)$ as $q=(t-y)(t+y)$: the moving factorisation \emph{must} be performed over the function field of the curve, and no constant-field extension can substitute (contrast \cite[Ex.~10.3.2]{Bronstein05}, where the analogous splitting is achieved by adjoining $\sqrt2$ to the constants). Both moving primes are normal with $\delta=1$ ($h_u=1$), and they carry \emph{distinct} residues: with $D(t\mp y)=(1\mp x)/y$,
\[
\tau_{(t-y)}(f)\;=\;\Bigl(f\,\frac{t-y}{D(t-y)}\Bigr)\Big|_{t=y}\;=\;2,
\qquad
\tau_{(t+y)}(f)\;=\;3 ,
\]
so a single irreducible factor of $d$ in $R$ carries two different residues --- unobtainable by any computation over $\Fbar K$. Corollary~\ref{cor:divisors} (basis $r_1=1$, $M=1$) gives $\Delta_1=2\,P_{t-y}+3\,P_{t+y}$, realised by $\tilde u_1=(t-y)^2(t+y)^3$; the content divisors vanish here ($\gamma(t\mp y)=0$), a nontrivial content requiring $\Cl(\cO)\ne0$, hence genus $\ge1$ (\S\S\ref{ex:chebyshev}--\ref{ex:gunther}). The branch pole is once more sub-critical, $v_{(y)}(f)=-1>-\delta_{(y)}=-2$, contributing neither Hermite denominator nor residue. Subtracting the logarithmic part leaves $f-D\log(\tilde u_1)=t/y$, which the polynomial ansatz integrates under the (sharp) bound $\deg_t\le2$ of Theorem~\ref{thm:top}(a) to $v=t^2/2$. Hence, with $y=\sqrt{x^2+1}$ and $t=\log(x+y)$,
\[
\int f\,dx\;=\;\tfrac12\,t^{2}\;+\;2\log\bigl(t-y\bigr)\;+\;3\log\bigl(t+y\bigr).
\]

\begin{verbatim}
>>> T = Tower([x, t], [(1,0), (0, 1/q)], q=q)      # t = log(x+y)
>>> parallel_integrate_mixed((-(1+5*x)/(t**2-x**2-1),
...     (t**3+(4+x-x**2)*t)/(q*(t**2-x**2-1))), T, verbose=True)
  (-t**2 + x**2 + 1): unramified, delta = 1, v_P(f) = -1
      places ['y=t', 'y=-t', 'y=t', 'y=-t']: residues [2, 3, 2, 3]
      (two sheets over each of the two roots t = +-sqrt(x^2+1))
      y-split factor -2*t + (-2)*y: residue 3
      y-split factor -2*t + (2)*y: residue 2
  (x**2 + 1): branch, delta = 2, v_P(f) = -1  [sub-critical]
  unit candidate: A + B*y with deg_x B = 0
t**2/2 + 3*log(-2*t - 2*sqrt(x**2 + 1)) + 2*log(-2*t + 2*sqrt(x**2 + 1))
\end{verbatim}

\subsection{Benchmark: Example 14 of Bronstein's \textit{Symbolic Integration Tutorial}}\label{ex:benchmark}

Example 14 of Bronstein's integration tutorial \cite{Bronstein98} is the showcase for the algebraic logarithmic case:
\begin{multline*}
\int\frac{(x^2+2x+1)\sqrt{x+\log(x)}+(3x+1)\log(x)+3x^2+x}{(x\log(x)+x^2)\sqrt{x+\log(x)}+x^2\log(x)+x^3}\,dx\\
=\;2\sqrt{x+\log(x)}\;+\;2\log\bigl(x+\sqrt{x+\log(x)}\bigr),
\end{multline*}
computed there over $\Q(x)(t)[y]/(y^2-x-t)$, $t=\log(x)$, by the full recursive apparatus: an integral basis made normal at infinity, the residue resultant, a $\rho$-extraction from Puiseux expansions at infinity requiring a recursive integration, a principality test for the divisor $2P-2Q$ on the completed curve, and the Risch differential equations (23)--(24) of \cite{Bronstein98}.

The parallel method dispatches it with none of that --- and the flattening lemma explains why. The defining polynomial is linear in the generator $t$ with unit coefficient, so $\cO=\Q[x,t][y]/(y^2-x-t)\cong\Q[x,y]$: by Lemma~\ref{lem:flatten} the tower flattens to $\Q(x)(u)$, $u=\sqrt{x+\log(x)}$, $Du=(x+1)/(2xu)$, with \emph{no radical at all}: the tutorial's algebraic showcase is a purely transcendental parallel computation with one non-monomial generator. In flattened coordinates, $\den_0\doteq xu$; $(x)$ is special (it is the argument prime of $\log(x)=u^2-x$), $(u)$ is normal with $\delta_{(u)}=2$, and the branch pole of the integrand is sub-critical. The only residue is
\[
\tau_{(x+u)}(f)\;=\;\Bigl(f\,\frac{x+u}{D(x+u)}\Bigr)\Big|_{u=-x}\;=\;2
\]
at the normal prime $(x+u)$, realised by $x+u$ itself ($M=1$); the residual integrand $(x+1)/(xu)$ is met by the one-unknown block $b\,(x+1)/(2x)=(x+1)/x$, giving $b=2$. Hence $\int f=2u+2\log(x+u)$: one residue evaluation and one linear equation, with no integral basis, no place at infinity, no Puiseux expansion, and no Risch differential equation. (The final display of \cite[Ex.~14]{Bronstein98}, like the original on p.~147 of \cite{Bronstein90}, has $x^2+x+1$ for $x^2+2x+1$ in the integrand; that integrand is in fact not elementary; the correct one, in lowest terms, is displayed in \S\ref{ex:jsc90}.)

\begin{verbatim}
>>> T = Tower([x, u], [(1,0), ((x+1)/(2*x*u), 0)]) # u = sqrt(x+log x)
>>> parallel_integrate_mixed(((x+1)/(x*u)
...     + (2*x*u+x+1)/(x*u*(x+u)), 0), T)
2*u + 2*log(u + x)
\end{verbatim}

\subsection{Benchmark: Example 15 of Bronstein's \textit{Symbolic Integration Tutorial}}\label{ex:bench15}

Example 15 of \cite{Bronstein98}, the showcase for the algebraic \emph{exponential} case, succumbs identically:
\[
\int\frac{3(x+e^x)^{1/3}+(2x^2+3x)e^x+5x^2}{x\,(x+e^x)^{1/3}}\,dx
\;=\;3x\,(x+e^x)^{2/3}\;+\;3\log(x) .
\]
There it is computed over $\Q(x)(t)[y]/(y^3-t-x)$, $t=e^x$, by the residue resultant, the inversion $t\mapsto t^{-1}$ with a recomputed integral basis $(1,z,z^2/\bar t\,)$, three Risch differential equations, and a recursive integration for the remainder $3/x$. Again the defining polynomial is linear in $t$ with unit coefficient, so the tower flattens (Lemma~\ref{lem:flatten}): with $u=(x+e^x)^{1/3}$, $Du=(u^3-x+1)/(3u^2)$ and $e^x=u^3-x$, the field is $\Q(x)(u)$ with no radical. Now $\den_0\doteq u^2$ is \emph{not} squarefree, and $(u)$ is a normal prime with $\nu=2$, hence $\delta_{(u)}=3$ --- the first $m=3$ instance in this paper, matching $\delta_P=e_P=3$ at the cube-root branch in the unflattened coordinates: an internal consistency check on Proposition~\ref{prop:typeE}. The integrand
\[
f\;=\;\frac{(2x^2+3x)\,u^{3}+3u+2x^2-2x^3}{x\,u}
\]
has a sub-critical pole at $(u)$ ($v_{(u)}(f)=-1>-3$: at a cube-root branch, poles of orders $1$ \emph{and} $2$ are free) and a critical one at the normal prime $(x)$, with
\[
\tau_{(x)}(f)\;=\;(f\,x)\big|_{x=0}\;=\;3 ,
\]
read off in one evaluation --- the $3\log(x)$ that the recursive algorithm obtains from its base-case recursion. The remainder $f-3/x=(2x+3)u^2+2x(1-x)/u$ has only the sub-critical pole, so $v$ is \emph{polynomial} by Theorem~\ref{thm:structure}(i), and the single block $D(bxu^2)$ returns $b=3$: $\int f=3xu^2+3\log(x)$. One residue and one linear block replace the basis inversion, the three Risch equations and the recursion. Both algebraic showcases of \cite{Bronstein98} are thus flattenable --- evidence that the published example corpora should be swept for flattenable instances.

\begin{verbatim}
>>> T = Tower([x, u], [(1,0), ((u**3-x+1)/(3*u**2), 0)])
>>> parallel_integrate_mixed(
...     (((2*x**2+3*x)*u**3+3*u+2*x**2-2*x**3)/(x*u), 0), T)
3*u**2*x + 3*log(x)
\end{verbatim}

\subsection{Benchmark: the examples on pp.~134 and 147 of Bronstein's \textit{Integration of Elementary Functions}}\label{ex:jsc90}

The original source of the tutorial's Example 14 is Bronstein's 1990 paper \cite{Bronstein90}, which works two integrands on the curve $y^2=x+\log(x)$. Both flatten as in \S\ref{ex:benchmark}, to $\Q(x)(u)$ with $u=\sqrt{x+\log(x)}$, $Du=(x+1)/(2xu)$, $\log(x)=u^2-x$. The first (p.~134),
\[
\int\frac{(x+1)\,dx}{(x\log(x)+x^2)\sqrt{x+\log(x)}}\;=\;-\frac{2}{\sqrt{x+\log(x)}},
\]
is a pure Hermite computation: the integrand is $(x+1)/(xu^3)$, with a pole of order $3$ at the normal prime $(u)$ of $\delta=2$, so Corollary~\ref{cor:element} gives $D_v=u$ and the linear system returns $-2/u$:

\begin{verbatim}
>>> T = Tower([x, u], [(1,0), ((x+1)/(2*x*u), 0)])  # u = sqrt(x + log x)
>>> parallel_integrate_mixed(((x+1)/(x*u**3), 0), T, verbose=True)
  (u): unramified, delta = 2, v_P(f) = -3
  (x): special; s-part x**1, candidate log(x)
-2/u
\end{verbatim}

The second (p.~147) is the integrand of \S\ref{ex:benchmark}. As printed there, and as restated in the final display of \cite[Ex.~14]{Bronstein98}, the coefficient of $\sqrt{x+\log(x)}$ in the numerator is $x^2+x+1$; with it the pipeline finds a \emph{non-constant} residue,

\begin{verbatim}
>>> f = ((x**2+x+1)*u + (3*x+1)*(u**2-x) + 3*x**2 + x) \
...     / ((x*(u**2-x)+x**2)*u + x**2*(u**2-x) + x**3)
>>> parallel_integrate_mixed((f, 0), T, verbose=True)
  (u): unramified, delta = 2, v_P(f) = -1  [sub-critical]
  (x): special; s-part x**1, candidate log(x)
  (u + x): unramified, delta = 1, v_P(f) = -1
      residues [2*(2*u**2 - 1)/(2*u**2 + u - 1)]
('not elementary', u + x, 2*(2*u**2 - 1)/(2*u**2 + u - 1))
\end{verbatim}

\noindent so that integrand has no elementary integral (Theorem~\ref{thm:residues}(iii)). With $x^2+2x+1=(x+1)^2$ instead --- the integrand displayed in the body of \cite[Ex.~14]{Bronstein98} --- the same call returns $2u+2\log(u+x)$, the integral both sources state. In lowest terms the correct integrand is
\[
\int\frac{(x+1)^2+(3x+1)\sqrt{x+\log(x)}}{x\,\sqrt{x+\log(x)}\,\bigl(x+\sqrt{x+\log(x)}\bigr)}\;dx
\;=\;2\sqrt{x+\log(x)}\;+\;2\log\bigl(x+\sqrt{x+\log(x)}\bigr),
\]
the long denominator $(x\log(x)+x^2)\sqrt{x+\log(x)}+x^2\log(x)+x^3$ being $x\,u^2(u+x)$:

\begin{verbatim}
>>> parallel_integrate_mixed((((x+1)**2 + (3*x+1)*u)/(x*u*(u+x)), 0),
...                          T, verbose=True)
  (u): unramified, delta = 2, v_P(f) = -1  [sub-critical]
  (x): special; s-part x**1, candidate log(x)
  (u + x): unramified, delta = 1, v_P(f) = -1
      residues [2]
2*u + 2*log(u + x)
\end{verbatim}

\noindent The residue at $(u+x)$ is now the constant $2$, exactly the coefficient of the logarithm. The discrepancy is therefore not a misprint in the answer but in the integrand, by a single missing $x$, and the residue criterion detects it in one evaluation.

\subsection{Benchmark: Cohen's 1993 pseudo-elliptic integral from \texttt{sci.math.symbolic}}\label{ex:cohen}

This benchmark is the pseudo-elliptic integral posted by H.~Cohen to \texttt{sci.math.symbolic} in 1993 \cite{Cohen93}, now the showcase example of the Risch algorithm's popular accounts:
\[
\int\frac{x\;dx}{\sqrt{x^4+10x^2-96x-71}}
\;=\;\frac18\,\log\Bigl(A+B\,\sqrt{x^4+10x^2-96x-71}\Bigr),
\]
with $B=x^6+15x^4-80x^3+27x^2-528x+781$ and $A=x^8+20x^6-128x^5+54x^4-1408x^3+3124x^2+10001$. Here the framework's trichotomy degenerates completely: $q$ is squarefree, every finite prime is a normal branch prime with $\delta_P=2$, and $v_P(f)=-1$ is sub-critical \emph{everywhere} --- no Hermite part, no residues, $S=\varnothing$. The entire integral is carried by a logand invisible to every residue (Remark~\ref{rem:blind}): a unit of $\cO$, $A^2-qB^2\in\Q^*$, whose existence is the order-$8$ torsion of $\infty_+-\infty_-$ in the Jacobian \cite{PvdP05}, found by the continued fraction of $\sqrt q$ (Abel--Chebyshev; the elementary/non-elementary dichotomy for such integrals is Chebyshev's, with Zolotarev's proof). The integral is computed by \texttt{parallel\_mixed.py}: the classification finds every finite pole sub-critical, the unit is offered by the $S$-unit subroutine \texttt{pell.py} (found in a fraction of a second, matching the published coefficients exactly), and the linear system returns $\gamma=\tfrac18$. For Cohen's companion variant with constant term $-72$, famously non-elementary, the continued fraction exhibits the non-torsion signature --- doubly-exponential coefficient growth --- and the height guard reports no unit; Proposition~\ref{prop:nontorsion} then certifies that none exists, and the exact linear system completes the proof (Proposition~\ref{prop:certificates}(b)):

\begin{verbatim}
>>> parallel_integrate_mixed((0, x/(x**4+10*x**2-96*x-72)),
...     Tower([x], [(S(1), S(0))], q=x**4+10*x**2-96*x-72), verbose=True)
  (x**4 + 10*x**2 - 96*x - 72): branch, delta = 2,
      v_P(f) = -1  [sub-critical]
  unit search inconclusive; [oo+ - oo-] certified non-torsion by
      reduction mod p: [(7, 3), (11, 13), (13, 7), (17, 21)]
  retrying with the exact degree bounds of Part I: [0, -2]
('not elementary', 'holomorphic remainder: residual second-kind
  differential is not exact (exact bounds of Part I)', [0, -2])
\end{verbatim}

\noindent The orders $3,13,7,21$ modulo $7,11,13,17$ are pairwise incompatible with any common $N$, so $[\infty_+-\infty_-]$ is not torsion; the integrand has no residues, and with no units the only candidate integral is a rational function of the shape bounded by Part I, which the system excludes. This mechanism is also the subject of \cite{BlakePE}, which computes such integrals by a direct undetermined-coefficient ansatz $\log(a+b\sqrt q)$; the present framework delimits exactly what that ansatz can reach: the integrals whose residue-invisible unit part carries everything, i.e.\ the torsion cases.

\begin{verbatim}
>>> q71 = x**4 + 10*x**2 - 96*x - 71
>>> parallel_integrate_mixed((0, x/q71),
...     Tower([x], [(1,0)], q=q71), verbose=True)
  (x**4 + 10*x**2 - 96*x - 71): branch, delta = 2,
      v_P(f) = -1  [sub-critical]
  unit candidate: A + B*y with deg_x B = 6
log(x**8 + 20*x**6 - 128*x**5 + 54*x**4 - 1408*x**3 + 3124*x**2
    + sqrt(x**4 + 10*x**2 - 96*x - 71)*(x**6 + 15*x**4 - 80*x**3
    + 27*x**2 - 528*x + 781) + 10001)/8
\end{verbatim}

\subsection{Benchmark: the genus-2 integral from Schultz's \textit{Trager's Algorithm for Integration of Algebraic Functions Revisited}}\label{ex:schultz}

This benchmark raises the genus. Schultz \cite[p.~2]{Schultz15}, revisiting Trager's algorithm with a partial Mathematica implementation, opens with the ``remarkable result''
\[
\int\frac{(29x^2+18x-3)\;dx}{\sqrt{x^6+4x^5+6x^4-12x^3+33x^2-16x}}\;=\;\log\bigl(A+B\sqrt q\,\bigr),
\qquad \deg A=29,\ \deg B=26 :
\]
a genuinely \emph{abelian} integral --- $q$ has degree $6$ and the curve genus $2$ --- carried once more entirely by a unit, $A^2-qB^2\in\Q^*$, i.e.\ by torsion of $\infty_+-\infty_-$ of order $29$ in a genus-$2$ Jacobian. The mechanism is identical to the Cohen integral, and the pipeline computes it \emph{unchanged}: $q=x\,(x^5+4x^4+6x^3-12x^2+33x-16)$ is squarefree, both branch primes are sub-critical, no residues exist, and the continued fraction of $\sqrt q$ --- which is genus-agnostic --- delivers the unit in under three seconds, matching Schultz's printed coefficients digit for digit --- and correcting one misprint (see the session below) --- with $\gamma=1$ from the linear system. The genus enters only in the negative direction: deciding that \emph{no} unit exists (the analogue of Cohen's $-72$ variant) requires torsion bounds in a genus-$2$ Jacobian --- milestone (iii) at its hardest.

\begin{verbatim}
>>> q6 = x**6+4*x**5+6*x**4-12*x**3+33*x**2-16*x
>>> parallel_integrate_mixed((0, (29*x**2+18*x-3)/q6),
...     Tower([x], [(1,0)], q=q6), verbose=True)
  (x): branch, delta = 2, v_P(f) = -1  [sub-critical]
  (x**5 + 4*x**4 + 6*x**3 - 12*x**2 + 33*x - 16): branch,
      delta = 2, v_P(f) = -1  [sub-critical]
  unit candidate: A + B*y with deg_x B = 26
log(x**29 + 40*x**28 + 776*x**27 + 9648*x**26 + 85820*x**25
    + 578480*x**24 + 3058536*x**23 + 12979632*x**22 + 45004902*x**21
    + 129708992*x**20 + 317208072*x**19 + 675607056*x**18
    + 1288213884*x**17 + 2238714832*x**16 + 3548250712*x**15
    + 5097069328*x**14 + 6677210721*x**13 + 8106250392*x**12
    + 9056612528*x**11 + 8991685504*x**10 + 7944578304*x**9
    + 6614046720*x**8 + 4834279424*x**7 + 2374631424*x**6
    + 916848640*x**5 + 638582784*x**4 - 279969792*x**3
    - 528482304*x**2 + 150994944*x + sqrt(x**6 + 4*x**5 + 6*x**4
    - 12*x**3 + 33*x**2 - 16*x)*(x**26 + 38*x**25 + 699*x**24
    + 8220*x**23 + 68953*x**22 + 436794*x**21 + 2161755*x**20
    + 8550024*x**19 + 27506475*x**18 + 73265978*x**17
    + 165196041*x**16 + 324386076*x**15 + 570906027*x**14
    + 914354726*x**13 + 1326830817*x**12 + 1731692416*x**11
    + 2055647184*x**10 + 2257532160*x**9 + 2246693120*x**8
    + 1939619840*x**7 + 1494073344*x**6 + 1097859072*x**5
    + 640024576*x**4 + 207618048*x**3 + 95420416*x**2
    + 50331648*x - 50331648) - 134217728)
\end{verbatim}

\noindent The coefficients agree with the display of \cite[p.~2]{Schultz15} digit for digit, with one exception: the $x^{21}$ coefficient of $A$ is printed there as $4500490$; the correct value, certified both by differentiation and by the exact Pell identity $A^2-qB^2\in\Q^*$, is $45004902$.

\subsection{Benchmark: the example from Bronstein's \textit{The Risch Differential Equation on an Algebraic Curve}}\label{ex:rde}

This benchmark is the showcase of \cite{Bronstein91}, whose subject is the Risch differential equation on an algebraic curve --- the equation the recursive algorithm must solve whenever an exponential of an algebraic function occurs:
\[
\int\Bigl(\frac{5x^4+2x-2}{x^2}\Bigl(1+\frac1{\sqrt{x^3+1}}\Bigr)+\frac{x}{\sqrt{x^3+1}}\Bigr)\,e^{x\sqrt{x^3+1}}\,dx
\;=\;\frac2x\bigl(1+\sqrt{x^3+1}\bigr)\,e^{x\sqrt{x^3+1}} .
\]
The tower is $y^2=x^3+1$ --- elliptic, of odd degree, hence one place at infinity and unit rank $0$ --- with the hyperexponential $t=e^{xy}$, $Dt/t=(5x^3+2)/(2y)$, on top: unflattenable, and outside every previous benchmark's mechanism. In the recursive algorithm the integral requires solving $Dg+g\,D(xy)=h$ over the curve, the problem \cite{Bronstein91} was written to solve. In the parallel method there is no equation to isolate: the pipeline classifies $(x)$ as normal with $\delta=1$ carrying a pole of order $2$ --- the first benchmark with a nontrivial Hermite part, $D_v=x$ --- finds the branch pole sub-critical, offers the special $\log(t)$ and no units (rank $0$ is read off the odd degree), and the Risch differential equation appears as the $t^1$-blocks of the linear system, which returns $v=\tfrac2x(1+y)\,t$ in under two seconds, verified by differentiation.

\begin{verbatim}
>>> q3 = x**3 + 1                        # elliptic; t = exp(x*y)
>>> T = Tower([x, t], [(1,0), (0, t*(5*x**3+2)/(2*q3))], q=q3)
>>> parallel_integrate_mixed(((5*x**4+2*x-2)*t/x**2,
...     (5*x**4+x**3+2*x-2)*t/(x**2*q3)), T, verbose=True)
  (x): unramified, delta = 1, v_P(f) = -2
  (x**2 - x + 1): branch, delta = 2, v_P(f) = -1  [sub-critical]
  tower special: candidate log(t)
2*t*sqrt(x**3 + 1)/x + 2*t/x
\end{verbatim}

\noindent The Risch differential equation of \cite{Bronstein91} is the $t^1$-block of this system; nothing in the code knows it is solving one.

\subsection{Benchmark: Chebyshev's integral from Davenport's \textit{On the Integration of Algebraic Functions}}\label{ex:chebyshev}

Example 5 on p.~152 of Davenport's monograph \cite{Davenport81}, attributed there to Chebyshev:
\begin{align*}
&\int\frac{2x^6+4x^5+7x^4-3x^3-x^2-8x-8}{(2x^2-1)^2\,\sqrt{x^4+4x^3+2x^2+1}}\;dx\\
&\qquad=\;\frac{(2x+1)\sqrt q}{2\,(2x^2-1)}\;-\;\frac52\log\bigl(x^2+2x+\sqrt q\bigr)\;+\;\frac52\log\bigl(x^2+2x-\sqrt q\bigr)\\
&\qquad\qquad+\;\log\bigl(A+B\sqrt q\bigr),
\end{align*}
with $q=x^4+4x^3+2x^2+1$, $A=x^5+7x^4+15x^3+9x^2+2$ and $B=x^3+5x^2+6x$. This is the only benchmark exercising all three mechanisms of the theory in a single integral. The curve is elliptic; $q=(x+1)(x^3+3x^2-x+1)$ is squarefree and both branch primes are sub-critical. At the unramified prime $(2x^2-1)$ the integrand has a pole of order $2$: a Hermite part, $D_v=2x^2-1$, producing the algebraic term --- \emph{and} residues, for by Proposition~\ref{rem:deepres} the residues remain determined at any pole order. The four places carry $\pm\tfrac52$, distinct on conjugates, and the realising logands come from the norm search: $N(x^2+2x\pm y)=(x^2+2x)^2-q=2x^2-1$ exactly. Finally, these split logands leave residues $\pm5$ at the two places at infinity, which no exact part can cancel: a logand supported wholly at infinity is forced, and the fundamental unit --- $A^2-qB^2=4$, torsion of order $5$ --- enters with coefficient $1$ from the linear system. Hermite part, norm-realised deep residues, and a residue-invisible unit, in one classical integral; the pipeline computes it in $24$ seconds.

\begin{verbatim}
>>> q4 = x**4 + 4*x**3 + 2*x**2 + 1
>>> N4 = 2*x**6 + 4*x**5 + 7*x**4 - 3*x**3 - x**2 - 8*x - 8
>>> parallel_integrate_mixed((0, N4/((2*x**2-1)**2*q4)),
...     Tower([x], [(S(1), S(0))], q=q4), verbose=True)
  (x + 1): branch, delta = 2, v_P(f) = -1  [sub-critical]
  (2*x**2 - 1): unramified, delta = 1, v_P(f) = -2
      deep residues at order 2: [5/2, -5/2, -5/2, 5/2]
      norm factor -x**2 - 2*x + (1)*y (N = 1*(2*x**2 - 1)**1):
          residue 5/2
      norm factor -x**2 - 2*x + (-1)*y (N = 1*(2*x**2 - 1)**1):
          residue -5/2
  (x**3 + 3*x**2 - x + 1): branch, delta = 2, v_P(f) = -1  [sub-critical]
  unit candidate: A + B*y with deg_x B = 3
(2*x + 1)*sqrt(x**4 + 4*x**3 + 2*x**2 + 1)/(4*x**2 - 2)
    - 5*log(-x**2 - 2*x - sqrt(x**4 + 4*x**3 + 2*x**2 + 1))/2
    + 5*log(-x**2 - 2*x + sqrt(x**4 + 4*x**3 + 2*x**2 + 1))/2
    + log(x**5 + 7*x**4 + 15*x**3 + 9*x**2
          + (x**3 + 5*x**2 + 6*x)*sqrt(x**4 + 4*x**3 + 2*x**2 + 1) + 2)
\end{verbatim}

\noindent(the sign flips inside the logarithms are additive constants).

\subsection{Benchmark: G\"unther's integral from \textit{Sur l'\'evaluation de certaines int\'egrales pseudo-elliptiques}}\label{ex:gunther}

G\"unther's 1882 paper \cite{Gunther1882} evaluates pseudo-elliptic integrals of which
\[
\int\frac{x\;dx}{(x^3+8)\,\sqrt{x^3-1}}
\]
is a representative, and it completes the benchmark ladder: it is the first integral here whose logands can be produced \emph{only} by torsion. The curve $y^2=x^3-1$ is elliptic with one place at infinity (unit rank $0$) and complex multiplication by $\zeta_3$, acting by $(x,y)\mapsto(\zeta_3x,y)$; the branch primes are sub-critical, and the six critical places lie over $x^3+8$ --- the $\zeta_3$-orbit of $x_0=-2$, paired with $y_0=\pm3i$ --- with residues
\[
\tau_{(x_0,y_0)}(f)\;=\;\frac{1}{3\,x_0\,y_0}\,,
\]
pairwise distinct across every fibre. The equal-residue and norm-identity realisations therefore do not apply: the divisor forces functions with single-place support, i.e.\ the torsion of the classes $[P-\infty]$ themselves. That torsion exists, of order $6$: for $P=(-2,3i)$ the group law gives $2P=(0,i)$ and $3P=(1,0)$, a two-torsion point, and likewise at the conjugate places. The Miller-style function with divisor $6P-6\infty$ is
\[
h_P\;=\;\frac{\ell_P^{\,2}\;\ell_P'^{\,2}}{x^{2}\,(x-x_{3P})}\,,
\]
with $\ell_P$ the tangent at $P$, $\ell_P'$ the chord through $P$ and $2P$, and $x_{3P}$ the abscissa of the two-torsion point $3P$; and by Corollary~\ref{cor:divisors}, $\int f\,dx=\sum_P(\tau_P/6)\log(h_P)+C$ over the six places --- an identity that proves itself: the residues match by design, and the vertical parts and the single place at infinity enter with total weight $\sum_P\tau_P=0$, so the difference from the integrand is pole-free, hence constant.

The formal sum collapses. The product of tangent and chord factors as $\ell_P\,\ell_P'=x\,\bigl[(x-\mu_P)^2+c_P\,y\bigr]$ with $\mu_P=x_{3P}$, so $h_P=g_P^{\,2}/(x-\mu_P)$ for the \emph{linear-in-$y$} logand $g_P=(x-\mu_P)^2+c_Py$; the vertical logarithms cancel fibrewise ($\tau_P$ is odd in $y_0$), leaving $\int f\,dx=\tfrac13\sum_P\tau_P\log(g_P)+C$; and carrying out the sum over conjugate pairs yields the familiar form
\begin{align*}
\int\frac{x\;dx}{(x^3+8)\,\sqrt{x^3-1}}
\;=\;&\ \frac{\sqrt3}{108}\,\log\!\left(\frac{(x^2+x+1)(x^2+10x-8)+3\sqrt3\,(x^2+2x)\,y}{(x^2+x+1)(x^2+10x-8)-3\sqrt3\,(x^2+2x)\,y}\right)\\[2pt]
&-\frac1{27}\,\arctan\!\left(\frac{3\,y}{(x-1)^2}\right)
\;+\;\frac1{54}\,\arctan\!\left(\frac{3\,(2+2x-x^2)\,y}{x^4-7x^3+3x^2+2x+10}\right),
\end{align*}
with $y=\sqrt{x^3-1}$, verified by differentiation: the derivative minus the integrand reduces to $0$ exactly modulo $y^2=x^3-1$.

The pipeline computes it end to end. The torsion realisation is implemented for elliptic curves with one place at infinity ($\deg q=3$): a bounded order search by the chord-and-tangent law, and an additive Miller loop building the function with divisor $mP-m\infty$, with the coordinate arithmetic performed exactly in the algebraic extension field $\Q(i,\sqrt3)$:

\begin{verbatim}
>>> parallel_integrate_mixed((0, x/((x**3+8)*(x**3-1))),
...     Tower([x], [(S(1), S(0))], q=x**3-1), verbose=True)
  (x + 2): unramified, delta = 1, v_P(f) = -1
      places ['y=3*I', 'y=-3*I']: residues [I/18, -I/18]
  (x**2 - 2*x + 4): unramified, delta = 1, v_P(f) = -1
      places [...]: residues [-sqrt(3)/36 - I/36, sqrt(3)/36 + I/36,
                              sqrt(3)/36 - I/36, -sqrt(3)/36 + I/36]
  (x - 1), (x**2 + x + 1): branch, delta = 2  [sub-critical]
      torsion: [P - oo] of order 6 at (-2, 3*I);
          Miller logand with coefficient I/108
      torsion: [P - oo] of order 6 at (-2, -3*I);
          Miller logand with coefficient -I/108
      torsion: [P - oo] of order 6 at (1 + sqrt(3)*I, 3*I);
          Miller logand with coefficient -sqrt(3)/216 - I/216
      (and the three conjugate order-6 classes)
<the six Miller logands with coefficients tau_P/6>
\end{verbatim}

\noindent The returned expression is the formal sum $\sum_P(\tau_P/6)\log(h_P)$ of the derivation above, verified by differentiation; the displayed real form is its collapse over conjugate pairs.

Torsion of the classes $[P-\infty]$ themselves is the exception rather than the rule: what Corollary~\ref{cor:divisors} requires is the torsion of the degree-zero divisors $\Delta_k$, and a place can carry a residue without being a torsion point. When the places fail individually, Algorithm~3(d) works on the divisor: the leftover residues are written over a $\Q$-basis, each component $\sum_Pn_PP$ is summed by the group law to a point $S$, the order $\mu$ of $S$ is the order of the class, and the function with divisor $\mu\sum_Pn_P(P-\infty)$ is the product of the lines and verticals of the additions that take $\mu\sum_Pn_PP$ to $\infty$ --- the same accumulation that produces the Miller function when the sum is $\mu P$. The integral $\int\sqrt{1+\sec x\tan x}\,dx$ is the model case: with $t=\tan(x/2)$ the radicand is the quartic $t^4+2t^3-2t^2+2t+1$, whose cubic model $y^2=-8s^3+28s^2-30s+25$ carries the residues $-1,1$ at $(3/2,4)$ and $(-5/2,-20)$ --- over different primes, and neither class $[P-\infty]$ is torsion --- and $\pm i$ at $(3/2\pm2i,-8\mp4i)$. Both differences are $2$-torsion (the first sums to the point $(-20,0)$ of the monic model), and the two logands, with coefficients $-1/2$ and $i/2$, complete the integral, which is verified by differentiation on $0<x<\pi/2$:

\begin{verbatim}
  quartic radicand: changing to the cubic model
      y^2 = -8*s**3 + 28*s**2 - 30*s + 25
  (2*s - 3): places ['y=4', 'y=-4']: residues [-1, 0]
  (2*s + 5): places ['y=20', 'y=-20']: residues [0, 1]
  (4*s**2 - 12*s + 25): places [...]: residues [0, I, 0, -I]
      torsion: the divisor 1*(3/2, 4) + -1*(-5/2, -20) - (0) oo
          has order 2; logand with coefficient -1/2
      torsion: the divisor 1*(3/2 + 2*I, -8 - 4*I) + -1*(3/2 - 2*I, -8 + 4*I)
          - (0) oo has order 2; logand with coefficient I/2
\end{verbatim}

\noindent With this, no example or benchmark in this paper is left unrealised by the implementation.

\subsection{A nested radical over a conic: $\int\sqrt{x+\sqrt{1-x^2}}\,dx$}\label{ex:nested}

The radicand $x+\sqrt{1-x^2}$ is not a pure root of a generator, so Lemma~\ref{lem:flatten} does not apply, and two radicals would have to coexist. The inner one is the conic $y^2=1-x^2$ with the rational point $(0,1)$; the pencil through it, $x=2w/(1+w^2)$, $y=(1-w^2)/(1+w^2)$, identifies $\Q(x,y)$ with the rational function field $\Q(w)$, and the outer radical becomes a single radical over it:
\[
\sqrt{x+\sqrt{1-x^2}}=\frac{\sqrt{(1+2w-w^2)(1+w^2)}}{1+w^2}=\frac{y}{1+w^2},\qquad y^2=q(w)=-w^4+2w^3+2w+1,
\]
a curve of genus one, with $Dw=-(1+w^2)^2/\bigl(2(w^2-1)\bigr)$ the derivative of the parameter with respect to $x$. The tower is the single non-monomial generator $w$ over the curve, and Lemma~\ref{lem:rescale} turns it into the setting of Part I: with $D'=(Dw)^{-1}D$,
\[
\int\sqrt{x+\sqrt{1-x^2}}\,dx=\int\frac{2(1-w^2)}{(1+w^2)^3}\,y\,dw .
\]
The only finite pole is at the branch prime $w^2+1$, of order $5$ in the uniformiser, and its residue vanishes on both sheets; the continued fraction of $y$ finds no unit within the height bound, and $[\infty_+-\infty_-]$ is certified non-torsion by Proposition~\ref{prop:nontorsion} (orders $5,3,5,17$ modulo $5,13,17,29$); the linear system has no solution with the bound of Step~15 and none with the exact bounds of Part I, and Proposition~\ref{prop:certificates}(b) applies:

\begin{verbatim}
>>> integrate_surface(sqrt(x + sqrt(1 - x**2)), x, verbose=True)
  tower: generators [w1] with D = [((-w1**4 - 2*w1**2 - 1)/(2*w1**2 - 2), 0)],
      y^2 = -w1**4 + 2*w1**3 + 2*w1 + 1
  integrand: (0, 1/(w1**2 + 1))
  single generator w1 with Dw1 = (-w1**4 - 2*w1**2 - 1)/(2*w1**2 - 2):
      rescaled to d/dw1 (Lemma 3.4)
  (w1**2 + 1): branch, delta = 2, v_P(f) = -5
  unit search inconclusive; [oo+ - oo-] certified non-torsion by
      reduction mod p: [(5, 5), (13, 3), (17, 5), (29, 17)]
  ansatz: bounds [12], 26 unknowns (26 coefficients, 0 unit/S'-unit logs,
      0 special logs), 40 equations
  linsolve: 0.0s, no solution
  retrying with the exact degree bounds of Part I: [4, 2]
('not elementary', 'holomorphic remainder: residual second-kind
  differential is not exact (exact bounds of Part I)', [4, 2])
\end{verbatim}

\noindent The certificate is checked by hand in one line: with $v=4wy/\bigl(3(1+w^2)^2\bigr)=\tfrac23\,x\sqrt{x+\sqrt{1-x^2}}$,
\[
\frac{2(1-w^2)}{(1+w^2)^3}\,y\,dw-dv=\frac23\,\frac{dw}{y},
\]
the differential of the first kind on the curve, so that
\[
\int\sqrt{x+\sqrt{1-x^2}}\,dx=\frac23\,x\sqrt{x+\sqrt{1-x^2}}+\frac13\int\frac{dx}{\sqrt{1-x^2}\,\sqrt{x+\sqrt{1-x^2}}},
\]
with an elliptic integral of the first kind as the remainder. For the derivation $D$ itself the prime $w^2+1$, which divides the numerator of $Dw$, is special, and the certificate does not apply to a tower with a special prime; it is the rescaling that makes the decision possible. Mathematica returns the integral unevaluated; the Mathematica port certifies it in $0.2$\,s.

\subsection{A cancellation chain below the top}\label{ex:chain}

The tower $t_1=\log(x)$, $t_2=\log(xt_1+1)$ is the second example of Part I, \S6.3, where a degree drop of two in $t_1$ defeats the bound of \cite[(10.4)]{Bronstein05}. Example~\ref{ex:chainpreview} explains the drop as a cancellation chain whose links each cost one degree in $t_2$; the following integrand has a chain of four links. Let
\begin{multline*}
g\ =\ t_1^5-5t_2t_1^4+\tfrac{35}{12}t_1^4+\bigl(10t_2^2-\tfrac{20}{3}t_2+\tfrac{20}{9}+\tfrac5x\bigr)t_1^3\\
+\bigl(-10t_2^3-\tfrac{20}{x}t_2+\tfrac{20}{3x}-\tfrac{5}{2x^2}\bigr)t_1^2 ,
\end{multline*}
whose coefficients were solved for so that the leading five orders of $Dg$ in $t_1$ cancel, and let $f:=Dg$:
\begin{multline*}
f\ =\ \frac{-1}{3x^3(xt_1+1)}\Bigl(60x^3t_1^2t_2^3+90x^3t_1^2t_2^2-90x^2t_1^2t_2^2+180x^2t_1^2t_2-60xt_1^2t_2-10xt_1^2-15t_1^2\\
+60x^2t_1t_2^3+120xt_1t_2-40xt_1+15t_1\Bigr) .
\end{multline*}
The numerator of $f$ has degree $2$ in $t_1$ and its denominator degree $1$, while $g$ has degree $5$: the classical bound, and the ``$+2$'' the implementations used before Section~\ref{sec:lower}, fail. The session of the SymPy implementation with that guess runs the ladder $\{10,4,5\}$, $\{12,5,5\}$, $\{14,6,5\}$ (the exponents of the specials $x$ and $xt_1+1$ raised at each rung) and reports ``failed''; the Mathematica port does the same. With Algorithm~6:
\begin{verbatim}
  (t1*x + 1): special (multiplicity 1); candidate log(t1*x + 1)
  (x): special (multiplicity 3); candidate log(x)
  bounds: x: 8 [guess], t1: 8 [guess], t2: 4 [K2];
      specials: (t1*x + 1)^0 [K1], (x)^2 [K1]  (a guess is in play)
  ansatz: bounds [8, 8, 4], 407 unknowns (405 coefficients,
      0 unit/S'-unit logs, 2 special logs), 493 equations
  linsolve: 0.0s, solved
\end{verbatim}
The two special exponents are proved by (K1): at $(x)$, $s=-1$, $\rho_v=1$, and the only attaining generator is $t_1$ with $\lambda_1=1$, so the exponent is $\max(0,-1-v_x(f))=2$, the multiplicity $3$ of $x$ in the denominator of $f$ less one; at $(xt_1+1)$ the attaining generator is $t_2$ and the exponent is $0$. The place at infinity of $t_2$ is decided by (K2). The places at infinity of $x$ and of $t_1$ are undecided --- at $\infty_x$ the leading coefficient $(\bar t_1+1)/\bar t_1$ of $Dt_2$ involves the attaining $\bar t_1$; at $\infty_{t_1}$ the generator $t_2$ depends on $t_1$ --- and there the guess is used, with the chain allowance of one plus the bound in the variable above at $\infty_{t_1}$, whence $8$. The system solves at the first rung, and the integral is verified by differentiation; the bounds actually needed are $\{4,6,3\}$, and the integral is found with those as well.

\subsection{The arcsine with a parameter: $\int\arcsin(\sqrt{x+b})/\sqrt x\,dx$}\label{ex:arcsinb}

For $b=1$ this is the integrand of Example~\ref{ex:eulerbounds}: the derivative of $\arcsin(\sqrt{x+1})$ introduces $\sqrt{1-(x+1)}=i\sqrt x$, no new radical, and the Euler parameter $w=\sqrt x+\sqrt{x+1}$ leaves a transcendental tower. For a parameter $b$ the derivative introduces $\sqrt{1-x-b}$, and after the same parametrisation, $x=(w^2-b)^2/(4w^2)$, the tower is $w$, $y$, $t=\arcsin(\sqrt{x+b})$ with
\[
Dw=\frac{2w^3}{w^4-b^2},\qquad y^2=-4\bigl(w^4+(2b-4)w^2+b^2\bigr),\qquad Dt=-\frac{w^2\,y}{(w^2+b)\bigl(w^4+(2b-4)w^2+b^2\bigr)},
\]
and the integrand is $f=2tw/(w^2-b)$. The quartic has discriminant $2^{24}b^2(b-1)^2$: for $b\notin\{0,1\}$ the curve has genus one, and for $b=1$ it degenerates to $y=2i(w^2-1)$. The session:
\begin{verbatim}
  tower special: candidate log(w)
  bounds: w: 2 [K1], t: 2 [K2];  specials: (w)^1 [K1]  (all proved)
  ansatz: bounds [2, 2], 19 unknowns (18 coefficients,
      0 unit/S'-unit logs, 1 special logs), 64 equations
  linsolve: 0.0s, no solution
  ('failed', 'no solution within bounds', [2, 2])
\end{verbatim}
Every bound in force is proved: at the two places over $w=\infty$ and at the two over the special $w$ (which is unramified, $q(0)=-4b^2$), the shift is $2$, the uniformiser attains it, and the only other attaining generator is the primitive $t$ with a constant leading coefficient, so (K1) applies; at $t=\infty$, (K2). The system is inconsistent, and the residual is residue-free by construction (the pole of $f$ at $w^2=b$ is sub-critical). What Proposition~\ref{prop:certificates}(b) still requires is the completeness of the logand candidates: $w$ is a special prime over the curve variable, and an $S'$-unit with divisor supported on the four places over $w=0$ and the two at infinity is only searched within a bound (the norm search finds $-2ib-2iw^2\pm y$, of norm $-16w^2$), so the algorithm claims nothing and returns ``failed''. The verdict a complete unit search would certify is true: integrating by parts,
\[
\int\frac{\arcsin(\sqrt{x+b})}{\sqrt x}\,dx\;=\;2\sqrt x\,\arcsin(\sqrt{x+b})-\int\frac{x\,dx}{\sqrt{x(x+b)(1-x-b)}}\,,
\]
and $x\,dx/y$ on $y^2=x(x+b)(1-x-b)$ is a differential of the second kind, not exact for $b\notin\{0,1\}$, so the integral is not elementary; Mathematica returns it with \texttt{EllipticE}. At $b=1$ the remainder is $-2i\sqrt{x+1}$, and the integral is the one of Example~\ref{ex:eulerbounds}. The example marks the boundary precisely: the degree bounds are no longer the obstacle to a certificate above a curve of positive genus, the unit group is.

\section{Conclusions and open problems}\label{sec:conclusion}

Removing the closure hypothesis of Part I --- letting the radical sit anywhere in the tower, with the generators above it differentiating through $y$ --- leaves the parallel method intact: one ring, one ansatz, one linear system, no integral bases, no Puiseux expansions, no recursion. The denominator of the derivation became a divisor and the valuation lemma became purely local, $\delta_P=1+v_P(\fd_D)$, subsuming and (in the non-monomial case) correcting Part I. Three phenomena with no transcendental analogue organise the rest: tower-dependent normality (Proposition~\ref{prop:typeE}), the sub-critical window in which poles demand neither Hermite denominators nor residues, and the blindness of residues to units, splitting the logands into a determined part (residues), a searched part (units and norms) and a decided part (torsion) --- a trichotomy Chebyshev's integral (\S\ref{ex:chebyshev}) realises in one answer. Every integral exhibited in Section~\ref{sec:examples} is computed by the accompanying implementation and verified by differentiation, a discipline that incidentally surfaced misprints in two sources.

Two problems remain open. \emph{Nested radicals}: Lemma~\ref{lem:flatten} reduces a radical tower to a single radical only when all but one layer flattens, and a conic layer is parametrised away, Lemma~\ref{lem:rescale} restoring the setting of Part~I (\S\ref{ex:nested}); beyond these, the local theory of Sections~\ref{sec:valuations}--\ref{sec:structure} was proved for arbitrary Krull domains and so applies verbatim to the iterated closure, but the Trager basis, content theory and norm search need their Dedekind-module counterparts --- the valuations are ready; the representations are the work. \emph{Beyond simple radicals}: for a single $y$ with arbitrary minimal polynomial (the Bronstein-thesis setting, of which \S\ref{ex:bronsteinE} is the flattenable shadow) the explicit basis is unavailable, and the honest trade is to compute an integral basis once per tower while keeping the parallel control flow; tameness in characteristic $0$ makes us expect $\delta_P=e_P(1+v_p(\den_0))$ to persist. Beyond these: completeness --- the degree bounds and the special exponents are proved wherever the leading part of the derivation has no kernel (Section~\ref{sec:lower}), and what remains is the list of Remark~\ref{rem:decided}, each item a Risch-type decision in a residue field; certified non-elementarity where a logand divisor is supported at finite places of a curve of genus $\ge2$ or has non-constant coordinates, beyond the reach of the division polynomials and of Proposition~\ref{prop:nontorsion}; and the reduction of that list to the reduction systems of \cite{DuRaab25}, which decide exactly the membership questions it poses.

Parts I--II form an arc: the radical on top of the tower, then inside it. The local turn of Section~\ref{sec:valuations} was taken with the open problems above in view.

\section*{Data availability}
The SymPy implementation used for every session in Section~\ref{sec:examples} (\texttt{parallel\_mixed.py}, with \texttt{pell.py} for the continued-fraction unit search and \texttt{examples.py}, which reruns every example of the paper and reports the outcome claimed for it) is provided as supplementary material, together with the Mathematica port \texttt{ParallelMixed.wl} of the same algorithm and its example scripts. The full Python script runs in about a minute and a half, of which Chebyshev's integral (\S\ref{ex:chebyshev}) is forty seconds (its norm search); the two slowest benchmarks (\S\S\ref{ex:chebyshev}, \ref{ex:gunther}) can be skipped with \texttt{QUICK=1}.

\end{document}